\documentclass[a4paper,10pt]{article}

\usepackage[utf8]{inputenc}
\usepackage{xparse}
\usepackage{amssymb}
\usepackage{amsmath}
\usepackage{amsthm}
\usepackage{mathtools}
\usepackage{bbm}
\usepackage{todonotes}
\usepackage{enumerate}
\usepackage[labelformat=simple]{subcaption}

\usepackage[left=30mm,right=30mm,top=25mm,bottom=25mm]{geometry} 
\usepackage{pi}
\usepackage{authblk}
\usepackage{xspace}
\usepackage{hyperref}
\hypersetup{
    colorlinks,
    citecolor=black,
    filecolor=black,
    linkcolor=black,
    urlcolor=black
}
\usepackage{mathabx} 

\DeclareDocumentCommand\R{}{\mathbb{R}}
\DeclareDocumentCommand\Z{}{\mathbb{Z}}
\DeclareDocumentCommand\Q{}{\mathbb{Q}}
\DeclareDocumentCommand\setdef{mo}{\left\{ #1 \IfValueTF{#2}{: #2}{} \right\}}
\DeclareDocumentCommand\zerovec{o}{\IfNoValueTF{#1}{\mathbb{O}}{\mathbb{O}_{#1}}}
\DeclareDocumentCommand\unitvec{m}{\mathbbm{e}_{#1}}
\DeclareDocumentCommand\dcup{}{\dot{\cup}}
\DeclareDocumentCommand\conv{o}{\operatorname{conv}\IfValueTF{#1}{\left(#1\right)}{}}
\DeclareDocumentCommand\cplxNP{}{\mathsf{NP}}
\DeclareDocumentCommand\scalprod{mm}{\left<#1,#2\right>}
\DeclareDocumentCommand\bincoeff{mm}{\genfrac{(}{)}{0pt}{}{#1}{#2}}
\DeclareDocumentCommand\transpose{m}{#1^{\intercal}}
\DeclareDocumentCommand\onevec{o}{\IfNoValueTF{#1}{\mathbbm{1}}{\mathbbm{1}_{#1}}}

\DeclareDocumentCommand\PmatchOne{}{P_{\text{match}}^{\text{1Q}}}
\DeclareDocumentCommand\PmatchOneUp{}{P_{\text{match}}^{\text{1Q}\uparrow}}
\DeclareDocumentCommand\PmatchOneDown{}{P_{\text{match}}^{\text{1Q}\downarrow}}
\DeclareDocumentCommand\Vspecial{}{V^*}

\DeclareDocumentCommand\facetsDown{}{\mathcal{S}^{\downarrow}}
\DeclareDocumentCommand\inequalitiesDown{}{\facetsDown_{\text{ext}}}
\DeclareDocumentCommand\facetsUp{}{\mathcal{S}^{\uparrow}}
\DeclareDocumentCommand\inequalitiesUp{}{\facetsUp_{\text{ext}}}

\newtheorem{theorem}{Theorem}[section]
\newtheorem{corollary}[theorem]{Corollary}
\newtheorem{proposition}[theorem]{Proposition}
\newtheorem{lemma}[theorem]{Lemma}
\newtheorem{claim}[theorem]{Claim}

\title{Complete Description of Matching Polytopes with One Linearized Quadratic Term for Bipartite Graphs}
\author[1]{Matthias Walter}
\affil[1]{RWTH Aachen University, walter@or.rwth-aachen.de}

\DeclareDocumentCommand\reviewComment{omm}{%
  \IfValueTF{#1}{%
    \todo[backgroundcolor=black!20!white,#1]{Referee: \textcolor{red}{\emph{#2}}; Answer: \textcolor{blue}{#3}}%
  }{%
    \todo[backgroundcolor=black!20!white]{Referee: \textcolor{red}{\emph{#2}}; Answer: \textcolor{blue}{#3}}%
  }
}
\DeclareDocumentCommand\reviewFix{m}{\textcolor{blue}{#1}\xspace}

\begin{document}

\maketitle

\begin{abstract}
  We consider, for complete bipartite graphs, the convex hulls of characteristic vectors of all matchings,
  extended by a binary entry indicating whether the matching contains two specific edges. 
  These polytopes are associated to the quadratic matching problems with a single linearized quadratic term.
  We provide a complete irredundant inequality description, which settles
  a conjecture by Klein (Ph.D.\ thesis, TU Dortmund, 2015).
  In addition, we also derive facetness and separation results for the polytopes.
  The completeness proof is based on a geometric relationship to a matching polytope
  of a nonbipartite graph.
  Using standard techniques, we finally extend the result to capacitated $b$-matchings.
\end{abstract}

\section{Introduction}
\label{SectionIntroduction}

Let $K_{m,n} = (V, E)$ be the complete bipartite graph with the node partition $V = U \dcup W$, $|U| = m$ and $|W| = n$ for $m,n \geq 2$.
The \emph{maximum weight matching problem} is to
maximize the sum $c(M) \coloneqq \sum_{e \in M} c_e$ over all matchings $M$
(i.e., $M \subseteq E$ and no two edges of $M$ share a node)
in $K_{m,n}$ for given edge weights $c \in \Q^E$.
Note that we generally abbreviate $\sum_{j \in J} v_j$ as $v(J$) for vectors $v$ and subsets $J$ of their index sets.

Following the usual approach in polyhedral combinatorics, we identify the matchings $M$ with their
\emph{characteristic vectors} $\chi(M) \in \setdef{0,1}^E$, which satisfy $\chi(M)_e = 1$ if and only if $e \in M$.
The maximum weight matching problem is then equivalent to the problem of maximizing the linear objective $c$
over the \emph{matching polytope}, i.e., the convex hull of all characteristic vectors of matchings.
In order to use linear programming techniques,
one requires a description of that polytope in terms of linear inequalities.
Such a description is well-known~\cite{Birkhoff46} and consists of the constraints
\begin{alignat}{6}
  x_e           & \geq 0        &\qquad& \text{for all } e \in E         \label{ConstraintMatchingNonnegative} \\
  x(\delta(v))  & \leq 1        &\qquad& \text{for all } v \in U \dcup W, \label{ConstraintMatchingDegree}
\end{alignat}
where $\delta(v)$ denotes the set of edges incident to $v$.
For general (nonbipartite) graphs, Edmonds~\cite{Edmonds65b,Edmonds65a} proved
that adding the following \emph{Blossom Inequalities} is sufficient to describe the matching polytope:
\begin{alignat*}{6}
  x(E[S])       & \leq \frac{1}{2}(|S|-1)       &\qquad& \text{for all } S \subseteq V \text{, $|S|$ odd,}
\end{alignat*}
where $E[S] \coloneqq \setdef{ \setdef{u,v} \in E }[ u,v \in S ]$.
His result is based on a primal-dual optimization algorithm, which also
proved that the weighted matching problem can be solved in polynomial time.
Later, Schrijver~\cite{Schrijver83} gave a direct (and more geometric) proof of the polyhedral result.
Note that one also often considers the special case of \emph{perfect} matchings,
which are those matchings covering every node of the graph.
The associated \emph{perfect matching polytope} is the face of the matching polytope
obtained by requiring that all Inequalities~\eqref{ConstraintMatchingDegree} are satisfied with equality:
\begin{alignat}{6}
  x(\delta(v))  & = 1        &\qquad& \text{for all } v \in U \dcup W.      \label{ConstraintPerfectMatchingDegree}
\end{alignat}
For more background on matchings and the matching polytopes we refer to Parts~II and~III of Schrijver's book~\cite{Schrijver03}.
For a basic introduction on polytopes and linear programming we recommend to read~\cite{Schrijver86}.

In this paper, we consider the more general \emph{quadratic matching problem} for which we have, in addition to $c$,
a set $\mathcal{Q} \subseteq \bincoeff{E}{2}$ and weights $p \colon \mathcal{Q} \to \Q$
for the edge-pairs in $\mathcal{Q}$.
The objective is now to maximize $c(M) + \sum_{q \in \mathcal{Q}, q \subseteq M} p_q$, again over all matchings $M$.
Before we discuss the case $|\mathcal{Q}| = 1$ in detail, we focus on the more general case.
By requiring the matchings to be perfect, we obtain as a special case the \emph{quadratic assignment problem},
a problem that is not just $\cplxNP$-hard~\cite{SahniG76}, but also hard to solve in practice (see~\cite{LoiolaABHQ07} for a survey).

\medskip

A common strategy is then to linearize this quadratic objective function by introducing
additional variables $y_{e,f} = x_e \cdot x_f $ for all $\setdef{e,f} \in \mathcal{Q}$.
Usually, the straight-forward linearization of this product equation is very weak,
and one seeks to find (strong) inequalities that are valid for the associated polytope.
There were several polyhedral studies, in particular for the quadratic assignment problem,
e.g., by Padberg and Rijal~\cite{PadbergR96} and Jünger and Kaibel~\cite{JuengerK96,Kaibel97,Kaibel98}.

One way of finding such inequalities, recently suggested by Buchheim and Klein~\cite{BuchheimK13}, is the so-called \emph{one term linearization technique}.
The idea is to consider the special case of $|\mathcal{Q}| = 1$ in which the optimization problem is still polynomially solvable.
By the polynomial-time equivalence of separation and optimization~\cite{GroetschelLS81,KarpP80,PadbergR81},
one can thus hope to characterize all (irredundant) valid inequalities and develop separation algorithms.
These inequalities remain valid when more than one monomial is present, and hence
one can use the results of this special case in the more general setting.
Buchheim and Klein suggested this for the quadratic spanning-tree problem
and conjectured a complete description of the associated polytope.
This conjecture was later confirmed by Fischer and Fischer~\cite{FischerF13}
and Buchheim and Klein~\cite{BuchheimK14}. 
Fischer et al.~\cite{FischerFM18} recently generalized this result to matroids and multiple monomials,
which must be nested in a certain way.
In her dissertation~\cite{Klein14}, Klein considered several other combinatorial polytopes, in particular the quadratic assignment polytope.
Hupp et al.~\cite{HuppKL15} generalized these results, in particular proofs for certain inequality
classes to be facet-defining, to nonbipartite matchings.
They carried out a computational study on the practical strength of this approach, using these inequalities during branch-and-cut.

\medskip

The main goal of this paper is to prove that the description for bipartite graphs conjectured by Klein~\cite{Klein14} is indeed complete.
Moreover, we extend the theoretical work of Klein to non-perfect matchings.
Our setup is as follows:
Consider two disjoint edges $e_1 = \setdef{u_1,w_1}$ and $e_2 = \setdef{u_2,w_2}$ (with $u_i \in U$ and $w_i \in W$ for $i=1,2$) in $K_{m,n}$
and denote by $\Vspecial \coloneqq \setdef{u_1,u_2,w_1,w_2}$ the union of their node sets.
Our polytopes of interest are
the convex hulls of all vectors $(\chi(M),y)$ for which $M$ is a matching in $K_{m,n}$,
$y \in \setdef{0,1}$ and one of the relationships between $M$ and $y$ holds:
\begin{itemize}
\item
  $\PmatchOneDown \coloneqq \PmatchOneDown(K_{m,n},e_1,e_2)$: $y = 1$ implies $e_1,e_2 \in M$.
\item
  $\PmatchOneUp \coloneqq \PmatchOneUp(K_{m,n},e_1,e_2)$: $y = 0$ implies $e_1 \notin M$ or $e_2 \notin M$.
\item
  $\PmatchOne \coloneqq \PmatchOne(K_{m,n},e_1,e_2)$: $y = 1$ if and only if $e_1,e_2 \in M$.
\end{itemize}
Note that $\PmatchOneDown$ (resp.\ $\PmatchOneUp$) is the \emph{downward} (resp.\ \emph{upward})
monotonization of $\PmatchOne$ with respect to the $y$-variable,
and that
\begin{gather*}
  \PmatchOne = \conv(\PmatchOneDown \cap \PmatchOneUp \cap (\Z^E \times \Z)).
\end{gather*}
Clearly, Constraints~\eqref{ConstraintMatchingNonnegative} and~\eqref{ConstraintMatchingDegree} as well as
the bound constraints
\begin{alignat}{6}
  0 \leq y      & \leq 1                                                 \label{ConstraintMatchingBound}
\end{alignat}
are valid for all three polytopes.
Additionally, the two inequalities 
\begin{alignat}{6}
  y             & \leq x_{e_i}  &\qquad& i=1,2                          \label{ConstraintMatchingQuadraticGood}
\end{alignat}
are also valid for $\PmatchOne$ and $\PmatchOneDown$ (and belong to the standard linearization of $y = x_{e_1} \cdot x_{e_2}$).
Klein~\cite{Klein14} introduced two more inequality classes, and proved them to be facet-defining (see Theorems~6.2.2 and~6.2.3 in~\cite{Klein14}).
They are indexed by subsets $\facetsDown$ and $\facetsUp$ of nodes (see Figure~\ref{FigureSupportGraphs}), defined via
\begin{align*}
  \facetsDown  &\coloneqq \{ S \subseteq U \dcup W : \text{$|S|$ odd and either } \\
              &~~\quad S \cap \Vspecial = \setdef{u_1,u_2} \text{ and } |S \cap U| = |S \cap W| + 1 \text{ or } \\
              &~~\quad S \cap \Vspecial = \setdef{w_1,w_2} \text{ and } |S \cap W| = |S \cap U| + 1 \} \text{, and } \\
  \facetsUp &\coloneqq \{ S \subseteq U \dcup W : \text{$|S \cap U| = |S \cap W|$ and either $S \cap \Vspecial = \setdef{u_1,w_2}$ or $S \cap \Vspecial = \setdef{u_2,w_1}$} \},
\end{align*}
and read
\begin{alignat}{6}
  x(E[S]) + y                         &\leq \frac{1}{2}(|S|-1) &\qquad& \text{for all } S \in \facetsDown \text{ and } \label{ConstraintMatchingQuadraticDown} \\
  x(E[S]) + x_{e_1} + x_{e_2} - y     &\leq \frac{1}{2}|S|   &\qquad& \text{for all } S \in \facetsUp. \label{ConstraintMatchingQuadraticUp}
\end{alignat}

\begin{figure}[htpb]
  \begin{center}
    \begin{subfigure}[b]{0.45\linewidth}
      \begin{center}
        \begin{tikzpicture}
          \piInput{tikz-quadratic-bipartite-matching-graphs}[%
            down support graph
          ]
        \end{tikzpicture}
        \vspace{-2em}
      \end{center}
      \caption{A set $S \in \facetsDown$ indexing Inequality~\eqref{ConstraintMatchingQuadraticDown}.}
      \label{FigureSupportGraphsDown}
    \end{subfigure}
    \hspace{1mm}
    \begin{subfigure}[b]{0.45\linewidth}
      \begin{center}
        \begin{tikzpicture}
          \piInput{tikz-quadratic-bipartite-matching-graphs}[%
            up support graph
          ]
        \end{tikzpicture}
        \vspace{-2em}
      \end{center}
      \caption{A set $S \in \facetsUp$ indexing Inequality~\eqref{ConstraintMatchingQuadraticUp}.}
      \label{FigureSupportGraphsUp}
    \end{subfigure}
  \end{center}
  \vspace{-1em}
  \caption{Node sets indexing additional facets.}
  \label{FigureSupportGraphs}
\end{figure}
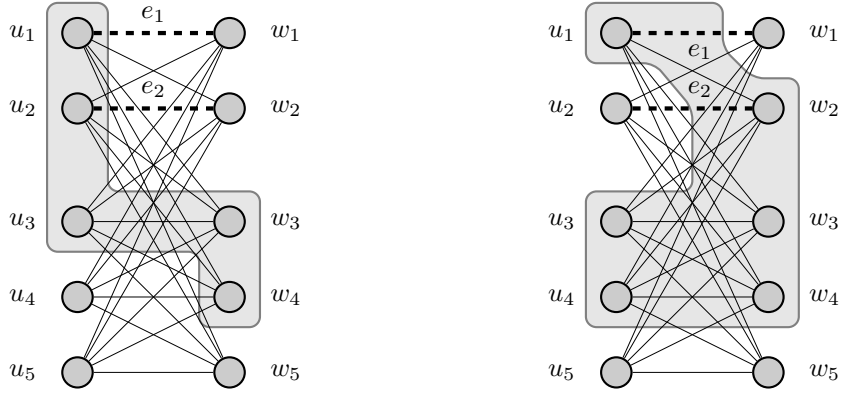

Klein~\cite{Klein14} even conjectured, that Constraints~\eqref{ConstraintMatchingNonnegative} and \eqref{ConstraintPerfectMatchingDegree}--\eqref{ConstraintMatchingQuadraticUp}
completely describe the mentioned face of $\PmatchOne$.
We will confirm this conjecture in Corollary~\ref{TheoremPerfectComplete}.

In contrast to the two proofs for the one-quadratic-term spanning-tree polytopes~\cite{FischerF13,BuchheimK14},
our proof technique is not based on linear programming duality.
In fact, the two additional inequality families presented above
introduce two sets of dual multipliers, which seem to make
this proof strategy hard, or at least quite technical.
Instead, we were heavily inspired by Schrijver's direct proof~\cite{Schrijver83} for the matching polytope.

\medskip

\paragraph{Outline.}
The paper is structured as follows:
In Section~\ref{SectionResults} we present our main results together with their proofs, which are based on two key lemmas,
one for $\PmatchOneDown$ and one for $\PmatchOneUp$.
At the end of the section, we establish the corresponding results for the special case of perfect matchings.
The proofs for the two key lemmas are similar with respect to the general strategy,
but are still quite different due to the specific constructions they depend on.
Hence, we present the general technique and then each lemma in its own dedicated section.
Although Klein already proved that the new inequalities are facet-defining, she only did so
for the case of perfect matchings. 
Hence, for the sake of completeness, we do the same for the general case in Section~\ref{SectionFacets}.
The algorithmic parts are covered in Section~\ref{SectionSeparation} where we present separation algorithms for the two
classes of exponentially many facets.
The polyhedral result on the matching polytope is used as a black-box result in Section~\ref{SectionBMatchings}
in order to prove a generalization for capacitated $b$-matchings (which are defined in that section).
We conclude this paper with a short discussion on our proof strategy and on a property of $\PmatchOne$.

\section{Main results}
\label{SectionResults}

We will prove our result using two key lemmas, each of which
is proved within its own section.

\begin{lemma}
  \label{TheoremDownCombination}
  Let $(\hat{x}, \hat{y}) \in \Q^E \times \Q$ satisfy
  Constraints~\eqref{ConstraintMatchingNonnegative},
  \eqref{ConstraintMatchingDegree},
  \eqref{ConstraintMatchingBound},
  \eqref{ConstraintMatchingQuadraticGood}
  and~\eqref{ConstraintMatchingQuadraticDown}.
  Let furthermore $(\hat{x}, \hat{y})$ satisfy
  at least one of the Inequalities~\eqref{ConstraintMatchingQuadraticGood} for $i^* \in \setdef{1,2}$
  or~\eqref{ConstraintMatchingQuadraticDown} for a set $S^* \in \facetsDown$ with equality.
  Then $(\hat{x}, \hat{y})$ is a convex combination of vertices of $\PmatchOne$.
\end{lemma}

Lemma~\ref{TheoremDownCombination} will be proved in Section~\ref{SectionDown}.

\begin{lemma}
  \label{TheoremUpCombination}
  Let $(\hat{x}, \hat{y}) \in \Q^E \times \Q$ satisfy
  Constraints~\eqref{ConstraintMatchingNonnegative}, \eqref{ConstraintMatchingDegree}, \eqref{ConstraintMatchingBound}, and Inequalities~\eqref{ConstraintMatchingQuadraticUp} for all $S \in \facetsUp$.
  Let furthermore $(\hat{x}, \hat{y})$ satisfy at least one of the Inequalities~\eqref{ConstraintMatchingQuadraticUp} for a set $S^* \in \facetsUp$ with equality.
  Then $(\hat{x}, \hat{y})$ is a convex combination of vertices of $\PmatchOne$.
\end{lemma}

Lemma~\ref{TheoremUpCombination} will be proved in Section~\ref{SectionUp}.
We continue with the consequences of the two lemmas.

\begin{theorem}
  \label{TheoremDownComplete}
  $\PmatchOneDown$ is equal to the set of $(x,y) \in \R^E \times \R$ that satisfy
  Constraints~\eqref{ConstraintMatchingNonnegative},
  \eqref{ConstraintMatchingDegree},
  \eqref{ConstraintMatchingBound},
  \eqref{ConstraintMatchingQuadraticGood} and
  \eqref{ConstraintMatchingQuadraticDown}.
\end{theorem}

\begin{proof}
  Let $P$ be the polytope defined by Constraints~\eqref{ConstraintMatchingNonnegative}, \eqref{ConstraintMatchingDegree}, \eqref{ConstraintMatchingBound}, \eqref{ConstraintMatchingQuadraticGood} and~\eqref{ConstraintMatchingQuadraticDown}.
  We first show $\PmatchOneDown \subseteq P$ by showing $(\chi(M),y) \in P$ for all feasible integer pairs $(\chi(M),y)$, i.e., matchings $M$ in $K_{m,n}$ and $y \in \setdef{0,1}$ satisfying $e_1,e_2 \in M$ if $y = 1$.
  Clearly, $\chi(M)$ satisfies Constraints~\eqref{ConstraintMatchingNonnegative} and~\eqref{ConstraintMatchingDegree}, and $y$ satisfies~\eqref{ConstraintMatchingBound}.

  Let $S \in \facetsDown$, define $\bar{S} \coloneqq S \setminus \setdef{u_1,u_2,w_1,w_2}$, and observe that $|\bar{S}|$ is odd.
  If $y = 1$, then $e_1,e_2 \in M$, i.e., Constraint~\eqref{ConstraintMatchingQuadraticGood} is satisfied. 
  Hence, only nodes in $\bar{S}$ can be matched to other nodes in $S$, and there are at most $\lfloor |\bar{S}|/2 \rfloor = (|S|-3)/2$ of them.
  If $y = 0$, then the validity follows from the fact that $S$ has odd cardinality.
  This shows that Constraint~\eqref{ConstraintMatchingQuadraticDown} is always satisfied.

\medskip

  To show $P \subseteq \PmatchOneDown$,
  we consider a vertex $(\hat{x},\hat{y})$ of $P$.
  Note that since $P$ is rational we have $(\hat{x},\hat{y}) \in \Q^E \times \Q$.
  If it satisfies at least one of 
  the Inequalities~\eqref{ConstraintMatchingQuadraticGood} for some $i^* \in \setdef{1,2}$
  or~\eqref{ConstraintMatchingQuadraticDown} for some $S^* \in \facetsDown$ with equality,
  Lemma~\ref{TheoremDownCombination} yields that $(\hat{x},\hat{y})$ is a convex combination of 
  vertices of $\PmatchOne$, which are vertices of $\PmatchOneDown$.

  Hence, $(\hat{x},\hat{y})$ is even a vertex of the polytope defined only by the
  Constraints~\eqref{ConstraintMatchingNonnegative},
  \eqref{ConstraintMatchingDegree} 
  and~\eqref{ConstraintMatchingBound}.
  Thus, $\hat{y} \in \setdef{0,1}$ and $\hat{x} = \chi(M)$ for some matching $M$ in $K_{m,n}$.
  Since Inequalities~\eqref{ConstraintMatchingQuadraticGood} are \emph{strictly} satisfied,
  we must have $\hat{y} = 0$, which concludes the proof.
\end{proof}

\begin{theorem}
  \label{TheoremUpComplete}
  $\PmatchOneUp$ is equal to the set of $(x,y) \in \R^E \times \R$ that satisfy
  Constraints~\eqref{ConstraintMatchingNonnegative},
  \eqref{ConstraintMatchingDegree},
  \eqref{ConstraintMatchingBound} and
  \eqref{ConstraintMatchingQuadraticUp}.
\end{theorem}

\begin{proof}
  Let $P$ be the polytope defined by Constraints~\eqref{ConstraintMatchingNonnegative}, \eqref{ConstraintMatchingDegree}, \eqref{ConstraintMatchingBound} and~\eqref{ConstraintMatchingQuadraticUp}.
  We first show $\PmatchOneUp \subseteq P$ by showing $(\chi(M),y) \in P$ for all feasible integer pairs $(\chi(M),y)$, i.e., matchings $M$ in $K_{m,n}$ and $y \in \setdef{0,1}$ satisfying ($e_1 \notin M$ or $e_2 \notin M$) if $y = 0$.
  Clearly, $\chi(M)$ satisfies Constraints~\eqref{ConstraintMatchingNonnegative} and~\eqref{ConstraintMatchingDegree}, and $y$ satisfies~\eqref{ConstraintMatchingBound}.

  For $S \in \facetsUp$,
  $M$ contains at most $\frac{1}{2}(|S \cup e_1 \cup e_2|) = \frac{1}{2}|S| + 1$
  edges in $E[S] \cup \setdef{e_1,e_2}$.
  Thus, if $y = 1$, Constraint~\eqref{ConstraintMatchingQuadraticUp} is satisfied.
  If $y = 0$ and $e_1,e_2 \notin M$, then it is trivially satisfied.
  Otherwise, i.e., if $y = 0$ and $M$ contains exactly one of the two edges,
  we can assume w.l.o.g. $e_1 \in M$ and $e_2 \notin M$.
  Since $S \setminus e_1$ has odd cardinality,
  at most $|S|/2 - 1$ edges of $M$ can have both endnodes in $S$, 
  the constraint is also satisfied in this case.

\medskip

  To show $P \subseteq \PmatchOneUp$,
  we consider a vertex $(\hat{x},\hat{y})$ of $P$.
  Note that since $P$ is rational we have $(\hat{x},\hat{y}) \in \Q^E \times \Q$.
  If it satisfies at least one of 
  the Inequalities~\eqref{ConstraintMatchingQuadraticUp} for some $S^* \in \facetsUp$ with equality,
  Lemma~\ref{TheoremUpCombination} yields that $(\hat{x},\hat{y})$ is a convex combination of 
  vertices of $\PmatchOne$, which are vertices of $\PmatchOneUp$.

  Hence, $(\hat{x},\hat{y})$ is even a vertex of the polytope defined only by the
  Constraints~\eqref{ConstraintMatchingNonnegative},
  \eqref{ConstraintMatchingDegree} 
  and~\eqref{ConstraintMatchingBound}.
  Thus, $\hat{y} \in \setdef{0,1}$ and $\hat{x} = \chi(M)$ for some matching $M$ in $K_{m,n}$.
  If $\hat{y} = 0$, then Inequality~\eqref{ConstraintMatchingQuadraticUp} for $S = \setdef{u_1,w_2}$
  reads $x_{u_1,w_2} + x_{e_1} + x_{e_2} - 0 \leq 1$,
  and thus implies $e_1 \notin M$ or $e_2 \notin M$,
  which concludes the proof.
\end{proof}

\begin{theorem}
  \label{TheoremComplete}
  $\PmatchOne$ is equal to the set of $(x,y) \in \R^E \times \R$ that satisfy
  Constraints~\eqref{ConstraintMatchingNonnegative},
  \eqref{ConstraintMatchingDegree},
  \eqref{ConstraintMatchingBound},
  \eqref{ConstraintMatchingQuadraticGood},
  \eqref{ConstraintMatchingQuadraticDown} and
  \eqref{ConstraintMatchingQuadraticUp}, 
  i.e., $\PmatchOne = \PmatchOneDown \cap \PmatchOneUp$.
\end{theorem}

\begin{proof}
  Let $P$ be the polytope defined by Constraints~\eqref{ConstraintMatchingNonnegative},
  \eqref{ConstraintMatchingDegree},
  \eqref{ConstraintMatchingBound},
  \eqref{ConstraintMatchingQuadraticGood},
  \eqref{ConstraintMatchingQuadraticDown} and
  \eqref{ConstraintMatchingQuadraticUp}.
  By Theorems~\eqref{TheoremDownComplete} and~\eqref{TheoremUpComplete}
  we have $\PmatchOne \subseteq \PmatchOneDown \cap \PmatchOneUp = P$.

\medskip

  To show $P \subseteq \PmatchOne$,
  we consider a vertex $(\hat{x},\hat{y})$ of $P$.
  Note that since $P$ is rational we have $(\hat{x},\hat{y}) \in \Q^E \times \Q$.
  If it satisfies at least one of 
  the Inequalities~\eqref{ConstraintMatchingQuadraticGood} for some $i^* \in \setdef{1,2}$
  or~\eqref{ConstraintMatchingQuadraticDown} for some $S^* \in \facetsDown$ with equality,
  Lemma~\ref{TheoremDownCombination} yields that $(\hat{x},\hat{y})$ is a convex combination of 
  vertices of $\PmatchOne$.
  The same result holds by Lemma~\ref{TheoremUpCombination}
  if the point satisfies at least one of 
  the Inequalities~\eqref{ConstraintMatchingQuadraticUp} for some $S^* \in \facetsUp$ with equality.

  Hence, $(\hat{x},\hat{y})$ is even a vertex of the polytope defined only by the
  Constraints~\eqref{ConstraintMatchingNonnegative},
  \eqref{ConstraintMatchingDegree} 
  and~\eqref{ConstraintMatchingBound}.
  Thus, $\hat{y} \in \setdef{0,1}$ and $\hat{x} = \chi(M)$ for some matching $M$ in $K_{m,n}$.
  Inequalities~\eqref{ConstraintMatchingQuadraticGood} and Inequality~\eqref{ConstraintMatchingQuadraticUp}
  for $S = \setdef{u_1,w_2}$ imply that $y = 1$ if and only if $e_1,e_2 \in M$,
  which concludes the proof.
\end{proof}

\paragraph{Perfect matchings.}
\label{SectionPerfectMatchings}
We now assume $m = n$, since otherwise, $K_{m,n}$ does not contain perfect matchings.
Since the formulations for perfect matchings are obtained
by replacing Inequalities~\eqref{ConstraintMatchingDegree} by Equations~\eqref{ConstraintPerfectMatchingDegree},
the corresponding polytopes are faces of the ones defined in the Section~\ref{SectionIntroduction},
and we immediately obtain the following results from the corresponding theorems in Section~\ref{SectionResults}:

\begin{corollary}
  \label{TheoremPerfectDown}
  The convex hull of all $(\chi(M),y) \in \setdef{0,1}^E \times \setdef{0,1}$,
  for which $M$ is a perfect matching $M$ in $K_{n,n}$ and $y = 1$ implies
  $e_1,e_2 \in M$, is equal to the set of $(x,y) \in \R^E \times \R$ that satisfy
  Constraints~\eqref{ConstraintMatchingNonnegative},
  \eqref{ConstraintMatchingBound},
  \eqref{ConstraintMatchingQuadraticGood},
  \eqref{ConstraintMatchingQuadraticDown} and
  \eqref{ConstraintPerfectMatchingDegree}.
\end{corollary}

\begin{corollary}
  \label{TheoremPerfectUp}
  The convex hull of all $(\chi(M),y) \in \setdef{0,1}^E \times \setdef{0,1}$,
  for which $M$ is a perfect matching $M$ in $K_{n,n}$ and $y = 0$ implies
  $e_1 \notin M$ or $e_2 \notin M$, is equal to the set of $(x,y) \in \R^E \times \R$ that satisfy
  Constraints~\eqref{ConstraintMatchingNonnegative},
  \eqref{ConstraintMatchingBound},
  \eqref{ConstraintMatchingQuadraticUp} and
  \eqref{ConstraintPerfectMatchingDegree}.
\end{corollary}

\begin{corollary}
  \label{TheoremPerfectComplete}
  The convex hull of all $(\chi(M),y) \in \setdef{0,1}^E \times \setdef{0,1}$,
  for which $M$ is a perfect matching $M$ in $K_{n,n}$ and $y = 1$ if and only if
  $e_1,e_2 \in M$, is equal to the set of $(x,y) \in \R^E \times \R$ that satisfy
  Constraints~\eqref{ConstraintMatchingNonnegative},
  \eqref{ConstraintMatchingBound},
  \eqref{ConstraintMatchingQuadraticGood},
  \eqref{ConstraintMatchingQuadraticDown},
  \eqref{ConstraintMatchingQuadraticUp} and
  \eqref{ConstraintPerfectMatchingDegree}.
\end{corollary}

\section{Proofs of main lemmas}
\label{SectionTechnique}

The technique we will use to proof Lemmas~\ref{TheoremDownCombination} and~\ref{TheoremUpCombination} is quite technical.
Hence, we present it in this section in a more abstract fashion (see Figure~\ref{FigureProofTechnique}).
To make the proofs more accessible, we also list the required steps that have to be done.
Consider, a description of a polytope $P$ in terms of linear inequalities for which
we want to show $P = \conv(X)$ for some (implicitly known) $X$.
\begin{enumerate}
\item
  Consider an initial fractional point of $P$ that satisfies a certain inequality with equality.
\item
  Modify the point such that the resulting point lies in a face $F$ of a polytope $Q$ that we have under control.
  \emph{Prove that the modified point lies in $F$ (and hence in $Q$)}.
\item
  Write the modified point as a (special) convex combination of vertices of $F$.
  \emph{Derive structural properties that are implied by the fact that the combination
  uses only points from $F$.}
\item
  Revert the modifications by replacing some of the vertices in the convex combination by others.
  \emph{Prove that the new vertices are contained in $X$. Prove that the modifications revert those of Step~2,
  i.e., that their convex combination equals the initial point.}
\end{enumerate}

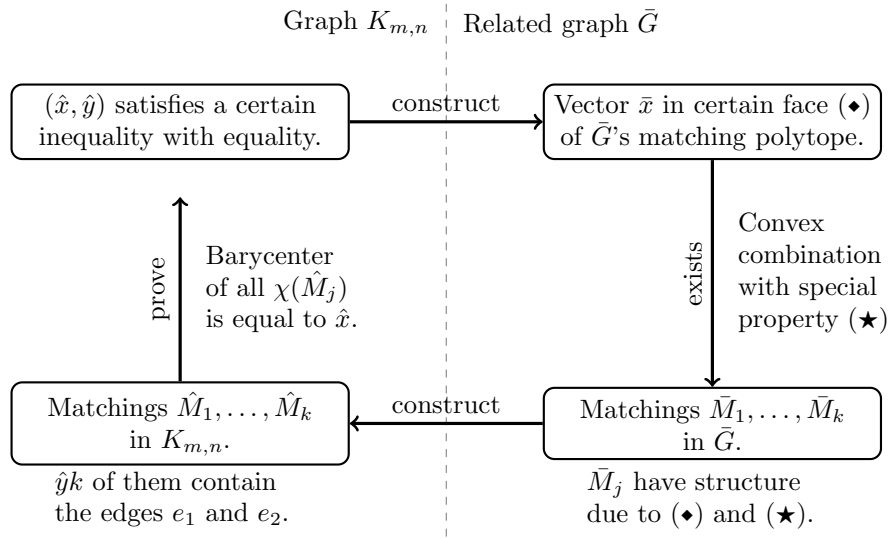
\begin{figure}[htpb]
  \begin{center}
    \begin{tikzpicture}
      \piInput{tikz-quadratic-bipartite-matching-strategy}[%
        main graph, auxiliary graph, hat point, bar point, bar matchings, convex combination,
        bar matchings structure, hat matchings, hat matchings portion, hat matchings barycenter,
        beamer=false
      ]
    \end{tikzpicture}
  \end{center}
  \vspace{-1em}
  \caption{Proof technique for Lemmas~\ref{TheoremDownCombination} and~\ref{TheoremUpCombination}.}
  \label{FigureProofTechnique}
\end{figure}

\subsection{Downward monotonization}
\label{SectionDown}

This section contains the proof of Lemma~\ref{TheoremDownCombination}.
We first introduce relevant objects which are fixed for the rest of this section,
and then present the main proof.
To improve readability, the proofs of several claims are deferred to the end of this section.

Let $(\hat{x},\hat{y}) \in \Q^E \times \Q$ be as stated in the lemma, i.e., it satisfies Constraints~\eqref{ConstraintMatchingNonnegative}, \eqref{ConstraintMatchingDegree}, \eqref{ConstraintMatchingBound}, \eqref{ConstraintMatchingQuadraticGood} and~\eqref{ConstraintMatchingQuadraticDown},
and it satisfies at least one of the Inequalities~\eqref{ConstraintMatchingQuadraticGood} for $i^* \in \setdef{1,2}$
or~\eqref{ConstraintMatchingQuadraticDown} for a set $S^* \in \facetsDown$ with equality.

Let $\bar{G} = (U \dcup W, \bar{E})$ be the graph $K_{m,n}$
with the additional edges $e_u \coloneqq \setdef{u_1,u_2}$ and $e_w \coloneqq \setdef{w_1,w_2}$,
i.e., $\bar{E} \coloneqq E \cup \setdef{e_u, e_w}$.
Define the vector $\bar{x} \in \Q^{\bar{E}}$ as follows (see Figure~\ref{FigureDownGadget}):
\begin{itemize}
\item
  $\bar{x}_e \coloneqq \hat{x}_e$ for all $e \in E \setminus \setdef{e_1,e_2}$.
\item
  $\bar{x}_{e_i} \coloneqq \hat{x}_{e_i} - \hat{y}$ for $i=1,2$.
\item
  $\bar{x}_{e_u} \coloneqq \bar{x}_{e_w} \coloneqq \hat{y}$.
\end{itemize}

\begin{figure}[htpb]
  \begin{center}
    \begin{tikzpicture}
      \piInput{tikz-quadratic-bipartite-matching-graphs}[%
        down gadget
      ]
    \end{tikzpicture}
  \end{center}
  \vspace{-1em}
  \caption{Graph $\bar{G}$ and vector $\bar{x}$ in the proof of Lemma~\ref{TheoremDownCombination}.}
  \label{FigureDownGadget}
\end{figure}
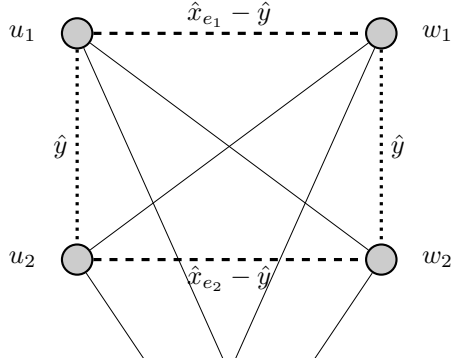

\begin{claim}
  \label{TheoremDownInMatchingPolytope}
  $\bar{x}$ is in the matching polytope of $\bar{G}$.
\end{claim}

By Claim~\ref{TheoremDownInMatchingPolytope}, and since $\bar{x}$ is rational,
it can be written as a convex combination of characteristic vectors of matchings
using only rational multipliers.
Multiplying with a sufficiently large integer $k$, we obtain that
$k \bar{x} = \sum_{j=1}^k \chi(\bar{M}_j)$
for matchings $\bar{M}_1, \ldots, \bar{M}_k$ in $\bar{G}$, where matchings may occur multiple times.
Let $J_u \coloneqq \setdef{ j \in [k] }[ e_u \in \bar{M}_j ]$ and $J_w \coloneqq \setdef{ j \in [k] }[ e_w \in \bar{M}_j ]$
(using the notation $[k] \coloneqq \setdef{1,2,\ldots,k}$),
and observe that $|J_u| = \hat{y} k = |J_w|$.
We may assume that the convex combination is chosen such that
$|J_u \setminus J_w|$ is minimum.

\begin{claim}
  \label{TheoremDownMinimalCombination}
  The convex combination satisfies $J_u = J_w$.
\end{claim}

By Claim~\ref{TheoremDownMinimalCombination} we can write $J \coloneqq J_u = J_w$.
We construct matchings $\hat{M}_j$ for $j \in [k]$ that are related to the corresponding $\bar{M}_j$.
To this end, let $C \coloneqq \setdef{e_1, e_2, e_u, e_w}$ and define $\hat{M}_j \coloneqq \bar{M}_j \Delta C$ for
all $j \in J$ and $\hat{M}_j \coloneqq \bar{M}_j$ for all $j \in [k] \setminus J$.
All $\hat{M}_j$ are matchings in $\bar{G}$ since for all $j \in J$,
the matchings $\bar{M}_j$ contain both edges $e_u$ and $e_w$.
In fact, none of the matchings $\hat{M}_j$ contains these edges, and hence they are even matchings in $K_{m,n}$.
In the following claim we exploit this property and consider the vectors $\chi(\hat{M}_j)$ with entries indexed by edges in $E$.

\begin{claim}
  \label{TheoremDownBarycenter}
  We have
  $\hat{x} = \frac{1}{k} \sum_{j=1}^k \chi(\hat{M_j})$.
\end{claim}

Together with $\hat{y}k = |J|$, Claim~\ref{TheoremDownBarycenter} yields
\begin{align*}
  (\hat{x},\hat{y}) &= \frac{1}{k} \left( \sum_{j \in J} (\chi(\hat{M}_j),1) + \sum_{j \in [k] \setminus J} (\chi(\hat{M}_j),0) \right),
\end{align*}
and it remains to prove that all participating vectors are actually feasible for $\PmatchOne$.
For the first sum, this is easy to see, since for all $j \in J$, 
the matchings $\hat{M}_j$ contain both edges $e_1$ and $e_2$ by construction.
The matchings in the second sum are considered in two claims, depending on $(\hat{x},\hat{y})$.

\begin{claim}
  \label{TheoremDownTightGood}
  Let $(\hat{x},\hat{y})$ satisfy Inequality~\eqref{ConstraintMatchingQuadraticGood} for some $i^* \in \setdef{1,2}$ with equality.
  Then $\hat{M}_j$ contains at most one of the two edges $e_1$, $e_2$ for all $j \in [k] \setminus J$.
\end{claim}

\begin{claim}
  \label{TheoremDownTightDown}
  Let $(\hat{x},\hat{y})$ satisfy Inequality~\eqref{ConstraintMatchingQuadraticDown} for some $S^* \in \facetsDown$ with equality.
  Then $\hat{M}_j$ contains at most one of the two edges $e_1$, $e_2$ for all $j \in [k] \setminus J$.
\end{claim}

Since, by the assumptions of Lemma~\ref{TheoremDownCombination}, the premise of at least one of the Claims~\ref{TheoremDownTightGood}
or~\ref{TheoremDownTightDown} is satisfied, $(\hat{x},\hat{y})$ is indeed a convex combination of
vertices of $\PmatchOne$, which concludes the proof of Lemma~\ref{TheoremDownCombination}.
\hfill $\qed$

\bigskip

Before actually proving the claims of this section, we list some implied valid inequalities that will turn out to be useful.

\begin{proposition}
  \label{TheoremDownInequalities}
  Let $(\hat{x},\hat{y})$ satisfy
  Constraints~\eqref{ConstraintMatchingNonnegative}, \eqref{ConstraintMatchingDegree},\eqref{ConstraintMatchingBound},
  \eqref{ConstraintMatchingQuadraticGood} and~\eqref{ConstraintMatchingQuadraticDown},
  and define
  $$\inequalitiesDown\coloneqq \{ S \subseteq U \dcup W : \text{$|S|$ is odd and $S \cap \Vspecial \in \setdef{ \setdef{u_1,u_2}, \setdef{w_1,w_2} }$} \}.$$
  Then $(\hat{x},\hat{y})$ satisfies
  $x(E[S]) + y \leq \frac{1}{2}(|S|-1)$ (i.e., Inequality~\eqref{ConstraintMatchingQuadraticDown}) for all $S \in \inequalitiesDown \supseteq \facetsDown$.
\end{proposition}

\begin{proof}[Proof of Proposition~\ref{TheoremDownInequalities}]
  We only have to prove the statement for $S \in \inequalitiesDown \setminus \facetsDown$.
  W.l.o.g.\ we assume that $S \cap \Vspecial = \setdef{u_1,u_2}$, since the other case is similar.
  Let $U' \coloneqq S \cap U$ and $W' \coloneqq S \cap W$, and remember that we assume $|U'| \neq |W'| + 1$.

  If $|U'| < |W'| + 1$, we have $|U'| \leq |W'| - 1$ since $|S|$ is odd.
  Then the sum of $\hat{x}(\delta(u)) \leq 1$ for all $u \in U'$ plus the sum of $-\hat{x}_e \leq 0$ for all $e \in \delta(U') \setminus (E[S] \cup \setdef{e_1})$ reads $\hat{x}(E[S]) + \hat{x}_{e_1} \leq |U'| \leq \frac{1}{2}(|S| - 1)$.
  Adding $\hat{y} \leq \hat{x}_{e_1}$ yields the desired inequality.
  
  If $|U'| > |W'| + 1$, we have $|U'| \geq |W'| + 3$ since $|S|$ is odd.
  Then the sum of $\hat{x}(\delta(w)) \leq 1$ for all $w \in W'$ plus the sum of $-\hat{x}_e \leq 0$ for all $e \in \delta(W') \setminus E[S]$ reads $\hat{x}(E[S]) \leq |W'| \leq \frac{1}{2}(|S|-1)$.
  Adding $\hat{y} \leq 1$ yields the desired inequality, which concludes the proof.
\end{proof}

\begin{proof}[Proof of Claim~\ref{TheoremDownInMatchingPolytope}]
  From $\hat{x} \geq \zerovec$ and~\eqref{ConstraintMatchingQuadraticGood} we obtain that also $\bar{x} \geq \zerovec$.
  By construction and since $\hat{x}$ satisfies~\eqref{ConstraintMatchingDegree}, $\bar{x}(\delta(v)) \leq 1$ for every node $v \in U \dcup W$.
  
  Suppose, for the sake of contradiction, that $\bar{x}(E[S]) > \frac{1}{2}(|S|-1)$ for some odd-cardinality set $S \subseteq U \dcup W$.
  From $\hat{x}(E[S]) \leq \frac{1}{2}(|S|-1)$ we deduce $\bar{x}(E[S]) > \hat{x}(E[S])$,
  i.e., $E[S]$ contains at least one of the edges $\{e_u,e_w\}$, since only for these edges the $\bar{x}$-value is strictly greater than
  the corresponding $\hat{x}$-value.
  Observe that $E[S]$ also must contain at most one of these edges, since otherwise it would also contain the two edges
  $e_1,e_2$, which yielded $\bar{x}(E[S]) = \hat{x}(E[S]) \leq \frac{1}{2}(|S|-1)$.
  Hence, we have $S \in \inequalitiesDown$,
  and thus $\bar{x}(E[S]) = \hat{x}(E[S]) + \hat{y} \leq \frac{1}{2}(|S|-1)$ by Proposition~\ref{TheoremDownInequalities}.
  This proves that $\bar{x}$ is in the matching polytope of $\bar{G}$.
\end{proof}

\begin{proof}[Proof of Claim~\ref{TheoremDownMinimalCombination}]
  Suppose, for the sake of contradiction, that $J_u \neq J_w$.
  Let $j_u \in J_u \setminus J_w$ and let $j_w \in J_w \setminus J_u$, which exist due to $|J_u| = |J_w|$.
  Consider the matchings $\bar{M}_{j_u}$ and $\bar{M}_{j_w}$ and note that
  $\bar{M}_{j_u} \Delta \bar{M}_{j_w}$ contains both edges $e_u$ and $e_w$.
  Let $C_u$ and $C_w$ be (the edge sets of) the connected components of $\bar{M}_{j_u} \Delta \bar{M}_{j_w}$
  that contain $e_u$ and $e_w$, respectively.
  
  We claim that $C_u$ and $C_w$ are not the same component.
  Assume, for the sake of contradiction, that $C \coloneqq C_u = C_w$ is a connected component
  (i.e., an alternating cycle or path) of $M_{j_u} \Delta M_{j_w}$ that contains $e_u$ and $e_w$.
  Consider a path $P \subseteq C \setminus \setdef{e_u,e_w}$ that connects an endnode of $e_u$ with an endnode of $e_w$
  (if $C$ is an alternating cycle, there exist two such paths and we pick one arbitrarily).
  On the one hand, $(U \dcup W, \bar{E} \setminus \setdef{e_u,e_w})$ is bipartite and thus $P$ must have odd length.
  On the other hand, $e_u \in \bar{M}_{j_u}$ and $e_w \in \bar{M}_{j_w}$,
  and hence $P$ must have even length, yielding a contradiction.

  Define two new matchings $\bar{M}'_{j_u} \coloneqq \bar{M}_{j_u} \Delta C_u$ and $\bar{M}'_{j_w} \coloneqq \bar{M}_{j_w} \Delta C_u$,
  and note that $\chi(\bar{M}_{j_u}) + \chi(\bar{M}_{j_w}) = \chi(\bar{M}'_{j_u}) + \chi(\bar{M}'_{j_w})$,
  i.e., we can replace $\bar{M}_{j_u}$ and $\bar{M}_{j_w}$ by $\bar{M}'_{j_u}$ and $\bar{M}'_{j_w}$ in the convex combination.
  The fact that $\bar{M}'_{j_u}$ contains none of the two edges $e_u$ and $e_w$, while
  $\bar{M}'_{j_w}$ contains both, yields a contradiction to the assumption that $|J_u \setminus J_w|$ is minimum.
\end{proof}

\begin{proof}[Proof of Claim~\ref{TheoremDownBarycenter}]
  Consider the vector
  $d \coloneqq \sum_{j=1}^k (\chi(\hat{M}_j) - \chi(\bar{M}_j))$.
  By construction of the $\hat{M}_j$, we have $d_e = 0$ for all $e \notin C$
  and $d_{e_1} = d_{e_2} = -d_{e_u} = -d_{e_w} = |J| = k \hat{y}$.
  A simple comparison with the construction of $\bar{x}$ from $\hat{x}$ concludes the proof.
\end{proof}

\begin{proof}[Proof of Claim~\ref{TheoremDownTightGood}]
  From $\hat{x}_{e_{i^*}} = \hat{y}$ we obtain that $\bar{x}_{e_{i^*}} = 0$, and thus $e_{i^*} \notin \bar{M}_j$.
  Since $j \notin J$, we have $\hat{M}_j = \bar{M}_j$, which concludes the proof.
\end{proof}

\begin{proof}[Proof of Claim~\ref{TheoremDownTightDown}]
  From $\hat{x}(E[S^*]) + \hat{y} = \frac{1}{2}(|S^*|-1)$ and the construction of $\bar{x}$
  we obtain $\bar{x}(E[S^*]) = \frac{1}{2}(|S^*|-1)$.
  But since $x(E[S^*]) \leq \frac{1}{2}(|S^*|-1)$ is valid for all $\chi(\bar{M}_j)$, equality must hold
  for all $j \in [k]$. 
  Thus, $|\bar{M}_j \cap \setdef{e_1,e_2}| \leq |\bar{M}_j \cap \delta(S^*)| \leq 1$
  for all $j$, which concludes the proof.
\end{proof}

\subsection{Upward monotonization}
\label{SectionUp}

This section contains the proof of Lemma~\ref{TheoremUpCombination}.
The setup is similar to that of the previous section, starting with the relevant objects.

Let $(\hat{x},\hat{y}) \in \Q^E \times \Q$ be as stated in the lemma, i.e., it satisfies Constraints~\eqref{ConstraintMatchingNonnegative}, \eqref{ConstraintMatchingDegree}, \eqref{ConstraintMatchingBound}, and~\eqref{ConstraintMatchingQuadraticUp}, and it satisfies at least one of the Inequalities~\eqref{ConstraintMatchingQuadraticUp} for a set $S^* \in \facetsUp$ with equality.

Let $\bar{G} = (\bar{V}, \bar{E})$ be the graph $K_{m,n}$
with two additional nodes $a$ and $b$, i.e., $\bar{V} = U \dcup W \dcup \setdef{a,b}$,
and edge set $\bar{E} \coloneqq E \cup \setdef{ \setdef{a,b}, \setdef{u_1,a}, \setdef{u_2,b}, \setdef{w_1,b}, \setdef{w_2,a} }$.
Define two vectors $\tilde{x},\bar{x} \in \R^{\bar{E}}$ as follows (see Figure~\ref{FigureUpGadget}):
\begin{itemize}
\item
  $\tilde{x}_e \coloneqq \hat{x}_e$ and $\bar{x}_e \coloneqq \hat{x}_e$ for all $e \in E \setminus \setdef{e_1,e_2}$.
\item
  $\tilde{x}_{e_i} \coloneqq \hat{x}_{e_i}$ and $\bar{x}_{e_i} \coloneqq \frac{1}{2}\hat{y}$ for $i=1,2$.
\item 
  $\tilde{x}_{\setdef{a,b}} \coloneqq 1$ and $\bar{x}_{\setdef{a,b}} \coloneqq 1 - \hat{x}_{e_1} - \hat{x}_{e_2} + \hat{y}$.
\item
  $\tilde{x}_{\setdef{u_1,a}} \coloneqq \tilde{x}_{\setdef{w_1,b}} \coloneqq 0$ and $\bar{x}_{\setdef{u_1,a}} \coloneqq \bar{x}_{\setdef{w_1,b}} \coloneqq \hat{x}_{e_1} - \frac{1}{2}\hat{y}$.
\item
  $\tilde{x}_{\setdef{u_2,b}} \coloneqq \tilde{x}_{\setdef{w_2,a}} \coloneqq 0$ and $\bar{x}_{\setdef{u_2,b}} \coloneqq \bar{x}_{\setdef{w_2,a}} \coloneqq \hat{x}_{e_2} - \frac{1}{2}\hat{y}$.
\end{itemize}
The vector $\tilde{x}$ is essentially a trivial lifting of $\hat{x}$ into $\R^{\bar{E}}$ by setting
the value for edge $\setdef{a,b}$ to $1$ and the values for the other new edges to $0$.
It is easy to see that $\tilde{x}$ is in the matching polytope of $\bar{G}$.

The vector $\bar{x}$ is a modification of $\tilde{x}$ on the edges of
the following two cycles:
\begin{align*}
  C_1 &\coloneqq \setdef{\setdef{u_1,a}, \setdef{a,b}, \setdef{b,w_1}, \setdef{w_1,u_1} }, \\
  C_2 &\coloneqq \setdef{\setdef{u_2,b}, \setdef{b,a}, \setdef{a,w_2}, \setdef{w_2,u_2} }.
\end{align*}
The values on the two opposite (in $C_1$) edges $\setdef{u_1,w_1}$ and $\setdef{a,b}$ are decreased
by $\hat{x}_{e_1} - \frac{1}{2}\hat{y}$, and increased by the same value on the other two edges.
Similarly, the values on the edges $\setdef{u_2,w_2}$ and $\setdef{a,b}$ are decreased by $\hat{x}_{e_2} - \frac{1}{2}\hat{y}$,
while they are increased by the same value on the other two edges of $C_2$.

\begin{figure}[htpb]
  \begin{center}
    \begin{tikzpicture}
      \piInput{tikz-quadratic-bipartite-matching-graphs}[%
        up gadget
      ]
    \end{tikzpicture}
    \hspace{5mm}
    \begin{tikzpicture}
      \piInput{tikz-quadratic-bipartite-matching-graphs}[%
        up cycles
      ]
    \end{tikzpicture}
  \end{center}
  \vspace{-1em}
  \caption{Graph $\bar{G}$, vector $\bar{x}$ and cycles $C_1$ and $C_2$ in the proof of Lemma~\ref{TheoremUpCombination}.}
  \label{FigureUpGadget}
\end{figure}
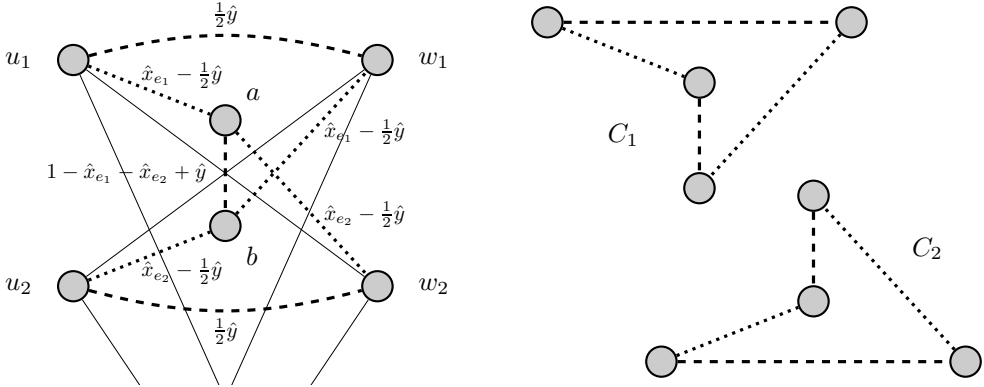

\begin{claim}
  \label{TheoremUpInMatchingPolytope}
  $\bar{x}$ is in the matching polytope of $\bar{G}$.
\end{claim}

By Claim~\ref{TheoremUpInMatchingPolytope}, and since $\bar{x}$ is rational,
it can be written as a convex combination of characteristic vectors of matchings
using only rational multipliers.
Multiplying with a sufficiently large integer $k$, we obtain
$k \bar{x} = \sum_{j=1}^k \chi(\bar{M}_j)$
for matchings $\bar{M}_1, \ldots, \bar{M}_k$ in $\bar{G}$, where matchings may occur multiple times.
We define the index sets
\begin{alignat*}{6}
  J_u  &\coloneqq \setdef{ j \in [k] }[ \setdef{u_1,a},\setdef{u_2,b} \in \bar{M}_j ],         \quad&  J_w &\coloneqq \setdef{ j \in [k] }[ \setdef{w_1,b},\setdef{w_2,a} \in \bar{M}_j ], \\
  J_1  &\coloneqq \setdef{ j \in [k] }[ \setdef{u_1,a}, \setdef{w_1,b} \in \bar{M}_j ],        \quad&  J_2 &\coloneqq \setdef{ j \in [k] }[ \setdef{u_2,b}, \setdef{w_2,a} \in \bar{M}_j ], \\
  J'_1 &\coloneqq \setdef{ j \in [k] }[ \setdef{u_2,w_2} \in \bar{M}_j ],                      \quad& J'_2 &\coloneqq \setdef{ j \in [k] }[ \setdef{u_1,w_1} \in \bar{M}_j ] \text{ and } \\
  N    &\coloneqq \setdef{ j \in [k] }[ \setdef{a,b} \in \bar{M}_j ].
\end{alignat*}
We assume that the convex combination is chosen such that $|J_u| + |J_w|$ is minimum.

Using the assumption from the lemma, that $(\hat{x},\hat{y})$ satisfies Inequality~\eqref{ConstraintMatchingQuadraticUp} for some set $S^* \in \facetsUp$ with equality,
we can derive the following statement.

\begin{claim}
  \label{TheoremUpCut}
  For every $j \in [k]$, the matching $\bar{M}_j$ contains at most one of the
  edges $e_1$, $e_2$, or $\setdef{a,b}$.
  It furthermore matches $a$ and $b$ (not necessarily to each other).
\end{claim}

\begin{claim}
  \label{TheoremUpMinimalCombination}
  The convex combination satisfies $J_u = J_w = \emptyset$,
  and $J_1 \dcup J_2 \dcup N$ is a partitioning of $[k]$.
\end{claim}

\begin{claim}
  \label{TheoremUpIndexSetContained}
  We have $J'_i \subseteq J_i$ for $i=1,2$ and thus $J'_1$ and $J'_2$ are disjoint.
\end{claim}

\begin{claim}
  \label{TheoremUpIndexSetFraction}
  We have $|J'_1 \dcup J'_2| = \hat{y}k$.
\end{claim}

We construct matchings $\tilde{M}_j$ and $\hat{M}_j$ for $j \in [k]$ that are related to the corresponding $\bar{M}_j$.
Define $\tilde{M}_j \coloneqq \bar{M}_j \Delta C_1$ for
all $j \in J_1$, $\tilde{M}_j \coloneqq \bar{M}_j \Delta C_2$ for all $j \in J_2$.
By Claim~\ref{TheoremUpMinimalCombination},
all remaining indices are the $j \in N$, and for those we define
$\tilde{M}_j \coloneqq \bar{M}_j$.
All $\tilde{M}_j$ are matchings in $\bar{G}$ since for all $j \in J_i$ ($i=1,2$)
the cycle $C_i$ is an $\bar{M}_j$-alternating cycle.
We define $\hat{M}_j \coloneqq \tilde{M}_j \setminus \setdef{a,b}$ for all $j \in [k]$,
which are matchings in $K_{m,n}$
since $\setdef{a,b} \in \tilde{M}_j$ for all $j \in [k]$.

In the following claim we exploit this property and consider the vectors $\chi(\tilde{M}_j)$ and $\chi(\hat{M}_j)$ with entries indexed by edges in $\bar{E}$ and $E$, respectively.
\begin{claim}
  \label{TheoremUpBarycenter}
  We have
  $\tilde{x} = \frac{1}{k} \sum_{j=1}^k \chi(\tilde{M}_j)$ and 
  $\hat{x} = \frac{1}{k} \sum_{j=1}^k \chi(\hat{M}_j)$.
\end{claim}

Claims~\ref{TheoremUpIndexSetFraction} and~\ref{TheoremUpBarycenter} yield
\begin{align*}
  (\hat{x},\hat{y}) &= 
    \frac{1}{k} \left( \sum_{j \in J'_1} (\chi(\hat{M}_j),1)
    + \sum_{j \in J'_2} (\chi(\hat{M}_j),1)
    + \sum_{j \in [k] \setminus (J'_1 \dcup J'_2)} (\chi(\hat{M}_j),0) \right),
\end{align*}
and it remains to prove that all participating vectors are actually feasible for $\PmatchOne$.

To this end, let $j \in J'_1$ and observe that
$\setdef{u_2,w_2} \in \bar{M}_j$ and, by Claim~\ref{TheoremUpIndexSetContained}, $\setdef{u_1,a},\setdef{w_1,b} \in \bar{M}_j$.
Thus, the symmetric difference with $C_1$ yields $\setdef{u_1,w_1},\setdef{u_2,w_2} \in \hat{M}_j$.
Similarly, we have $\setdef{u_1,w_1},\setdef{u_2,w_2} \in \hat{M}_j$ for all $j \in J'_2$.
Let $j \in [k] \setminus (J'_1 \dcup J'_2)$.
First, $\bar{M}_j$ contains none of the edges $\setdef{u_1,w_1}$, $\setdef{u_2,w_2}$.
Second, the construction of $\hat{M}_j$ from $\bar{M}_j$ adds at most
one of the two edges $\setdef{u_1,w_1}$, $\setdef{u_2,w_2}$,
which proves that $\hat{M}_j$ does not contain both of them.
This concludes the proof.
\hfill $\qed$

\bigskip

Before actually proving the claims of this section,
we list further valid inequalities.

\DeclareDocumentCommand\LabelTheoremUpInequalitiesGood{}{(d)}

\begin{proposition}
  \label{TheoremUpInequalities}
  Let $(\hat{x},\hat{y})$ satisfy Constraints~\eqref{ConstraintMatchingNonnegative},\eqref{ConstraintMatchingDegree} and~\eqref{ConstraintMatchingBound}
  as well as Inequality~\eqref{ConstraintMatchingQuadraticUp} for $S = \setdef{u_1, w_2}$.
  Define
  \begin{gather*}
    \inequalitiesUp\coloneqq \setdef{ S \subseteq U \dcup W }[ \text{$|S|$ is even and $S \cap \Vspecial \in \setdef{ \setdef{u_1,w_2}, \setdef{u_2,w_1} }$ } ].
  \end{gather*}
  Then $(\hat{x},\hat{y})$ satisfies the following inequalities:

  \begin{enumerate}[(a)]
  \item
    \label{TheoremUpInequalitiesBad}
    $x_{e_1} + x_{e_2} - y \leq 1$.
  \item
    \label{TheoremUpInequalitiesHalfUp}
    $x(E[S]) + x_{e_i} - \frac{1}{2}y \leq \frac{1}{2}|S|$ for $i \in \setdef{1,2}$, $S \subseteq U \dcup W$ with $|S|$ even and $e_i \in \delta(S)$.
  \item
    \label{TheoremUpInequalitiesFullUp}
    $x(E[S]) + x_{e_1} + x_{e_2} - y \leq \frac{1}{2}|S|$ (i.e., Inequality~\eqref{ConstraintMatchingQuadraticUp}) for all $S \in \inequalitiesUp \supseteq \facetsUp$
  \end{enumerate}
  If $(\hat{x},\hat{y})$ satisfies Inequality~\eqref{ConstraintMatchingQuadraticUp} for some $S^* \in \facetsUp$ with equality or violates that equality, then the following inequalities hold as well:
  \begin{enumerate}[(a)]
  \item[\LabelTheoremUpInequalitiesGood{}]
    $\hat{x}_{e_i} - \frac{1}{2}\hat{y} \geq 0$ for $i=1,2$.
  \end{enumerate}
\end{proposition}

\begin{proof}[Proof of Proposition~\ref{TheoremUpInequalities}]
  We prove validity for each inequality individually:
  \begin{enumerate}[(a)]
  \item
    The inequality is the sum of Inequality~\eqref{ConstraintMatchingQuadraticUp} for $S = \setdef{u_1, w_2}$ and $-\hat{x}_{\setdef{u_1,w_2}} \leq 0$.
  \item
    Since $|S \cup e_i|$ is odd (and since $K_{m,n}$ is bipartite), the Blossom Inequality $\hat{x}(E[S \cup e_i]) \leq \frac{1}{2}|S|$ is implied by Constraints~\eqref{ConstraintMatchingNonnegative} and~\eqref{ConstraintMatchingDegree}.
    Adding $-\hat{x}_e \leq 0$ for all $e \in E[S \cup e_1] \setminus (E[S] \cup \setdef{e_1})$ and $-\frac{1}{2}\hat{y} \leq 0$ yields the desired inequality.
  \item
    We only have to prove the statement for $S \in \inequalitiesUp \setminus \facetsUp$.
    We can furthermore assume w.l.o.g.\ $S \cap \Vspecial = \setdef{u_1,w_2}$ and $|S \cap U| < |S \cap W|$, since the other cases are similar.
    Let $U' \coloneqq S \cap U$ and $W' \coloneqq S \cap W$ and observe that $|U'| \leq |W'| - 2$ because $|S|$ is even.
    Then the sum of $\hat{x}(\delta(u)) \leq 1$ for all $u \in U'$ plus the sum of $-\hat{x}_e \leq 0$ for all $e \in \delta(U') \setminus (E[S] \cup \setdef{e_1})$ reads $\hat{x}(E[S]) + x_{e_1} \leq |U'| = \frac{1}{2}|S| -1$.
    Adding $\hat{x}_{e_2} \leq 1$ and $-\hat{y} \leq 0$ yields the desired inequality.
  \item
    Let $i \in \setdef{1,2}$ and $j \coloneqq 3-i$.
    Similar to the proof of~\eqref{TheoremUpInequalitiesHalfUp} we have that the Blossom inequality $\hat{x}(E[S^*]) + \hat{x}_{e_j} \leq \frac{1}{2}|S^*|$ is implied by Constraints~\eqref{ConstraintMatchingNonnegative} and~\eqref{ConstraintMatchingDegree}.
    Subtracting this from $\hat{x}(E[S^*]) + \hat{x}_{e_1} + \hat{x}_{e_2} - \hat{y} \geq \frac{1}{2}|S^*|$ and adding $\frac{1}{2} \hat{y} \geq 0$ we derive $\hat{x}_{e_i} - \frac{1}{2} \hat{y} \geq 0$.
  \end{enumerate}
  This concludes the proof.
\end{proof}

In this section we are in the situation that Inequality~\eqref{ConstraintMatchingQuadraticUp} is satisfied for all $S \in \facetsUp$ (and not just for $S = \setdef{u_1,w_2}$) and that it is satisfied with equality for $S^*$.
In Section~\ref{SectionSeparation} we will discuss separation algorithms, for which we need the refined conditions, i.e., we will exploit that the inequality only has to be satisfied for the single set $S$ and that an inequality may be violated by the given point.

\begin{proof}[Proof of Claim~\ref{TheoremUpInMatchingPolytope}]
  Since $(\hat{x},\hat{y}) \geq \zerovec$, Parts~\eqref{TheoremUpInequalitiesBad} and \LabelTheoremUpInequalitiesGood{} of Proposition~\ref{TheoremUpInequalities}
  yield $\bar{x} \geq \zerovec$.
  The degree constraints are also satisfied, since $\bar{x}(\delta(w)) = \hat{x}(\delta(w))$ for the nodes $w \in \Vspecial$
  and since $\bar{x}(\delta(a)) = \bar{x}(\delta(b)) = 1$.

  Suppose, for the sake of contradiction, that $\bar{x}(E[\bar{S}]) > \frac{1}{2}(\bar{S}|-1)$
  for some odd-cardinality set $\bar{S} \subseteq \bar{V}$.
  Clearly, $\tilde{x}(E[\bar{S}]) \leq \frac{1}{2}(|\bar{S}|-1)$,
  i.e., $(\bar{x} - \tilde{x})(E[\bar{S}]) > 0$.
  This implies that $E[\bar{S}]$ must intersect some $C_i$ ($i=1,2$) in such a way that the sum of the respective modifications (increase or decrease by $\hat{x}_{e_i} - \frac{1}{2}\hat{y}$) is positive.
  Similar to the proof of Claim~\ref{TheoremDownInMatchingPolytope}, we conclude that $\bar{S}$ must touch one of the cycles in precisely two nodes, whose connecting edge $e$ satisfies $\bar{x}_e > \tilde{x}_e$.
  Hence, (at least) one of the following four conditions must be satisfied:
  \begin{alignat*}{4}
    \text{(1) } & \bar{S} \cap \setdef{u_1,w_1,a,b} = \setdef{u_1,a}, & \qquad\text{(2) } & \bar{S} \cap \setdef{u_1,w_1,a,b} = \setdef{w_1,b}, \\
    \text{(3) } & \bar{S} \cap \setdef{u_2,w_2,a,b} = \setdef{u_2,b}, & \qquad\text{(4) } & \bar{S} \cap \setdef{u_2,w_2,a,b} = \setdef{w_2,a}.
  \end{alignat*}
  We define $\bar{V}^* \coloneqq \setdef{u_1,u_2,w_1,w_2,a,b}$ and $S \coloneqq \bar{S} \setminus \setdef{a,b}$.
  Note that we always have $|S| = |\bar{S}| -  1$ since each of the four conditions implies that either $a$ or $b$ is contained in $\bar{S}$, and hence $|S|$ is even.
  We now make a case distinction, based on $\bar{S} \cap \bar{V}^*$.
  All potential intersections $\bar{S} \cap \bar{V}^*$ arise from those above by adding a subset of the missing two elements, e.g., in (1) we have to consider adding any subset of $\setdef{u_2, w_2}$ to $\setdef{u_1,a}$.
  After elimination of two duplicates, this leads to 14 possible node sets, which we take care of in three cases.
  The cases arise by inspecting the modifications (from $\hat{x}$ to $\bar{x}$) that occur within $E[\bar{S}]$.

\medskip
  \textbf{Case 1}: $\bar{S} \cap \bar{V}^*$ is equal to
  $\setdef{u_1,a}$, $\setdef{u_1,a,u_2}$, $\setdef{u_1,a,u_2,w_2}$, $\setdef{w_1,b}$, $\setdef{w_1,b,w_2}$ or $\setdef{w_1,b,w_2,u_2}$.

  In this case $\bar{x}(E[\bar{S}]) = \hat{x}(E[S]) + \hat{x}_{e_1} - \frac{1}{2}\hat{y} \leq \frac{1}{2}|S| = \frac{1}{2}(|\bar{S}|-1)$
  by Proposition~\ref{TheoremUpInequalities}~\eqref{TheoremUpInequalitiesHalfUp}, which yields a contradiction.

\medskip
  \textbf{Case 2}: $\bar{S} \cap \bar{V}^*$ is equal to 
  $\setdef{u_2,b}$, $\setdef{u_2,b,u_1}$, $\setdef{u_2,b,u_1,w_1}$, $\setdef{w_2,a}$, $\setdef{w_2,a,w_1}$, $\setdef{w_2,a,w_1,u_1}$.

  In this case $\bar{x}(E[\bar{S}]) = \hat{x}(E[S]) + \hat{x}_{e_2} - \frac{1}{2}\hat{y} \leq \frac{1}{2}|S| = \frac{1}{2}(\bar{S}|-1)$
  by Proposition~\ref{TheoremUpInequalities}~\eqref{TheoremUpInequalitiesHalfUp}, which yields a contradiction.

\medskip
  \textbf{Case 3}: $\bar{S} \cap \bar{V}^*$ is equal to 
  $\setdef{u_1,a,w_2}$ or $\setdef{u_2,b,w_1}$ and $S \in \inequalitiesUp$.

  In this case
  $\bar{x}(E[\bar{S}]) = \hat{x}(E[S]) + \hat{x}_{e_1} + \hat{x}_{e_2} - \hat{y} \leq \frac{1}{2}|S| = \frac{1}{2}(|\bar{S}|-1)$
  by Proposition~\ref{TheoremUpInequalities}~\eqref{TheoremUpInequalitiesFullUp}, which yields a contradiction.
\end{proof}

\begin{proof}[Proof of Claim~\ref{TheoremUpCut}]
  Let $S^* \in \facetsUp$ be such that $(\hat{x},\hat{y})$ satisfies Inequality~\eqref{ConstraintMatchingQuadraticUp} with equality.
  If $u_1,w_2 \in S^*$, then we define $\bar{S} \coloneqq S^* \cup \setdef{a}$,
  and otherwise $\bar{S} \coloneqq S^* \cup \setdef{b}$.
  A simple calculation shows that $\bar{x}(E[\bar{S}]) = \hat{x}(E[S^*]) + \hat{x}_{e_1} + \hat{x}_{e_2} - \hat{y} = \frac{1}{2}|S^*| = \frac{1}{2}(|\bar{S}| - 1)$,
  i.e., $\bar{x}$ satisfies the Blossom Inequality induced by $\bar{S}$ with equality.
  Furthermore, $\bar{x}$ satisfies the degree inequalities \eqref{ConstraintMatchingDegree} for nodes $a$ and $b$
  with equality.
  This implies that, for all $i \in [k]$, the characteristic vector $\chi(\bar{M}_j)$ satisfies these three
  inequalities with equality, i.e.,
  we have $|\bar{M}_j \cap E[\bar{S}]| = \frac{1}{2}(|\bar{S}|-1)$ and $|\bar{M}_j \cap \delta(a)| = |\bar{M}_j \cap \delta(b)| = 1$.
  From the first equation we derive $|\bar{M}_j \cap \delta(\bar{S})| \leq 1$.
  This, together with the second equation proves the claimed properties.
\end{proof}

\begin{proof}[Proof of Claim~\ref{TheoremUpMinimalCombination}]
  First, the sets $J_u$, $J_w$, $J_1$, $J_2$ and $N$ are disjoint
  since the indexed matchings all match nodes $a$ and $b$ in different ways.
  Second, their union is equal to $[k]$ due to the second part of Claim~\ref{TheoremUpCut}.
  From this we obtain
  $|J_u| + |J_1| = \bar{x}_{\setdef{u_1,a}} k = \left(\hat{x}_{e_1} - \frac{1}{2}y\right)k = \bar{x}_{\setdef{w_1,b}} k = |J_w| + |J_1|$,
  and conclude that $|J_u| = |J_w|$.
  
  Now suppose, for the sake of contradiction, that $J_u \neq \emptyset$  
  (and thus $|J_w| = |J_u| \geq 1$).
  Let $j \in J_u$ and $j' \in J_w$ and let $C$ be the connected component (i.e., an alternating cycle or path) of $\bar{M}_j \Delta \bar{M}_{j'}$
  that contains $\setdef{u_2,b}$.

  We claim that $\setdef{u_1,a} \notin C$.
  Assuming the contrary, there must exist an odd-length (alternating) path in $K_{m,n}$ that connects either $u_1$ with $u_2$ or $w_1$ with $w_2$
  or there must exist an even-length (alternating) path in $K_{m,n}$ that connects either $u_1$ with $w_1$ or $u_2$ with $w_2$.
  Since $K_{m,n}$ is bipartite, none of these paths exist, which proves $\setdef{u_1,a} \notin C$.

  Define two new matchings $\bar{M}'_j \coloneqq \bar{M}_j \Delta C$ and $\bar{M}'_{j'} \coloneqq \bar{M}_{j'} \Delta C$,
  and note that $\chi(\bar{M}_j) + \chi(\bar{M}_{j'}) = \chi(\bar{M}'_j) + \chi(\bar{M}'_{j'})$,
  i.e., we can replace $\bar{M}_j$ and $\bar{M}_{j'}$ by $\bar{M}'_j$ and $\bar{M}'_{j'}$ in the convex combination.
  The fact that $\bar{M}'_j$ contains the edges $\setdef{u_1,a}$ and $\setdef{w_1,b}$
  and that $\bar{M}'_{j'}$ contains the edges $\setdef{w_2,a}$ and $\setdef{u_2,b}$ contradicts
  the assumption that the convex combination was chosen with minimum $|J_u| + |J_w|$.
  Hence, $J_u = J_w = \emptyset$.
\end{proof}

\begin{proof}[Proof of Claim~\ref{TheoremUpIndexSetContained}]
  Let $j \in J'_1$. 
  Using $\setdef{u_2,w_2} \in \bar{M}_j$, Claim~\ref{TheoremUpCut}
  shows that $\setdef{a,b} \notin \bar{M}_j$, and thus (since $u_2$ and $w_2$ are already matched to each other)
  that $\bar{M}_j$ contains the two edges $\setdef{u_1,a}$ and $\setdef{w_1,b}$, i.e., $j \in J_1$.
  The proof of $J'_2 \subseteq J_2$ is similar.

  From Claim~\ref{TheoremUpMinimalCombination} we have $J_1 \cap J_2 = \emptyset$, and hence $J'_1 \cap J'_2 = \emptyset$ holds as well.
\end{proof}

\begin{proof}[Proof of Claim~\ref{TheoremUpIndexSetFraction}]
  By Claim~\ref{TheoremUpIndexSetContained}, we have $J'_1 \cap J'_2 = \emptyset$.
  Hence, $|J'_1 \dcup J'_2| = |J'_1| + |J'_2| = k \bar{x}_{e_1} + k \bar{x}_{e_2} = k \cdot \frac{1}{2}\hat{y} + k \cdot \frac{1}{2} \hat{y} = k \hat{y}$, which concludes the proof.
\end{proof}

\begin{proof}[Proof of Claim~\ref{TheoremUpBarycenter}]
  Similar to the proof of Claim~\ref{TheoremDownBarycenter}, we
  consider the vector
  $d \coloneqq \sum_{j=1}^k (\chi(\tilde{M}_j) - \chi(\bar{M}_j))$.
  By construction of the $\tilde{M}_j$, we have
  \begin{itemize}
  \item
    $d_e = 0$ for all $e \notin C_1 \cup C_2$,
  \item
    $d_{e_1} = -d_{\setdef{u_1,a}} = -d_{\setdef{w_1,b}} = |J_1| = (\hat{x}_{e_1} - \frac{1}{2}\hat{y})k$,  
  \item
    $d_{e_2} = -d_{\setdef{u_2,b}} = -d_{\setdef{w_2,a}} = |J_2| = (\hat{x}_{e_2} - \frac{1}{2}\hat{y})k$, and
  \item
    $d_{\setdef{a,b}} = d_{e_1} + d_{e_2} = (\hat{x}_{e_1} + \hat{x}_{e_2} - \hat{y})k$.
  \end{itemize}
  A simple comparison with the construction of $\tilde{x}$ and $\bar{x}$ from $\hat{x}$
  proves the first part.

  The construction of $\hat{M}_j$ from $\tilde{M}_j$ by removing edge $\setdef{a,b}$
  corresponds to the fact that $\hat{x}$ is the orthogonal projection of $\tilde{x}$ onto $\R^E$,
  which proves the second part.
\end{proof}

\section{Facet proofs}
\label{SectionFacets}

We start by establishing the dimensions of the three polytopes and then consider all inequality classes regarding whether they induce facets.

\begin{proposition}
  The polytopes $\PmatchOne$, $\PmatchOneDown$ and $\PmatchOneUp$ are full-dimensional.
\end{proposition}

\begin{proof}
  The point $(\chi(\emptyset), 0)$, the points $(\chi(\setdef{e}),0)$ for all $e \in E$ and the point $(\chi(\setdef{e_1,e_2}),1)$ are  $|E|+2$ affinely independent points that are contained in all three polytopes.
  This proves the statement.
\end{proof}

\begin{proposition}
  Let $e^* \in E$.
  Then Inequality~\eqref{ConstraintMatchingNonnegative} defines a facet for $\PmatchOneUp$.
  Furthermore, it defines a facet for $\PmatchOne$ (and thus for $\PmatchOneDown$) if and only if $e^* \notin \setdef{e_1,e_2}$.
\end{proposition}

\begin{proof}
  If $e^* \notin \setdef{e_1,e_2}$, we consider the following set of $|E|+1$ points:
  \begin{gather*}
    (\chi(\emptyset),0),\quad
    (\chi(\setdef{e_1,e_2}),1),\quad
    (\chi(\setdef{e}),0) \text{ for all } e \in E \setminus \setdef{e^*}.
  \end{gather*}
  Since they are clearly affinely independent, satisfy $x_{e^*} \geq 0$ with equality,
  and are contained in all three polytopes,
  we obtain that Inequality~\eqref{ConstraintMatchingNonnegative}
  is facet-defining for each of them.
  Otherwise, consider $e^* = e_i$ for some $i \in \setdef{1,2}$.
  For $\PmatchOneUp$ we can replace $(\chi(\setdef{e_1,e_2}),1)$ by $(\chi(\emptyset),1)$
  to obtain the same result.
  For the other two polytopes, $x_{e^*} \geq 0$ is clearly implied by $0 \leq y$ and $y \leq x_{e_i}$,
  and hence not facet-defining.
\end{proof}

\begin{proposition}
  \label{TheoremFacetDegree}
  Let $v^* \in U \dcup W$
  and let $k \coloneqq |\delta(v^*)|$.
  Then Inequality~\eqref{ConstraintMatchingDegree} is facet-defining for
  \begin{itemize}
  \item
    \label{TheoremFacetDegreeMatcheOneUp}
    $\PmatchOneUp$ in any case, for
  \item
    \label{TheoremFacetDegreeMatcheOne}
    $\PmatchOneDown$ if and only if $k \geq 3$ or $v^* \in \Vspecial$, and for
  \item
    $\PmatchOne$ if and only if $k \geq 3$.
  \end{itemize}
\end{proposition}

\begin{proof}
  First note that $k = m$ or $k = n$, since we consider the complete bipartite graph, and thus $k \geq 2$.
  Now assume $k \geq 3$ and consider the points
  \begin{gather*}
    (\chi(\setdef{e}),0) \text{ for all } e \in \delta(v^*).
  \end{gather*}
  For each $e \in E \setminus \delta(v^*)$, let $f_e \in \delta(v^*) \setminus \setdef{e_1,e_2}$
  be an edge that is disjoint from $e$,
  which must exist since $v^*$ has degree at least $3$.
  Then consider the points
  \begin{gather*}
    (\chi(\setdef{e,f_e}),0) \text{ for all } e \in E \setminus \delta(v^*).
  \end{gather*}
  Due to $f_e \neq e_i$ for $i=1,2$ we have that all points are feasible.
  If $v^* \in \Vspecial$, then the two sets of points can be extended with
  $(\chi(\setdef{e_1,e_2}),1)$ to a set of $|E|+1$ affinely independent points in $\PmatchOne$
  that satisfy Inequality~\eqref{ConstraintMatchingDegree} for $v = v^*$ with equality.

  Otherwise, i.e., if $v^* \notin V^*$, let $\bar{e} \in \delta(v^*)$ be any edge disjoint from $\Vspecial$.
  Then the two sets of points can be extended with
  $(\chi(\setdef{e_1,e_2,\bar{e}}),1)$ to a set of $|E|+1$ affinely independent points in $\PmatchOne$
  that satisfy Inequality~\eqref{ConstraintMatchingDegree} for $v = v^*$ with equality.
  This proves the theorem for $k \geq 3$ for all three polytopes.

  Consider the case of $k = 2$ and $v^* \notin \Vspecial$.
  By symmetry, we can assume w.l.o.g.\ $v^* \in U$.
  Then Inequality~\eqref{ConstraintMatchingDegree} for $v = v^*$ is the sum
  of Inequality~\eqref{ConstraintMatchingQuadraticDown} for $S = \setdef{v^*,w_1,w_2}$ and $-y \leq 0$.
  Since both inequalities are valid for $\PmatchOne$ and $\PmatchOneDown$,
  Inequality~\eqref{ConstraintMatchingDegree} does not define a facet for these polytopes.
  To see that it is facet-defining for $\PmatchOneUp$, one can easily check that the points
  \begin{align*}
    &(\chi(\setdef{\setdef{v^*,w_1}}),0),\quad
    (\chi(\setdef{\setdef{v^*,w_2}}),0),\quad
    (\chi(\setdef{\setdef{v^*,w_1}}),1), \text{ and } \\
    &(\chi(\setdef{\setdef{v^*,w_i},e}),0) \text{ for all } e \in E \setminus \delta(v^*) 
      \text{ and } i \in \setdef{1,2} \text{ with $w_i \notin e$}
  \end{align*}
  are contained in $\PmatchOneUp$, are affinely independent and satisfy $x(\delta(v^*)) = 1$.

  It remains to consider the case of $k = 2$ and $v^* \in \Vspecial$.
  Again by symmetry we can assume w.l.o.g.\ $v^* = u_1$.
  Then Inequality~\eqref{ConstraintMatchingDegree} is the sum
  of Inequality~\eqref{ConstraintMatchingQuadraticUp} for $S = \setdef{u_1,w_2}$
  and Inequality~\eqref{ConstraintMatchingQuadraticGood} for $i = 2$.
  Both inequalities are valid for $\PmatchOne$,
  and hence Inequality~\eqref{ConstraintMatchingDegree} for $v = v^*$
  cannot be facet-defining for this polytope.
  To see that it is facet-defining for the other two polytopes, we consider the points
  \begin{align*}
    &(\chi(\setdef{e_1}),0),\quad
    (\chi(\setdef{\setdef{u_1,w_2}}),0),\quad
    (\chi(\setdef{e_1,e_2}),1),\text{ and } \\
    &(\chi(\setdef{\setdef{a_1,w_i},e}),0) \text{ for all } e \in E \setminus \setdef{e_1,e_2,\setdef{u_1,w_2}}
      \text{ and } i \in \setdef{1,2} \text{ with $w_i \notin e$}
  \end{align*}
  in $\PmatchOne$.
  On the one hand, they can be extended with
  $(\chi(\setdef{e_1,e_2}),0)$ to a set of $|E|+1$ affinely independent points in $\PmatchOneDown$
  that satisfy Inequality~\eqref{ConstraintMatchingDegree} for $v = v^*$ with equality.
  On the other hand, they can be extended with
  $(\chi(\setdef{e_1}),1)$ to a set of $|E|+1$ affinely independent points in $\PmatchOneUp$
  that satisfy Inequality~\eqref{ConstraintMatchingDegree} for $v = v^*$ with equality.
  This concludes the proof.
\end{proof}

\begin{proposition}
  \label{TheoremFacetBounds}
  The inequality $y \geq 0$ is facet-defining for $\PmatchOneUp$, $\PmatchOneDown$ and $\PmatchOne$,
  while $y \leq 1$ is facet-defining for $\PmatchOneUp$, but not for $\PmatchOneDown$ and $\PmatchOne$.
\end{proposition}

\begin{proof}
  For fixed value $k \in \setdef{0,1}$, the point $(\chi(\emptyset), k)$ and the points $(\chi(\setdef{e}), k)$ for all $e \in E$ are $|E|+1$ affinely independent points.
  For $k = 0$, they are contained in all three polytopes and satisfy $y \geq 0$ with equality, which proves the first statement.
  For $k = 1$, they are contained in $\PmatchOneUp$ and satisfy $y \leq 1$ with equality,
  which proves one direction of the second statement.
  For the reverse direction, observe that $y \leq 1$ is the sum of Inequality~\eqref{ConstraintMatchingQuadraticGood}
  for $i \in \setdef{1,2}$ and $x_{e_i} \leq 1$, which concludes the proof.
\end{proof}

\begin{proposition}
  For $i^*=1,2$,
  Inequalities~\eqref{ConstraintMatchingQuadraticGood}
  define facets for $\PmatchOne$ and $\PmatchOneDown$.
\end{proposition}

\begin{proof}
  Let $i^* \in \setdef{1,2}$.
  The points $(\chi(\emptyset), 0)$ and $(\chi(\setdef{e_1,e_2}),1)$ 
  and the points $(\chi(\setdef{e}),0)$ for all $e \in E \setminus \setdef{e_{i^*}}$
  are $|E|+1$ affinely independent points that are contained in both polytopes and satisfy $x_{i^*} \leq y$ with equality,
  which proves the statement.
\end{proof}

For the remaining two proofs we will consider a set $S^* \subseteq U \dcup W$ of nodes
and denote by $U^* \coloneqq S^* \cap U$ and $W^* \coloneqq S^* \cap W$ the induced sides of the bipartition.
For a matching $M$ in $K_{m,n}$ we denote by $y(M) \in \setdef{0,1}$ its corresponding $y$-value,
i.e., $y(M) = 1$ if and only if $e_1,e_2 \in M$.
Note that this implies $(\chi(M),y(M)) \in \PmatchOne$.
Another concept from matching theory also turns out to be useful:
We say that a matching is \emph{near-perfect} in a set of nodes if
it matches all nodes but one of this set.

\begin{proposition}
  For all $S^* \in \facetsDown$,
  Inequalities~\eqref{ConstraintMatchingQuadraticDown}
  define facets for $\PmatchOne$ and $\PmatchOneDown$.
\end{proposition}

\begin{proof}
  Let $S^* \in \facetsDown$.
  We will assume w.l.o.g.\ that $|U^*| = |W^*| + 1$ (i.e., $u_1,u_2 \in S^*$), since the proof for $|U^*| = |W^*| - 1$ is similar.
  Let $\mathcal{M}$ denote the set of matchings $M$ in $K_{m,n}$ that
  induce a near-perfect matching in $E[S^*]$ or
  induce a near-perfect matching in $E[S^* \setminus \Vspecial]$
  and contain edges $e_1$ and $e_2$.
  In the first case we have $|M \cap E[S^*]| = \frac{1}{2}(|S^*|-1)$ and $y(M) = 0$,
  and in the second case we have $|M \cap E[S^*]| = \frac{1}{2}(|S^*|-3)$ and $y(M) = 1$.
  Hence, for all $M \in \mathcal{M}$, the vector $(\chi(M),y(M))$ satisfies Inequality~\eqref{ConstraintMatchingQuadraticDown} with equality.

  Let $\scalprod{c}{x} + \gamma y \leq \delta$ dominate Inequality~\eqref{ConstraintMatchingQuadraticDown} for $S = S^*$,
  i.e., it is valid for $\PmatchOneDown$
  and for all $M \in \mathcal{M}$ we have $\scalprod{c}{\chi(M)} + \gamma y(M) = \delta$.
  We now analyze the coefficients and the right-hand side of the inequality.
  \begin{enumerate}
  \item[(i)]
    Let $e \in E \setminus E[S^*]$.
    If $e$ intersects $S^*$, then let $v \in e \cap S^*$ be its endnode in $S^*$, otherwise
    let $v \in S^*$ be arbitrary.
    If $v \in U^*$, then let $M_1$ be a perfect matching in $E[S^* \setminus \setdef{v}]$
    (which exists due to $|U^* \setminus \setdef{v}| = |W^*|$).
    Otherwise, let $M'$ be a perfect matching in $E[S^* \setminus \setdef{v,u_1,u_2}]$
    (which exists due to $|U^* \setminus \setdef{u_1,u_2}| = |W^* \setminus \setdef{v}|$),
    and extend it to the matching $M_1 \coloneqq M' \cup \setdef{e_1,e_2}$.
    Then $e$ does not intersect any edge of $M_1$ and thus $M_2 \coloneqq M_1 \dcup \setdef{e}$ is also a matching
    that satisfies $y(M_1) = y(M_2)$.
    By construction we have $M_1,M_2 \in \mathcal{M}$,
    and hence $\scalprod{c}{\chi(M_1)} + \gamma y(M_1) = \delta = \scalprod{c}{\chi(M_2)} + \gamma y(M_2)$.
    This proves $c_e = 0$.
  \item[(ii)]
    Let $u \in U^*$ and let $e = \setdef{u,v}$ and $f = \setdef{u,w}$ be two incident edges with 
    endnodes $v,w \in W^*$.
    Let $M_1$ be a perfect matching in $E[S^* \setminus \setdef{v}]$ that uses edge $f$.
    Then $M_2 \coloneqq (M_1 \setminus \setdef{f}) \dcup \setdef{e}$ is perfect in $E[S^* \setminus \setdef{w}]$.
    Clearly, $M_1,M_2 \in \mathcal{M}$ by construction,
    and we obtain $\scalprod{c}{\chi(M_1)} + \gamma y(M_1) = \delta = \scalprod{c}{\chi(M_2)} + \gamma y(M_2)$,
    i.e., $c_f = c_e$.
  \item[(iii)]
    If $|W^*| \geq 2$, then also $|U^*| \geq 3$.
    Let $v,w \in W^*$ be two nodes, let $u \in U^* \setminus \setdef{u_1,u_2}$,
    and let $e \coloneqq \setdef{u,v}$ and $f \coloneqq \setdef{u,w}$.
    Let $M'$ be a perfect matching in $E[S^* \setminus \setdef{u_1,u_2,u,v,w}]$ 
    (which exists due to $|U^* \setminus \setdef{u_1,u_2,u}| = |W^* \setminus \setdef{v,w}|$).
    Define matchings $M_1 \coloneqq M' \dcup \setdef{e,e_1,e_2}$ and $M_2 \coloneqq M' \dcup \setdef{f,e_1,e_2}$
    and observe that $M_1,M_2 \in \mathcal{M}$ and $y(M_1) = 1 = y(M_2)$.
    Thus, $\scalprod{c}{\chi(M_1)} + \gamma y(M_1) = \delta = \scalprod{c}{\chi(M_2)} + \gamma y(M_2)$,
    i.e., $c_e = c_f$.
  \item[(iv)]
    Let $M_1$ be a perfect matching in $E[S^* \setminus \setdef{u_1}]$ and
    let $e \in M_1$ be the edge that matches $u_2$.
    Define matching $M_2 \coloneqq (M_1 \setminus \setdef{e}) \cup \setdef{e_1,e_2}$,
    and note that $M_1,M_2 \in \mathcal{M}$, $y(M_1) = 0$ and $y(M_2) = 1$.
    By (i), we have $c_{e_1} = c_{e_2} = 0$,
    and using
    $\scalprod{c}{\chi(M_1)} + \gamma y(M_1) = \delta = \scalprod{c}{\chi(M_2)} + \gamma y(M_2)$,
    we obtain $c_e = \gamma$.
  \end{enumerate}
  The arguments above already fix $(c,\gamma)$ up to multiplication with a scalar.
  Hence we can assume that $\gamma = 1$, which proves that
  $(c,\gamma)$ is equal to the coefficient vector of Inequality~\eqref{ConstraintMatchingQuadraticDown} for $S = S^*$.
  Since there always exists a near-perfect matching $M$ in $E[S^*]$, and
  since such a matching has cardinality $|M| = \frac{1}{2}(|S^*|-1)$,
  we derive $\delta = \frac{1}{2}(|S^*|-1)$,
  which concludes the proof.
\end{proof}

\begin{proposition}
  For all $S^* \in \facetsUp$,
  Inequalities~\eqref{ConstraintMatchingQuadraticUp}
  define facets for $\PmatchOne$ and $\PmatchOneUp$.
\end{proposition}

\begin{proof}
  Let $S^* \in \facetsUp$.
  We will assume w.l.o.g.\ that $u_1,w_2 \in S^*$, since the proof for $u_2,w_1 \in S^*$ is similar.
  Let $\mathcal{M}$ denote the set of matchings $M$ in $K_{m,n}$ that
  either induce a perfect matching in $E[S^*]$
  or contain exactly one edge $e \in \setdef{e_1,e_2}$ and induce a near-perfect matching in $E[S^* \setminus e]$
  or contain both, $e_1$ and $e_2$, and induce a perfect matching in $E[S^* \setminus \setdef{u_1,w_2}]$.
  In the first two cases we have $|M \cap (E[S^*] \cup \setdef{e_1,e_2})| = \frac{1}{2}|S^*|$ and $y(M) = 0$,
  and in the third case we have $|M \cap (E[S^*] \cup \setdef{e_1,e_2})| = \frac{1}{2}(|S^*|) + 1$ and $y(M) = 1$.
  Hence, for all $M \in \mathcal{M}$, the vector $(\chi(M),y(M))$
  satisfies Inequality~\eqref{ConstraintMatchingQuadraticUp} with equality.
  Let $\scalprod{c}{x} + \gamma y \leq \delta$ dominate Inequality~\eqref{ConstraintMatchingQuadraticUp} for $S = S^*$,
  i.e., it is valid for $\PmatchOneUp$ 
  and for all $M \in \mathcal{M}$ we have $\scalprod{c}{\chi(M)} + \gamma y(M) = \delta$.
  We now analyze the coefficients and the right-hand side of the inequality.
  \begin{itemize}
  \item
    Let $e \in E \setminus (E[S^*] \cup \setdef{e_1,e_2})$.
    If $e$ intersects $S^*$, then let $v \in e \cap S^*$ be its endnode in $S^*$, otherwise
    let $v \in S^*$ be arbitrary.
    If $v \in U^*$, then let $M'$ be a perfect matching in $E[S^* \setminus \setdef{v,w_2}]$
    (which exists due to $|U^* \setminus \setdef{v}| = |W^* \setminus \setdef{w_2}|$),
    and extend it to the matching $M_1 \coloneqq M' \cup \setdef{e_2}$.
    Otherwise, let $M'$ be a perfect matching in $E[S^* \setminus \setdef{v,u_1}]$
    (which exists due to $|U^* \setminus \setdef{u_1}| = |W^* \setminus \setdef{v}|$),
    and extend it to the matching $M_1 \coloneqq M' \cup \setdef{e_1}$.
    Then $e$ does not intersect any edge of $M_1$ and thus $M_2 \coloneqq M_1 \dcup \setdef{e}$ is also a matching
    that satisfies $y(M_1) = 0 = y(M_2)$.
    By construction we have $M_1,M_2 \in \mathcal{M}$,
    and hence $\scalprod{c}{\chi(M_1)} + \gamma y(M_1) = \delta = \scalprod{c}{\chi(M_2)} + \gamma y(M_2)$.
    This proves $c_e = 0$.
  \item
    Let $u \in S^* \setminus \setdef{u_1,w_2}$ and let $e = \setdef{u,v}$ and $f = \setdef{u,w}$ be two incident edges with 
    endnodes $v,w \in S^*$.
    W.l.o.g.\ we can assume $u \in U^*$, since the case of $u \in W^*$ is similar.
    Let $M'$ be a perfect matching in $E[S^* \setminus \setdef{u_1,u,v,w}]$
    (which exists due to $|U^* \setminus \setdef{u_1,u}| = |W^* \setminus \setdef{v,w}|$).
    Define the two matchings $M_1 \coloneqq M' \dcup \setdef{e_1,e}$ and $M_2 \coloneqq M' \dcup \setdef{e_1,f}$,
    and observe that $M_1,M_2 \in \mathcal{M}$ and $y(M_1) = 0 = y(M_2)$.
    From $\scalprod{c}{\chi(M_1)} + \gamma y(M_1) = \delta = \scalprod{c}{\chi(M_2)} + \gamma y(M_2)$
    we obtain that $c_e = c_f$.
  \item
    Let $M'$ be a perfect matching in $E[S^* \setminus \setdef{u_1,w_2}]$.
    Define the matchings
    $M_1 \coloneqq M' \dcup \setdef{\setdef{u_1,w_2}}$,
    $M_2 \coloneqq M' \dcup \setdef{e_1}$,
    $M_3 \coloneqq M' \dcup \setdef{e_{2}}$ and
    $M_4 \coloneqq M' \dcup \setdef{e_1,e_2}$.
    By construction we have $M_1,M_2,M_3,M_4 \in \mathcal{M}$, $y(M_1) = y(M_2) = y(M_3) = 0$ and $y(M_4) = 1$.
    Thus, $\scalprod{c}{\chi(M_i)} + \gamma y(M_i) = \delta$ for $i=1,2,3,4$,
    which proves
    $c_{\setdef{u_1,w_2}} = c_{e_1} = c_{e_2} = c_{e_1} + c_{e_2} - \gamma$.
  \end{itemize}
  The arguments above already fix $(c,\gamma)$ up to multiplication with a scalar.
  Hence we can assume $\gamma = 1$, which proves that
  $(c,\gamma)$ is equal to the coefficient vector of Inequality~\eqref{ConstraintMatchingQuadraticUp} for $S = S^*$.
  Since there always exists a perfect matching $M$ in $E[S^*]$, and
  since such a matching has cardinality $|M| = \frac{1}{2}|S^*|$,
  we derive $\delta = \frac{1}{2}|S^*|$.
  This concludes the proof.
\end{proof}

\section{Separation problems}
\label{SectionSeparation}

By the polynomial-time equivalence of separation and optimization~\cite{GroetschelLS81,KarpP80,PadbergR81},
using the fact that we can optimize over the polytopes in polynomial time,
it is evident that the separation problems for the three polytope families can be
solved in polynomial time.
Furthermore, Klein~\cite{Klein14} presents separation algorithms for Constraints~\eqref{ConstraintMatchingQuadraticDown} and~\eqref{ConstraintMatchingQuadraticUp}
in the context of perfect matchings.
A closer look into the proofs reveals that the correctness of these algorithms only requires the ``$\leq$''-part of Equations~\eqref{ConstraintPerfectMatchingDegree}, i.e., the algorithms work for arbitrary matchings as well.
In fact, they require that, for each of the inequality classes, a separation algorithm (such as the famous Padberg-Rao algorithm~\cite{PadbergR82}) for the Blossom Inequalities has to be run in two (symmetric) auxiliary graphs.

In view of this fact it is desirable to find separation algorithms that require only a single execution of such a separation routine per inequality class.
Fortunately, it turns out that the constructions from Sections~\ref{SectionDown} and~\ref{SectionUp} are in fact
reductions of the respective separation problems to the separation problem for Blossom Inequalities
in the respective auxiliary graphs, and hence have this desirable property.
In the remainder of this section we present the details of this observation.

\begin{proposition}
  The separation problem for $\PmatchOneDown$ can be solved in polynomial time.
\end{proposition}

\begin{proof}
  Let $(\hat{x},\hat{y}) \in \R^E \times \R$.
  We first check directly whether one of the Constraints~\eqref{ConstraintMatchingNonnegative}, \eqref{ConstraintMatchingDegree}, \eqref{ConstraintMatchingBound}
  or \eqref{ConstraintMatchingQuadraticGood} is violated and return a violated inequality if one exists.
  It remains to find a violated Inequality~\eqref{ConstraintMatchingQuadraticDown} if possible.

  To this end, construct $\bar{G}$ and $\bar{x}$ as in Section~\ref{SectionDown}.
  On the one hand, if $(\hat{x},\hat{y})$ violates Inequality~\eqref{ConstraintMatchingQuadraticDown} for some $S \in \facetsDown$, then $\bar{x}$ violates the corresponding Blossom Inequality in the auxiliary graph $\bar{G}$ (defined in Section~\ref{SectionDown}).
  On the other hand, if $(\hat{x},\hat{y}) \in \PmatchOneDown$, then $\bar{x}$ is in $\bar{G}$'s matching polytope, as proved in Claim~\ref{TheoremDownInMatchingPolytope}.
  Hence, in order to find a set $S \in \facetsDown$ that induces a violated Inequality~\eqref{ConstraintMatchingQuadraticDown} (if such a set exists) we just have to run the separation algorithm for the Blossom Inequalities in the graph $\bar{G}$ from Section~\ref{SectionDown} with respect to $\bar{x}$.
\end{proof}

Note that the proof above implies that the matching polytope for $\bar{G}$, intersected with the hyperplane defined by $x_{e_a} = x_{e_b}$,
is an extended formulation for $\PmatchOneDown$. 
In fact, the two polytopes must even be affinely isomorphic for dimension reasons.
This is also justified by the fact that for the proof of Lemma~\ref{TheoremDownCombination} we only used
the tightness of an Inequality~\eqref{ConstraintMatchingQuadraticDown} to control our matchings,
but not for the proof that $\bar{x}$ is in $\bar{G}$'s matching polytope.

This is different for the upward monotonization, for which we were not able
to identify a direct relation between $\PmatchOneUp$ and the matching polytope of the auxiliary graph.

\begin{proposition}
  The separation problem for $\PmatchOneUp$ can be solved in polynomial time.
\end{proposition}

\begin{proof}
  Let $(\hat{x},\hat{y}) \in \R^E \times \R$.
  We first check directly whether one of the Constraints~\eqref{ConstraintMatchingNonnegative}, \eqref{ConstraintMatchingDegree}, \eqref{ConstraintMatchingBound} or \eqref{ConstraintMatchingQuadraticUp} for $S = \setdef{a_1,b_2}$ is violated and return a violated inequality if one exists.
  It remains to find a violated Inequality~\eqref{ConstraintMatchingQuadraticUp} (for $S \neq \setdef{a_1,b_2}$) if possible.

  To this end, construct $\bar{x}$ as in Section~\ref{SectionUp}.
  Since we checked the single Inequality~\eqref{ConstraintMatchingQuadraticUp} beforehand, the requirements for
  Proposition~\ref{TheoremUpInequalities} are satisfied.
  If $\bar{x}$ contains a negative entry, the contrapositive of Proposition~\ref{TheoremUpInequalities}~\LabelTheoremUpInequalitiesGood{} implies that all Inequalities~\eqref{ConstraintMatchingQuadraticUp} must be satisfied.
  Otherwise, the proof of Claim~\ref{TheoremUpInMatchingPolytope} immediately shows that $\bar{x}$ is in the matching polytope of $\bar{G}$ if and only if $(\hat{x},\hat{y})$ is in $\PmatchOneUp$.
  Furthermore, Case~3 in the proof shows a one-to-one correspondence  Inequalities~\eqref{ConstraintMatchingQuadraticUp} for $\PmatchOneUp$ and the Blossom Inequalities for $\bar{G}$'s matching polytope.
  Hence, in order to find a set $S \in \facetsUp$ that induces a violated Inequality~\eqref{ConstraintMatchingQuadraticUp} (if such a set exists) we just have to run the separation algorithm for the Blossom Inequalities in the graph $\bar{G}$ from Section~\ref{SectionUp} with respect to $\bar{x}$.
\end{proof}

\section{Generalization to capacitated b-matchings}
\label{SectionBMatchings}

\DeclareDocumentCommand\PbmatchOne{o}{P_{\text{\IfValueTF{#1}{\ensuremath{#1}}{$b$}-match}}^{\text{1Q}}}
\DeclareDocumentCommand\bCut{}{\mathcal{S}}

Consider again the complete bipartite graph $K_{m,n}$ with node sets $U$ and $W$, and edge set $E$.
For a vector $b \in \Z_+^{U \dcup W}$, a vector $x \in \Z_+^E$ that satisfies $x(\delta(v)) \leq b_v$ for each $v \in U \dcup W$ is called a \emph{$b$-matching}.
The goal of this section is to extend the polyhedral results of Section~\ref{SectionResults},
first to uncapacitated $b$-matchings and then to capacitated $b$-matchings, i.e., $b$-matchings that satisfy a capacity constraint $x_e \leq c_e$ for some vector $c \in \Z_+^E$ for every edge.

A special case of this problem is the one with $c = \onevec[E]$ and $b = 2 \cdot \onevec[V]$, where $\onevec$ denotes the all-ones vector. Here, feasible solutions correspond to sets of node-disjoint cycles.
Thus, our results will yield a polyhedral description for the \emph{cycle cover problem} with one linearized quadratic term for bipartite graphs.
This may be used to model the quadratic cycle cover problem in which costs also depend on two subsequent edges (see~\cite{GalbiatiGM14}).
In fact, since Hamiltonian paths are cycles as well, it may also be used for the quadratic TSP problem~\cite{Fischer13,FischerFJKMG14}, although bipartite graphs play no major role for this problem.

The overall proof strategy is common to both extensions, and hence we summarize it here.
We will start by proving that a certain set of inequalities is valid.
Then we will write the (integral) polytope $P$ in question as a projection of another (integral) polytope $Q$ of which we know the description in terms of inequalities.
We then consider an arbitrary point $x$ that satisfies the inequalities, and prove that there exists a pre-image (with respect to the projection map) $\bar{x} \in Q$.
This suffices since then $\bar{x}$ is a convex combination of vertices of $Q$, and thus $x$ is a convex combination of the projected vertices, i.e., $x \in P$.

\subsection{Uncapacitated b-matchings}

We start by generalizing the polyhedral results of Section~\ref{SectionResults} to $b$-matchings.
In order to linearize a product of two \emph{binary} variables, we assume that $b_v = 1$ holds for all nodes $v \in \Vspecial$.
Note that the variables are already binary if one endnode of every edge has this property, but we will be able to handle this more general case as soon as we introduce capacities.
We consider the polytope
\begin{multline*}
  \PbmatchOne \coloneqq \conv\{ (x,y) \in \Z_+^E \times \setdef{0,1} : x(\delta(v)) \leq b_v \text{ for all $v \in U \dcup W$} \\
    \text{ and $y = 1$ if and only if $x_{e_1} = x_{e_2} = 1$ }  \}.
\end{multline*}
Clearly, the variable bounds~\eqref{ConstraintMatchingNonnegative} and~\eqref{ConstraintMatchingBound} as well as the
inequalities
\begin{align*}
  y \leq x_{e_i} \qquad \text{ for $i=1,2$}, \tag{\ref{ConstraintMatchingQuadraticGood}}
\end{align*}
and the generalized degree constraints
\begin{alignat}{6}
  x(\delta(v))  & \leq b_v        &\qquad& \text{for all } v \in U \dcup W \label{ConstraintbMatchingDegree}
\end{alignat}
are valid for $\PbmatchOne$.
Using the notation $\bCut \coloneqq \setdef{ S \subseteq U \dcup W }[ \text{$e_1,e_2 \in \delta(S)$} ]$,
we can state the generalizations of Constraints~\eqref{ConstraintMatchingQuadraticDown} and~\eqref{ConstraintMatchingQuadraticUp} as
\begin{align}
  x(E[S]) + y                         &\leq \left\lfloor \tfrac{1}{2}b(S) \right\rfloor \nonumber \\
                                      & \text{for all } S \in \bCut \text{ with } b(S) \text{ odd, and} \label{ConstraintbMatchingQuadraticDown} \\
  x(E[S]) + x_{e_1} + x_{e_2} - y     &\leq \left\lfloor \tfrac{1}{2}(b(S)+1) \right\rfloor \nonumber \\
                                      & \text{for all } S \in \bCut \text{ with } b(S) \text{ even.} \label{ConstraintbMatchingQuadraticUp}
\end{align}
Note that due to the parity conditions on $b(S)$ we could make the right-hand sides more explicitly,
e.g., by replacing $\lfloor \tfrac{1}{2}(b(S)+1) \rfloor$ by $\tfrac{1}{2}b(S)$.
We still prefer the slightly more complicated form since we will soon observe that the inequalities (the way they are stated) remain valid if the parity of $b(S)$ is different.

Our main result for $b$-matchings is then the following:
\begin{theorem}
  \label{TheoremUncapacitatedbMatching}
  For $b \in \Z_+^V$ with $b_v = 1$ for all $v \in \Vspecial$, $\PbmatchOne$ is equal to the set of $(x,y) \in \R^E \times \R$ that satisfy Constraints~\eqref{ConstraintMatchingNonnegative}, \eqref{ConstraintMatchingBound}, \eqref{ConstraintMatchingQuadraticGood}, \eqref{ConstraintbMatchingDegree}, \eqref{ConstraintbMatchingQuadraticDown}, and~\eqref{ConstraintbMatchingQuadraticUp}.
\end{theorem}

Our completeness proof is a modification of a completeness proof for the $b$-matching polytope on non-bipartite graphs, as presented in Schrijver's book (see Theorem~31.2 in~\cite{Schrijver03}), which in turn is based on a construction by Tutte~\cite{Tutte54}.

\begin{proof}
   The proof is structured as follows.
   We first show validity of the inequalities and describe the construction of an extended formulation based on an auxiliary graph.
   To establish the completeness of our proposed inequality description of $\PbmatchOne$ we will then show that any point that satisfies the proposed inequalities can be lifted to a point in the extended formulation.

   \begin{claim}
      \label{TheoremUncapacitatedbMatchingValidity}
      Inequalities~\eqref{ConstraintbMatchingQuadraticDown} and~\eqref{ConstraintbMatchingQuadraticUp} are valid for $\PbmatchOne$ for arbitrary sets $S \subseteq U \dcup W$ with $e_1, e_2 \in \delta(S)$, regardless of the parity of $b(S)$.
   \end{claim}

   \begin{proof}[Proof of Claim~\ref{TheoremUncapacitatedbMatchingValidity}]
      Consider an integer vector $(x,y) \in \PbmatchOne$.
      We have
      \begin{align*}
         x(E[S]) + y
         \underset{\text{\eqref{ConstraintMatchingQuadraticGood}}}{\leq} \tfrac{1}{2} ( 2x(E[S]) + x_{e_1} + x_{e_2})
         \leq \tfrac{1}{2} \sum_{v \in S} x(\delta(v))
         \leq \tfrac{1}{2} b(S),
      \end{align*}
      where the second inequality holds since every edge whose $x$-variable appears once (resp.\ twice) has one (resp.\ both) endnode(s) in $S$, and the third inequality by the definition of $b$-matchings.
      Inequality~\eqref{ConstraintbMatchingQuadraticDown} now follows, since the left-hand side of the formula is integral, allowing us to round the right-hand side down.
      The fact that we added several valid inequalities shows that the inequality is redundant for $S$ with even $b(S)$, since rounding has no effect in this case.
      Similarly, we obtain
      \begin{align*}
         x(E[S]) + x_{e_1} + x_{e_2} - y
         &= \tfrac{1}{2}( 2x(E[S]) + x_{e_1} + x_{e_2} ) + \tfrac{1}{2} (x_{e_1} + x_{e_2} - y) - \tfrac{1}{2}y \\
         &\leq \tfrac{1}{2} \sum_{v \in S} x(\delta(v))  + \tfrac{1}{2} (x_{e_1} + x_{e_2} - y) - \tfrac{1}{2}y \\
         &\leq \tfrac{1}{2} \sum_{v \in S} x(\delta(v)) + \tfrac{1}{2} - \tfrac{1}{2}y
         \leq \tfrac{1}{2} (b(S) + 1) - \tfrac{1}{2}y
         \underset{\text{\eqref{ConstraintMatchingBound}}}{\leq} \tfrac{1}{2} (b(S) + 1),
      \end{align*}
      where the first inequality holds since every edge whose $x$-variable appears in the first summand once (resp.\ twice) has one (resp.\ both) endnode(s) in $S$, the second due to $x_{e_1} \cdot x_{e_2} = y$, and the third by the definition of $b$-matchings.
      Inequality~\eqref{ConstraintbMatchingQuadraticUp} now follows, since the left-hand side of the formula is integral, allowing us to round the right-hand side down.
      Again, the fact that we added several valid inequalities shows that the inequality is redundant for $S$ with odd $b(S)$, since rounding has no effect in this case.
      This concludes the proof of the claim.
   \end{proof}

   We now continue with the completeness of the formulation.
   The theorem holds for $b = \onevec$ by Theorem~\ref{TheoremComplete}, since Constraints~\eqref{ConstraintbMatchingQuadraticDown} and~\eqref{ConstraintbMatchingQuadraticUp} imply Constraints~\eqref{ConstraintMatchingQuadraticDown} and~\eqref{ConstraintMatchingQuadraticUp} in this case.

\medskip

  \paragraph{Extended formulation and auxiliary graph.}
  Consider the graph $\bar{G} = (\bar{U} \dcup \bar{W}, \bar{E})$ obtained from $K_{m,n}$ by splitting each node $v \in U \dcup W$ in $b_v$ copies (denoted by the set $B_v \subseteq \bar{U} \dcup \bar{W}$).
  By the assumption $b_v = 1$ for all $v \in \Vspecial$, the nodes $u_1$, $u_2$, $w_1$ and $w_2$ are not split, and we call their representatives $\bar{u}_1$, $\bar{u}_2$, $\bar{w}_1$ and $\bar{w}_2$, respectively.
  Similarly, we denote by $\bar{e}_i \coloneqq \setdef{\bar{u}_i, \bar{w}_i}$ for $i=1,2$ the representatives of the edges $e_1$ and $e_2$.
  We will now consider the polytope $Q \coloneqq \PmatchOne(\bar{G},\bar{e}_1,\bar{e}_2)$ and the projection map $\pi \colon \R^{\bar{E}} \times \R \to \R^E \times \R$ defined via
  \[
    \pi((\bar{x},\bar{y})) \coloneqq \left( x , \bar{y} \right)
    \text{ with }
    x_{\setdef{u,w}} \coloneqq \sum_{\bar{u} \in B_u} \sum_{\bar{w} \in B_w} \bar{x}_{\setdef{\bar{u},\bar{w}}}
    \text{ for all } \setdef{u,w} \in E.
  \]
  It is easy to see that $\pi(Q) = \PbmatchOne$.

  Note that in Section~\ref{SectionResults} we provide a complete description for $\PmatchOne$ only for complete graphs, but $\PmatchOne$ for any subgraph is obtained by fixing variables to $0$, i.e., it is a face.

  Let $(x,y) \in \R_+^E \times [0,1]$ satisfy all constraints from the theorem.
  For each edge $\bar{e} = \setdef{\bar{u},\bar{w}} \in \bar{E}$ with $\bar{u}$ and $\bar{w}$ copies of $u \in U$ and $w \in W$, define $\bar{x}_{\bar{e}} \coloneqq x_e / (b_u \cdot b_w)$, where $e \coloneqq \setdef{u,w} \in E$.
  By letting $\bar{y} \coloneqq y$, it is evident that $\pi((\bar{x},\bar{y})) = (x,y)$, and it remains to show $(\bar{x},\bar{y}) \in Q$.

\medskip

  By construction, using the fact that $b_v = 1$ for all $v \in \Vspecial$, we have that $(\bar{x},\bar{y})$ satisfies Constraints~\eqref{ConstraintMatchingNonnegative}, \eqref{ConstraintMatchingBound} and~\eqref{ConstraintMatchingQuadraticGood}.
  Inequalities~\eqref{ConstraintbMatchingDegree}, \eqref{ConstraintbMatchingQuadraticDown} and~\eqref{ConstraintbMatchingQuadraticUp} are discussed in subsequent claims.

   \begin{claim}
      \label{TheoremUncapacitatedbMatchingDegree}
      The vector $(\bar{x},\bar{y})$ satisfies Constraint~\eqref{ConstraintbMatchingDegree} for $\bar{G}$ with respect to $\bar{b} \coloneqq \onevec[\bar{U} \dcup \bar{W}]$.
   \end{claim}

   \begin{proof}[Proof of Claim~\ref{TheoremUncapacitatedbMatchingDegree}]
      Let $\bar{v} \in B_v$ for some $v \in U \dcup W$.
      Then
      \begin{align}
         \bar{x}(\delta_{\bar{G}}(\bar{v})) 
         = \sum_{ \setdef{v,v'} \in \delta(v) } \sum_{\bar{v}' \in B_{v'}} \bar{x}_{\setdef{\bar{v}, \bar{v}' }}
         = \sum_{ \setdef{v,v'} \in \delta(v) } \sum_{\bar{v}' \in B_{v'}} x_{\setdef{v,v'}} / (b_v \cdot b_{v'})
         = \sum_{ \setdef{v,v'} \in \delta(v) } x_{\setdef{v,v'}} / b_v
         \leq 1
         \label{InequalitybMatchingProofLiftedDegree}
      \end{align}
      holds by Inequality~\eqref{ConstraintbMatchingDegree} for node $v$, which concludes the proof of the claim.
   \end{proof}

   \begin{claim}
      \label{TheoremUncapacitatedbMatchingDown}
      The vector $(\bar{x},\bar{y})$ satisfies Constraint~\eqref{ConstraintbMatchingQuadraticDown} for $\bar{G}$ with respect to $\bar{b} \coloneqq \onevec[\bar{U} \dcup \bar{W}]$.
   \end{claim}

   \begin{proof}[Proof of Claim~\ref{TheoremUncapacitatedbMatchingDown}]
      For the sake of contradiction, consider some $\bar{S} \subseteq \bar{U} \dcup \bar{W}$ with $\bar{e}_1,\bar{e}_2 \in \delta(\bar{S})$ such that $\bar{x}(\bar{E}[\bar{S}]) + \bar{y} > \lfloor \tfrac{1}{2} |\bar{S}| \rfloor $ and, among all such sets, with the minimum number of nodes $v \in U \dcup W$ for which $0 < |\bar{S} \cap B_v| < b_v$ holds.

      This number must be positive, since otherwise Constraint~\eqref{ConstraintbMatchingQuadraticDown} for $S = \setdef{ v \in U \dcup W }[ B_v \subseteq \bar{S} ]$ yields the contradiction
      \begin{gather*}
         \bar{x}(\bar{E}[\bar{S}]) + \bar{y} 
         = x(E[S]) + y 
         \underset{\text{\eqref{ConstraintbMatchingQuadraticDown}}}{\leq} \left\lfloor \tfrac{1}{2}b(S) \right\rfloor
         = \left\lfloor \tfrac{1}{2} |\bar{S}| \right\rfloor
         < \bar{x}(\bar{E}[\bar{S}]) + \bar{y},
      \end{gather*}
      where the last equation holds due to $\bar{S} = \bigcup_{v \in S} B_v$.

\medskip

      Hence, there exists a node $v \in U \dcup W$ with $0 < |\bar{S} \cap B_v| < b_v$.
      Let $\bar{S}_1 \coloneqq \bar{S} \setminus B_v$ and $\bar{S}_2 \coloneqq \bar{S} \cup B_v$.
      Note that we have $v \notin \Vspecial$, which implies $\bar{e}_1,\bar{e}_2 \in \delta(\bar{S}_i)$ for $i=1,2$.
      Since $v$ does not satisfy $0 < |\bar{S}_i \cap B_v| < b_v$ for $i=1,2$, the choice of $\bar{S}$ implies that Constraint~\eqref{ConstraintbMatchingQuadraticDown} is satisfied for $\bar{S}_1$ and for $\bar{S}_2$, i.e.,
      \begin{align}
         \bar{x}(\bar{E}[\bar{S}_1]) + \bar{y} \leq \left\lfloor \tfrac{1}{2} |\bar{S}_1| \right\rfloor
         \text{ and }
         \bar{x}(\bar{E}[\bar{S}_2]) + \bar{y} \leq \left\lfloor \tfrac{1}{2} |\bar{S}_2| \right\rfloor.
         \label{InequalitybMatchingProofDownInduction}
      \end{align}
      Moreover, we have
      \begin{align}
         \bar{x}(\bar{E}[\bar{S}_1]) + \bar{x}(\bar{E}[\bar{S}_2]) + 2\bar{y}
         \underset{\text{\eqref{ConstraintMatchingQuadraticGood}}}{\leq} \bar{x}(\bar{E}[\bar{S}_1]) + \bar{x}(\bar{E}[\bar{S}_2]) + \bar{x}_{\bar{e}_1} + \bar{x}_{\bar{e}_2}
         \leq \sum_{\bar{v} \in \bar{S}_1} \bar{x}(\delta_{\bar{G}}(\bar{v}))
         \underset{\text{\eqref{InequalitybMatchingProofLiftedDegree}}}{\leq} |\bar{S}_1|,
         \label{InequalitybMatchingProofDownSum}
      \end{align}
      where the second inequality holds since every edge whose $\bar{x}$-variable appears once (resp.\ twice) has at least one (resp.\ both) endnode(s) in $\bar{S}_1$.
      We will exploit this relation below.

      Now the multipliers $\lambda \coloneqq |B_v \cap \bar{S}| / b_v$ and $\mu \coloneqq |B_v \setminus \bar{S}| / b_v$ are nonnegative and satisfy $\lambda + \mu = 1$.
      Moreover, $\bar{E}[\bar{S}_1] \subseteq \bar{E}[\bar{S}_2]$, $\bar{E}[\bar{S}_2] \setminus \bar{E}[\bar{S}_1] \subseteq \delta_{\bar{G}}(B_v)$ and the fact that $\bar{x}$ is constant over all edges whose endnodes are copies of the same pair of original nodes\footnote{formally, $\bar{x}_{\bar{v},\bar{v}'} = \bar{x}_{\bar{v},\bar{v}''}$ for all $\bar{v}',\bar{v}'' \in B_v$ for some $v \in V$} imply
      \begin{align}
         \lambda \bar{x}(\bar{E}[\bar{S}_2]) + \mu \bar{x}(\bar{E}[\bar{S}_1])
         &= (\underbrace{\lambda + \mu}_{=\,1}) \bar{x}(\bar{E}[\bar{S}_1]) + \underbrace{\lambda \bar{x}(\bar{E}[\bar{S}_2] \setminus \bar{E}[\bar{S}_1])}_{=\,\bar{x}(\bar{E}[\bar{S}] \setminus \bar{E}[\bar{S}_1])} = \bar{x}(\bar{E}[\bar{S}]).
         \label{InequalitybMatchingProofDownConvex}
      \end{align}
      If $\mu - \lambda \geq 0$, we obtain
      \begin{align*}
         \bar{x}(\bar{E}[\bar{S}]) + \bar{y}
         &\underset{\text{\eqref{InequalitybMatchingProofDownConvex}}}{=} \mu \bar{x}(\bar{E}[S_1]) + \lambda \bar{x}(\bar{E}[S_2]) + \bar{y}
         = (\mu - \lambda) (\bar{x}(\bar{E}[\bar{S}_1]) + \bar{y}) + \lambda ( \bar{x}(\bar{E}[\bar{S}_1]) + \bar{x}(\bar{E}[\bar{S}_2]) + 2\bar{y} ) \\
         &\leq \tfrac{1}{2} (\mu - \lambda) |\bar{S}_1| + \lambda |\bar{S}_1|
         = \tfrac{1}{2} |\bar{S}_1|
         \leq \lfloor \tfrac{1}{2} |\bar{S}| \rfloor,
      \end{align*}
      where the second equation holds due to $(\mu - \lambda) + 2\lambda = 1$, the first inequality by~\eqref{InequalitybMatchingProofDownInduction} together with $\mu - \lambda \geq 0$ and~\eqref{InequalitybMatchingProofDownSum} together with $\lambda \geq 0$, the third equation by $\mu + \lambda = 1$, and the last inequality by $\bar{S}_1 \subsetneqq \bar{S}$.

      Otherwise, i.e., if $\lambda - \mu \geq 0$, we obtain
      \begin{align*}
         \bar{x}(\bar{E}[\bar{S}]) + \bar{y}
         &\underset{\text{\eqref{InequalitybMatchingProofDownConvex}}}{=} \lambda \bar{x}(\bar{E}[S_2]) + \mu \bar{x}(\bar{E}[S_1]) + \bar{y}
         = (\lambda - \mu) (\bar{x}(\bar{E}[\bar{S}_2]) + \bar{y}) + \mu ( \bar{x}(\bar{E}[\bar{S}_1]) + \bar{x}(\bar{E}[\bar{S}_2]) + 2\bar{y} ) \\
         &\leq \tfrac{1}{2} (\lambda - \mu)  |\bar{S}_2| + \mu |\bar{S}_1|
         = \tfrac{1}{2} |\bar{S}_1| + \tfrac{1}{2} (\lambda - \mu) |\bar{S}_2 \setminus \bar{S}_1|
         = \tfrac{1}{2} |\bar{S}_1| + \tfrac{1}{2} (\lambda - \mu) b_v \\
         &= \tfrac{1}{2} |\bar{S}_1| + \tfrac{1}{2} (|B_v \cap \bar{S}| - |B_v \setminus \bar{S}|)
         \leq \tfrac{1}{2} |\bar{S}_1| + \tfrac{1}{2} ( |B_v \cap \bar{S}| - 1 )
         \leq \lfloor \tfrac{1}{2} |\bar{S}| \rfloor,
      \end{align*}
      where the second equation holds due to $(\lambda - \mu) + 2\mu = 1$, the first inequality by~\eqref{InequalitybMatchingProofDownInduction} together with $\lambda - \mu \geq 0$ and~\eqref{InequalitybMatchingProofDownSum} together with $\mu \geq 0$, the third equation by $\mu + \lambda = 1$ and $\bar{S}_2 \supseteq \bar{S}_1$, the fourth equation by $\bar{S}_2 \setminus \bar{S}_1 = B_v$, the fifth by $\lambda b_v = |B_v \cap \bar{S}|$ and $\mu b_v = |B_v \setminus \bar{S}|$, the second inequality by $B_v \setminus \bar{S} \neq \emptyset$ and the last inequality by $\bar{S}_1 \dcup (B_v \cap \bar{S}) = \bar{S}$.
      Hence, both cases contradict the assumption that the inequality for $\bar{S}$ was violated, which concludes the proof of the claim.
   \end{proof}

   \begin{claim}
      \label{TheoremUncapacitatedbMatchingUp}
      The vector $(\bar{x},\bar{y})$ satisfies Constraint~\eqref{ConstraintbMatchingQuadraticUp} for $\bar{G}$ with respect to  $\bar{b} = \onevec[\bar{U} \dcup \bar{W}]$.
   \end{claim}

   \begin{proof}[Proof of Claim~\ref{TheoremUncapacitatedbMatchingUp}]
      For the sake of contradiction, consider some $\bar{S} \subseteq \bar{U} \dcup \bar{W}$  with $\bar{e}_1,\bar{e}_2 \in \delta(\bar{S})$ such that $\bar{x}(\bar{E}[\bar{S}]) + \bar{x}_{\bar{e}_1} + \bar{x}_{\bar{e}_2} - \bar{y} > \lfloor \tfrac{1}{2}(|\bar{S}| + 1) \rfloor $ and, among all such sets, with the minimum number of nodes $v \in U \dcup W$ for which $0 < |\bar{S} \cap B_v| < b_v$ holds.

      This number must be positive, since otherwise Constraint~\eqref{ConstraintbMatchingQuadraticUp} for $S = \setdef{ v \in U \dcup W }[ B_v \subseteq \bar{S} ]$ yields the contradiction 
      \begin{align*}    
         \bar{x}(\bar{E}[\bar{S}]) + \bar{x}_{\bar{e}_1} + \bar{x}_{\bar{e}_2} - \bar{y}
         &= x(E[S]) + x_{e_1} + x_{e_2} - y
         \underset{\text{\eqref{ConstraintbMatchingQuadraticUp}}}{\leq} \left\lfloor \tfrac{1}{2}(b(S)+1) \right\rfloor
         = \left\lfloor \tfrac{1}{2} (|\bar{S}| + 1) \right\rfloor \\
         &< \bar{x}(\bar{E}[\bar{S}]) + \bar{x}_{\bar{e}_1} + \bar{x}_{\bar{e}_2} - \bar{y},
      \end{align*}
      where the last equation holds due to $\bar{S} = \bigcup_{v \in S} B_v$.
      Hence, there exists a node $v \in U \dcup W$ with $0 < |\bar{S} \cap B_v| < b_v$.
      Let $\bar{S}_1 \coloneqq \bar{S} \setminus B_v$ and $\bar{S}_2 \coloneqq \bar{S} \cup B_v$.
      Note that we have $v \notin \Vspecial$, which implies $\bar{e}_1,\bar{e}_2 \in \delta(\bar{S}_i)$ for $i=1,2$.
      Since $v$ does not satisfy $0 < |\bar{S}_i \cap B_v| < b_v$ for $i=1,2$, the choice of $\bar{S}$ implies that Constraint~\eqref{ConstraintbMatchingQuadraticUp} is satisfied for $\bar{S}_1$ and for
      $\bar{S}_2$, i.e.,
      \begin{align}
         \bar{x}(\bar{E}[\bar{S}_1]) + \bar{x}_{\bar{e}_1} + \bar{x}_{\bar{e}_2} - \bar{y} \leq \left\lfloor \tfrac{1}{2} (|\bar{S}_1| + 1) \right\rfloor \text{ and }
         \bar{x}(\bar{E}[\bar{S}_2]) + \bar{x}_{\bar{e}_1} + \bar{x}_{\bar{e}_2} - \bar{y} \leq \left\lfloor \tfrac{1}{2} (|\bar{S}_2| + 1) \right\rfloor.
         \label{InequalitybMatchingProofUpInduction}
      \end{align}
      
      First, we have
      \begin{align*}
         \bar{x}(\bar{E}[\bar{S}_1]) + \bar{x}(\bar{E}[\bar{S}_2]) + \bar{x}_{\bar{e}_1} + \bar{x}_{\bar{e}_2}
         &\leq \sum_{\bar{v} \in \bar{S}_1} \bar{x}(\delta_{\bar{G}}(\bar{v}))
         \underset{\text{\eqref{InequalitybMatchingProofLiftedDegree}}}{\leq} |\bar{S}_1|,
      \end{align*}
      where the first inequality holds since every edge whose $\bar{x}$-variable appears once (resp.\ twice) has at least one (resp.\ both) endnode(s) in $\bar{S}_1$.
      Second,
      \begin{align*}
         \bar{x}_{\bar{e}_1} + \bar{x}_{\bar{e}_2} - \bar{y}
         &= x_{e_1} + x_{e_2} - y
         \leq 1
      \end{align*}
      holds, where the inequality corresponds to~\eqref{ConstraintbMatchingQuadraticUp} for $S = \setdef{u_1,u_2}$,
      which satisfies $b(S) = 2$.
      The sum of both inequalities and $-\bar{y} \leq 0$ yields
      \begin{align}
         \bar{x}(\bar{E}[\bar{S}_1]) + \bar{x}(\bar{E}[\bar{S}_2]) + 2\bar{x}_{\bar{e}_1} + 2\bar{x}_{\bar{e}_2}
         -2\bar{y}
         &\leq |\bar{S}_1| + 1,
         \label{InequalitybMatchingProofUpSum}
      \end{align}
      a relation we exploit below.

      Again, $\lambda \coloneqq |B_v \cap \bar{S}| / b_v$ and $\mu \coloneqq |B_v \setminus \bar{S}| / b_v$
      are nonnegative and satisfy $\lambda + \mu = 1$, and $\bar{x}$ satisfies Equation~\eqref{InequalitybMatchingProofDownConvex}.

      If $\mu - \lambda \geq 0$, we obtain
      \begin{align*}
         &\quad~\bar{x}(\bar{E}[\bar{S}]) + \bar{x}_{\bar{e}_1} + \bar{x}_{\bar{e}_2} - \bar{y}
         \underset{\text{\eqref{InequalitybMatchingProofDownConvex}}}{=} \mu \bar{x}(\bar{E}[S_1]) + \lambda \bar{x}(\bar{E}[S_2]) + \bar{x}_{\bar{e}_1} + \bar{x}_{\bar{e}_2} - \bar{y} \\
         &= (\mu - \lambda) (\bar{x}(\bar{E}[\bar{S}_1]) + \bar{x}_{\bar{e}_1} + \bar{x}_{\bar{e}_2} - \bar{y}) + \lambda ( \bar{x}(\bar{E}[\bar{S}_1]) + \bar{x}(\bar{E}[\bar{S}_2]) + 2\bar{x}_{\bar{e}_1} + 2\bar{x}_{\bar{e}_2} -2\bar{y} ) \\
         &\leq \tfrac{1}{2} (\mu - \lambda) (|\bar{S}_1| + 1) + \lambda (|\bar{S}_1| + 1)
         = \tfrac{1}{2} (|\bar{S}_1| + 1)
         \leq \lfloor \tfrac{1}{2} (|\bar{S}| + 1) \rfloor,
      \end{align*}
      where the second equation holds due to $(\mu - \lambda) + 2\lambda = 1$, the first inequality by~\eqref{InequalitybMatchingProofUpInduction} together with $\mu - \lambda \geq 0$ and~\eqref{InequalitybMatchingProofUpSum} together with $\lambda \geq 0$, the third equation by $\mu + \lambda = 1$, and the last inequality by $\bar{S}_1 \subsetneqq \bar{S}$.

      Otherwise, i.e., if $\lambda - \mu \geq 0$, we obtain
      \begin{align*}
         &\quad~\bar{x}(\bar{E}[\bar{S}]) + \bar{x}_{\bar{e}_1} + \bar{x}_{\bar{e}_2} - \bar{y}
         \underset{\text{\eqref{InequalitybMatchingProofDownConvex}}}{=} \mu \bar{x}(\bar{E}[S_1]) + \lambda \bar{x}(\bar{E}[S_2]) + \bar{x}_{\bar{e}_1} + \bar{x}_{\bar{e}_2} - \bar{y} \\
         &= (\lambda - \mu) (\bar{x}(\bar{E}[\bar{S}_2]) + \bar{x}_{\bar{e}_1} + \bar{x}_{\bar{e}_2} - \bar{y}) + \mu ( \bar{x}(\bar{E}[\bar{S}_1]) + \bar{x}(\bar{E}[\bar{S}_2]) + 2\bar{x}_{\bar{e}_1} + 2\bar{x}_{\bar{e}_2} - 2\bar{y} ) \\
         &\leq \tfrac{1}{2} (\lambda - \mu) (|\bar{S}_2| + 1) + \mu (|\bar{S}_1| + 1)
         =\tfrac{1}{2} (|\bar{S}_1| + 1) + \tfrac{1}{2} (\lambda - \mu) |\bar{S}_2 \setminus \bar{S}_1| \\
         &=\tfrac{1}{2} (|\bar{S}_1| + 1) + \tfrac{1}{2} (\lambda - \mu) b_v
         = \tfrac{1}{2} (|\bar{S}_1| + 1) + \tfrac{1}{2} (|B_v \cap \bar{S}| - |B_v \setminus \bar{S}|) \\
         &\leq \tfrac{1}{2} (|\bar{S}_1| + 1) + \tfrac{1}{2} ( |B_v \cap \bar{S}| - 1 )
         \leq \lfloor \tfrac{1}{2} (|\bar{S}| + 1) \rfloor,
      \end{align*}
      where the second equation holds due to $(\lambda - \mu) + 2\mu = 1$, the first inequality by~\eqref{InequalitybMatchingProofUpInduction} together with $\lambda - \mu \geq 0$ and~\eqref{InequalitybMatchingProofUpSum} together with $\mu \geq 0$, the third equation by $\mu + \lambda = 1$ and $\bar{S}_2 \supseteq \bar{S}_1$, the fourth by $\bar{S}_2 \setminus \bar{S}_1 = B_v$, the fifth by $\lambda b_v = |B_v \cap \bar{S}|$ and $\mu b_v = |B_v \setminus \bar{S}|$, the second inequality by $B_v \setminus \bar{S} \neq \emptyset$ and the last inequality by $\bar{S}_1 \dcup (B_v \cap \bar{S}) = \bar{S}$.
      Hence, both cases contradict the assumption that the inequality for $\bar{S}$ was violated, which concludes the proof of the claim.
   \end{proof}

   We showed that $(\bar{x},\bar{y}) \in Q$, which concludes the proof of the theorem.
\end{proof}

\subsection{Capacitated b-matchings}

We will now generalize the polyhedral results even more, by considering \emph{capacitated} $b$-matchings, i.e., $b$-matchings $x \in \Z_+^E$ that satisfy 
\begin{alignat}{6}
  x_e & \leq c_e        &\qquad& \text{for all } e \in E         \label{ConstraintbMatchingCapacity}
\end{alignat}
for a given capacity vector $c \in \Z_+^E$.
We can relax the requirement $b_v = 1$ for all $v \in \Vspecial$ (from the uncapacitated case) to $c_{e_1} = c_{e_2} = 1$, which already suffices to ensure that $x_{e_1}$ and $x_{e_2}$ are binary.
The generalizations of Constraints~\eqref{ConstraintbMatchingQuadraticDown} and~\eqref{ConstraintbMatchingQuadraticUp} read
\begin{align}
  x(E[S]) + x(F) + y
  &\leq \left\lfloor \tfrac{1}{2}(b(S) + c(F)) \right\rfloor \nonumber \\
  & \text{for all } S \in \bCut \text{ and } F \subseteq \delta(S) \setminus \setdef{e_1,e_2} \text{ with } b(S) + c(F) \text{ odd, and} \label{ConstraintCapacitatedbMatchingQuadraticDown} \\
  x(E[S]) + x(F) + x_{e_1} + x_{e_2} - y
  &\leq \left\lfloor \tfrac{1}{2}(b(S) + c(F) + 1) \right\rfloor \nonumber \\
  & \text{for all } S \in \bCut \text{ and } F \subseteq \delta(S) \setminus \setdef{e_1,e_2} \text{ with } b(S) + c(F) \text{ even.} \label{ConstraintCapacitatedbMatchingQuadraticUp}
\end{align}
Note again that, due to the parity conditions on $b(S) + c(F)$ we could make the right-hand sides more explicitly, e.g., by replacing $\lfloor \tfrac{1}{2}(b(S)+c(F)+1) \rfloor$ by $\tfrac{1}{2}(b(S) + c(F))$.
We still prefer the slightly more complicated form since we will soon observe that the inequalities (the way they are stated) remain valid if the parity of $b(S) + c(F)$ is different.

For matchings and $b$-matchings (satisfying $b_v = 1$ for all $v \in \Vspecial$), the inequality
\begin{align}
  x_{e_1} + x_{e_2} - y \leq 1,
  \label{ConstraintCapacitatedbMatchingQuadraticBad}
\end{align}
which is part of the standard linearization of the product, was implied by some other set of constraints.
This is not always true for capacitated $b$-matchings, and hence we have to consider it explicitly.

We can now state our main result for capacitated $b$-matchings.
\begin{theorem}
  \label{TheoremCapacitatedbMatching}
  The convex hull of all vectors $(x,y) \in \Z_+^E \times \setdef{0,1}$
  with $x \leq c$ and $x(\delta(v)) \leq b_v$ for all $v \in U \dcup W$
  that satisfy $y = 1$ if and only if $x_{e_1} = x_{e_2} = 1$
  is equal to the set of $(x,y) \in \R^E \times \R$ that satisfy
  Constraints~\eqref{ConstraintMatchingNonnegative},
  \eqref{ConstraintMatchingBound},
  \eqref{ConstraintMatchingQuadraticGood},
  \eqref{ConstraintbMatchingDegree},
  \eqref{ConstraintbMatchingCapacity},
  \eqref{ConstraintCapacitatedbMatchingQuadraticDown},
  \eqref{ConstraintCapacitatedbMatchingQuadraticUp}
  and~\eqref{ConstraintCapacitatedbMatchingQuadraticBad}.
\end{theorem}

Using the same proof strategy as in the proof of Theorem~\ref{TheoremUncapacitatedbMatching},
our completeness proof is a modification of a completeness proof for the capacitated $b$-matching polytope
on non-bipartite graphs, as presented in Schrijver's book (see Theorem~32.2 in~\cite{Schrijver03}),
again based on a construction by Tutte~\cite{Tutte54}.

\begin{proof}
   Let $P \coloneqq \setdef{ (x,y) \in \PbmatchOne(G,e_1,e_2) }[ x \leq c ]$.
   The proof is structured as follows.
   We first show validity of the inequalities and describe the construction of an extended formulation based on an auxiliary graph.
   To establish the completeness of our proposed inequality description of $P$ we will then show that any point that satisfies the proposed inequalities can be lifted to a point in the extended formulation.

   \begin{claim}
      \label{TheoremCapacitatedbMatchingValidity}
      Inequalities~\eqref{ConstraintCapacitatedbMatchingQuadraticDown} and~\eqref{ConstraintCapacitatedbMatchingQuadraticUp} are valid for $P$ for arbitrary sets $S \in \bCut$ and $F \subseteq \delta(S) \setminus \setdef{e_1,e_2}$, regardless of the parity of $b(S) + c(F)$.
   \end{claim}

   \begin{proof}[Proof of Claim~\ref{TheoremCapacitatedbMatchingValidity}]
      Consider an integer vector $(x,y) \in \PbmatchOne$ satisfying~\eqref{ConstraintbMatchingCapacity}.
      We have
      \begin{align*}
         x(E[S]) + x(F) + y
         &\underset{\text{\eqref{ConstraintMatchingQuadraticGood}}}{\leq} \tfrac{1}{2}( 2x(E[S]) + x(F) + x_{e_1} + x_{e_2}) + \tfrac{1}{2} x(F) \\
         &\leq \tfrac{1}{2} \sum_{v \in S} x(\delta(v)) + \tfrac{1}{2} x(F)
         \leq \tfrac{1}{2} b(S) + \tfrac{1}{2} x(F)
         \underset{\text{\eqref{ConstraintbMatchingCapacity}}}{\leq} \tfrac{1}{2} b(S) + \tfrac{1}{2} c(F),
      \end{align*}
      where the second inequality holds since every edge whose $x$-variable aappears once (resp.\ twice) has one (resp.\ both) endnode(s) in $S$, and the third by the definition of $b$-matchings.
      Inequality~\eqref{ConstraintCapacitatedbMatchingQuadraticDown} now follows, since the left-hand side of the formula is integral, allowing us to round the right-hand side down.
      The fact that we added several valid inequalities shows that the inequality is redundant for sets $S$ for which $b(S) + c(F)$ is even, since rounding has no effect in this case.

      Similarly, we obtain
      \begin{align*}
         &\quad~ x(E[S]) + x(F) + x_{e_1} + x_{e_2} - y \\
         &= \tfrac{1}{2} ( 2x(E[S]) + x(F) + x_{e_1} + x_{e_2} ) + \tfrac{1}{2}x(F) + \tfrac{1}{2}(x_{e_1} + x_{e_2} - y) - \tfrac{1}{2}y \\
         &\leq \tfrac{1}{2} \sum_{v \in S} x(\delta(v)) + \tfrac{1}{2}x(F) + \tfrac{1}{2}(x_{e_1} + x_{e_2} - y) - \tfrac{1}{2}y
         \leq \tfrac{1}{2} \sum_{v \in S} x(\delta(v)) + \tfrac{1}{2}x(F) + \tfrac{1}{2} - \tfrac{1}{2}y \\
         &\leq \tfrac{1}{2} b(S) + \tfrac{1}{2}x(F) + \tfrac{1}{2} - \tfrac{1}{2}y
         \underset{\text{\eqref{ConstraintbMatchingCapacity}}}{\leq} \tfrac{1}{2} b(S) + \tfrac{1}{2}c(F) + \tfrac{1}{2} - \tfrac{1}{2}y
         \underset{\text{\eqref{ConstraintMatchingBound}}}{\leq} \tfrac{1}{2} b(S) + \tfrac{1}{2}c(F) + \tfrac{1}{2},
      \end{align*}
      where the first inequality holds since every edge whose $x$-variable appears in the first summand once (resp.\ twice) has one (resp.\ both) endnode(s) in $S$, the second due to $x_{e_1} \cdot x_{e_2} = y$, and the third by the definition of $b$-matchings.
      Validity of Inequality~\eqref{ConstraintCapacitatedbMatchingQuadraticUp} for $P$ now follows, since the left-hand side of the formula is integral, allowing us to round the right-hand side down.
      
      Again, the fact that we added several valid inequalities shows that the inequality is redundant for sets $S$ for which $b(S) + c(F)$ is odd, since rounding has no effect in this case.
      This concludes the proof of the claim.
   \end{proof}

   We now continue with the completeness of the formulation.
   
   \paragraph{Extended formulation and auxiliary graph.}
   Consider the graph $\bar{G} = (\bar{V}, \bar{E})$ with 
   \begin{align*}
      \bar{V} &= U \dcup W \dcup R, \\
      R &\coloneqq \setdef{ (v,e) }[ v \in e \in E ] \text{ and } \\
      \bar{E} &\coloneqq \bigcup_{e = \setdef{u,w} \in E} \setdef{ \setdef{u,(u,e)}, \setdef{(u,e), (w,e)}, \setdef{(w,e), w} },
   \end{align*}
   i.e., every edge $e = \setdef{u,w}$ is replaced by a $3$-path $u$-$(u,e)$-$(w,e)$-$w$.
   The induced bipartition of $\bar{G}$ is given by $\bar{U} \coloneqq U \dcup \setdef{ (w,e) }[ w \in e \in E \text{ and } w \in W ]$ and $\bar{W} \coloneqq W \dcup (R \setminus \bar{U})$.
   We furthermore assign node values $\bar{b}$ via
   \begin{alignat}{10}
      \bar{b}_v           &\coloneqq b_v         && \text{ for all } v \in U \dcup W \qquad\text{ and }\qquad
      \bar{b}_{(v,e)}     &\coloneqq c_e         && \text{ for all } v \in e \in E.
      \label{EquationCapacitatedbMatchingB}
   \end{alignat}
   Our special edges in $\bar{G}$ are $\bar{e}_i \coloneqq \setdef{(u_i,e_i),(w_i,e_i)}$ for $i = 1,2$.

   We will now consider the $\bar{b}$-matching polytope for $\bar{G}$.
   Observe that the equation
   \begin{align}
      x_{\setdef{v,(v,e)}} + x_{\setdef{(v,e),(v',e)}} = x(\delta_{\bar{G}}((v,e))) = \bar{b}_{(v,e)} = c_e \text{ for all $v \in e \in E$ with $e = \setdef{v,v'}$} \label{EquationCapacitatedbMatchingFace}
   \end{align}

   defines a face $Q$ of $\PbmatchOne[\bar{b}](\bar{G},\bar{e}_1,\bar{e}_2)$.
   Next, consider the projection map $\pi \colon \R^{\bar{E}} \times \R \to \R^E \times \R$ defined via
   \begin{align*}
      \pi((\bar{x},\bar{y})) &\coloneqq \left( x , \bar{y} - \bar{x}_{\bar{e}_1} - \bar{x}_{\bar{e}_2} + 1 \right)
      \text{ with }
      x_e \coloneqq c_e - \bar{x}_{\setdef{(u,e),(w,e)}}
      \text{ for all }
      e = \setdef{u,w} \in E
   \end{align*}

   Note that $\bar{b}_{(v,\bar{e}_i)} = 1$ for $i=1,2$ and, since $Q$ is a face of the $\bar{b}$-matching polytope of $\bar{G}$, Theorem~\ref{TheoremUncapacitatedbMatching} yields a complete description.
   We now verify that $Q$ is indeed an extended formulation for $P$.

   \begin{claim}
      \label{TheoremCapacitatedbMatchingProjection}
      $\pi(Q) = P$.
   \end{claim}

   \begin{proof}[Proof of Claim~\ref{TheoremCapacitatedbMatchingProjection}]
      To see $\pi(Q) \subseteq P$, let $(\bar{x},\bar{y}) \in Q$ be a vertex and let $(x,y) = \pi((\bar{x},\bar{y}))$.
      In particular, we have $\bar{x} \in \Z^{\bar{E}}$ and $\bar{y} \in \setdef{0,1}$.
      From $\bar{x}_{\setdef{(u,e),(w,e)}} \geq 0$ for all $e = \setdef{u,w} \in E$ we derive $x \leq c$.
      For every node $v \in V$, we have
      \begin{align}
         x(\delta_G(v)) 
         = \sum_{\stackrel{e \in \delta_G(v)}{e = \setdef{v,v'}}} (c_e - \bar{x}_{\setdef{(v,e),(v',e)}})
         \underset{\text{\eqref{EquationCapacitatedbMatchingFace}}}{=} \bar{x}(\delta_{\bar{G}}(v)) \label{EquationCapacitatedbMatchingDegrees}
      \end{align}
      where the first equation holds by the definition of $\pi$.
      Thus, $x(\delta_G(v)) \leq \bar{b}_v = b_v$, i.e., $x$ is a $c$-capacitated $b$-matching in $G$.
      Moreover, since $\bar{x}_{\bar{e}_1}$, $\bar{x}_{\bar{e}_2}$ and $\bar{y}$ are binary, we have 
      \[
         y = \bar{y} - \bar{x}_{\bar{e}_1} - \bar{x}_{\bar{e}_2} + 1 = \bar{x}_{\bar{e}_1} \cdot \bar{x}_{\bar{e}_2} - \bar{x}_{\bar{e}_1} - \bar{x}_{\bar{e}_2} + 1 = (1 - \bar{x}_{\bar{e}_1}) \cdot (1 - \bar{x}_{\bar{e}_2}) = x_{e_1} \cdot x_{e_2},
      \]
      which proves $\pi(Q) \subseteq P$ since $Q$ is integral.

      For the reverse direction we consider a vertex $(x,y)$ of $P$, which is integral by the definition of $P$.
      We claim that there exists a unique pre-image $(\bar{x},\bar{y})$ of $(x,y)$ in $Q$.
      For every edge $e = \setdef{u,w} \in E$ it must satisfy
      $\bar{x}_{\setdef{(u,e),(w,e)}} = c_e - x_e$.
      Then, \eqref{EquationCapacitatedbMatchingFace} implies that also
      $\bar{x}_{\setdef{u,(u,e)}} = x_e = \bar{x}_{\setdef{(w,e),w}}$ holds,
      which also proves that the degree constraints are satisfied for all nodes in $R$.
      For nodes $v \in U \dcup W$, 
      \eqref{EquationCapacitatedbMatchingDegrees} is again satisfied, and hence
      $\bar{x}(\delta_{\bar{G}}(v)) \leq b_v = \bar{b}_v$.
      This proves that $\bar{x}$ is a $\bar{b}$-matching.
      Since $x_{e_1}$, $x_{e_2}$ and $y$ are binary, we obtain
      \begin{multline*}
         \bar{y}
         = y + \bar{x}_{\bar{e}_1} + \bar{x}_{\bar{e}_2} - 1
         = x_{e_1} \cdot x_{e_2} + \bar{x}_{\bar{e}_1} + \bar{x}_{\bar{e}_2} - 1
         = (1- \bar{x}_{\bar{e}_1}) \cdot (1 - \bar{x}_{\bar{e}_2}) + \bar{x}_{\bar{e}_1} + \bar{x}_{\bar{e}_2} - 1
         = \bar{x}_{\bar{e}_1} \cdot \bar{x}_{\bar{e}_2},
      \end{multline*}
      which establishes $P \subseteq \pi(Q)$ and concludes the proof of the claim.
   \end{proof}

   Note that Theorem~\ref{TheoremUncapacitatedbMatching} only provides a complete description for $\PbmatchOne[\bar{b}]$ for complete graphs, but $\PbmatchOne[\bar{b}]$ for any subgraph (in particular for $\bar{G}$) is obtained by fixing variables to $0$, i.e., it is a face.

   Let $(x,y) \in \R_+^E \times [0,1]$ satisfy all constraints from the theorem.
   For each edge $e = \setdef{u,w} \in E$ (with $u \in U$ and $w \in W$), define 
   \begin{align}
      \bar{x}_{\setdef{u,(u,e)}} \coloneqq x_e \text{, } \bar{x}_{\setdef{(u,e),(w,e)}} \coloneqq c_e - x_e \text{, } \bar{x}_{\setdef{w,(w,e)}} \coloneqq x_e \text{ and } \bar{y} \coloneqq y -x_{e_1} - x_{e_2} + 1.
      \label{EquationCapacitatedbMatchingLifting}
   \end{align}
   Since $\pi((\bar{x},\bar{y})) = (x,y)$, it remains to show $(\bar{x},\bar{y}) \in Q$.

\medskip

   Since $x$ satisfies Constraints~\eqref{ConstraintMatchingNonnegative} and~\eqref{ConstraintbMatchingCapacity} we have $\bar{x} \geq \zerovec$, i.e., $\bar{x}$ satisfies Constraint~\eqref{ConstraintMatchingNonnegative}.
   From $y - x_{e_1} \leq 0$ and $x_{e_2} \geq 0$ we obtain $\bar{y} \leq 1$, whereas $\bar{y} \geq 0$ follows directly from~\eqref{ConstraintCapacitatedbMatchingQuadraticBad}.
   Hence, Constraint~\eqref{ConstraintMatchingBound} is satisfied.
   For $i=1,2$, Constraint~\eqref{ConstraintMatchingQuadraticGood} reads $\bar{x}_{\bar{e}_i} \geq \bar{y}$, which is equivalent to $c_{e_i} - x_{e_i} \geq y - x_{e_1} - x_{e_2} + 1$, which is satisfied since $x_{e_j} \geq y$ for $j = 3-i$.
   Inequalities~\eqref{ConstraintbMatchingDegree}, \eqref{ConstraintbMatchingQuadraticDown} and~\eqref{ConstraintbMatchingQuadraticUp} are discussed in subsequent claims.
   
   \begin{claim}
      \label{TheoremCapacitatedbMatchingDegree}
      The vector $\bar{x}$ satisfies Constraint~\eqref{ConstraintbMatchingDegree} for $\bar{G}$ with respect to $\bar{b}$.
   \end{claim}

   \begin{proof}[Proof of Claim~\ref{TheoremCapacitatedbMatchingDegree}]
      For all $v \in U \dcup W$, $\bar{x}(\delta_{\bar{G}}(v)) \leq \bar{b}_v$ is implied by $x(\delta(v)) \leq b_v$ due to~\eqref{EquationCapacitatedbMatchingDegrees}.
      The same inequality holds for all nodes in $R$, since
      \begin{align*}
         \bar{x}(\delta((v,e))) 
         = \bar{x}_{\setdef{v,(v,e)}} + \bar{x}_{\setdef{(v,e),(v',e)}} 
         = x_e + (c_e - x_e) 
         = c_e 
         = \bar{b}_{(v,e)},
      \end{align*}
      where $e = \setdef{v,v'}$, i.e., $v'$ is the other endnode of $e$.
      This concludes the proof of the claim.
   \end{proof}
   Note that the claim implies (see Theorem~21.2 in~\cite{Schrijver03}) that $\bar{x}$ is in the $\bar{b}$-matching polytope of $\bar{G}$.
   In particular, $\bar{x}$ satisfies
   \begin{align}
      x(E[\bar{S}]) \leq \left\lfloor \tfrac{1}{2}\bar{b}(\bar{S}) \right\rfloor
      \text{ for all $\bar{S} \subseteq \bar{V}$},
      \label{ConstraintbMatchingBlossom}
   \end{align}
   although these inequalities are redundant since $\bar{G}$ is bipartite.
   
   Although the proofs for the two remaining inequality classes are independent, we combine them since we will carry out the same case analysis.
   \begin{claim}
      \label{TheoremCapacitatedbMatchingDownUp}
      The vector $(\bar{x},\bar{y})$ satisfies Constraints~\eqref{ConstraintbMatchingQuadraticDown} and~\eqref{ConstraintbMatchingQuadraticUp} for $\bar{G}$ with respect to $\bar{b}$.
   \end{claim}
   
   \begin{proof}[Proof of Claim~\ref{TheoremCapacitatedbMatchingDownUp}]
      Assume, for the sake of contradiction, that $(\bar{x},\bar{y})$ violates one of the constraints for $\bar{G}$.
      Let $\bar{S} \subseteq \bar{V}$ be the node set that induces the violated constraint (which was called $S$ in~\eqref{ConstraintbMatchingQuadraticDown} and~\eqref{ConstraintbMatchingQuadraticUp}).
      This means that $\bar{e}_1,\bar{e}_2 \in \delta(\bar{S})$ holds.
      We will distinguish several cases that correspond to different ways of how $\bar{S}$ intersects a $3$-path, and whether the $3$-path corresponds to one of the two special edges (see Figure~\ref{FigureCapacitatedCases}).
      To this end, let $k(\bar{S})$ denote the number of edges $e = \setdef{u,w} \in E$ that satisfy at least one of the following conditions:

      \begin{figure}
         \tikzset{
            graphNodeStyle/.style={draw=black, thick, fill=black!20!white, circle, inner sep=0, minimum size=4mm},
            graphNormalEdgeStyle/.style={draw=black, line width=2.2pt},
            graphSpecialEdgeStyle/.style={draw=black, line width=2.2pt, dashed},
         }
         \begin{center}
            \def\boundingBox{(-0.8,-3) rectangle (4.8,2)}
            \def\boundingBoxShorter{(-0.8,-2.3) rectangle (4.8,2)}
            \begin{subfigure}{0.28\textwidth}
            \begin{center}
               \scalebox{0.7}{%
               \begin{tikzpicture}
                  \draw[clip,use as bounding box,draw=none] \boundingBox;
                  \draw[thick,draw=black!50!white,fill=black!10!white, rounded corners]
                  (2.0,-0.7) -- (2.0,-2.7) -- (4.7,-2.7) -- (4.7,0.7) -- (3.7,0.7) -- (3.2,0.2) -- (0.8,0.2) -- (0.3,0.7) -- (-0.7,0.7) -- (-0.7,-0.7) -- cycle;
                  \node[graphNodeStyle,minimum size=7mm] (u) at (0,0) {$u$};
                  \node[graphNodeStyle,minimum size=7mm] (ue) at (1.25,1) {$u,e$};
                  \node[graphNodeStyle,minimum size=7mm] (we) at (2.75,1) {$w,e$};
                  \node[graphNodeStyle,minimum size=7mm] (w) at (4,0) {$w$};
                  \draw[graphNormalEdgeStyle] (u) -- (ue);
                  \draw[graphNormalEdgeStyle] (ue) -- (we);
                  \draw[graphNormalEdgeStyle] (we) -- (w);
                  \node[graphNodeStyle] (u1e1) at (1.2,-1.5) {};
                  \node[graphNodeStyle] (w1e1) at (2.8,-1.5) {};
                  \node[graphNodeStyle] (u2e2) at (1.2,-2.2) {};
                  \node[graphNodeStyle] (w2e2) at (2.8,-2.2) {};
                  \draw[graphSpecialEdgeStyle] (u1e1) -- (w1e1);
                  \draw[graphSpecialEdgeStyle] (u2e2) -- (w2e2);
               \end{tikzpicture}
               }
            \end{center}
            \caption{Both endnodes, no middle node.}
            \label{FigureCapacitatedCasesBothEnd}
            \end{subfigure}
            \hspace{1mm}
            \begin{subfigure}{0.28\textwidth}
            \begin{center}
               \scalebox{0.7}{%
               \begin{tikzpicture}
                  \draw[clip,use as bounding box,draw=none] \boundingBox;
                  \draw[thick,draw=black!50!white,fill=black!10!white, rounded corners]
                  (2.0,-0.7) -- (2.0,-2.7) -- (4.7,-2.7) -- (4.7,1.7) -- (2.0,1.7) -- (2.0,0.2) -- (0.8,0.2) -- (0.3,0.7) -- (-0.7,0.7) -- (-0.7,-0.7) -- cycle;
                  \node[graphNodeStyle,minimum size=7mm] (u) at (0,0) {$u$};
                  \node[graphNodeStyle,minimum size=7mm] (ue) at (1.25,1) {$u,e$};
                  \node[graphNodeStyle,minimum size=7mm] (we) at (2.75,1) {$w,e$};
                  \node[graphNodeStyle,minimum size=7mm] (w) at (4,0) {$w$};
                  \draw[graphNormalEdgeStyle] (u) -- (ue);
                  \draw[graphNormalEdgeStyle] (ue) -- (we);
                  \draw[graphNormalEdgeStyle] (we) -- (w);
                  \node[graphNodeStyle] (u1e1) at (1.2,-1.5) {};
                  \node[graphNodeStyle] (w1e1) at (2.8,-1.5) {};
                  \node[graphNodeStyle] (u2e2) at (1.2,-2.2) {};
                  \node[graphNodeStyle] (w2e2) at (2.8,-2.2) {};
                  \draw[graphSpecialEdgeStyle] (u1e1) -- (w1e1);
                  \draw[graphSpecialEdgeStyle] (u2e2) -- (w2e2);
               \end{tikzpicture}
               }
            \end{center}
            \caption{Both endnodes and a middle node.}
            \label{FigureCapacitatedCasesBothEndOneMiddle}
            \end{subfigure}
            \hspace{1mm}
            \begin{subfigure}{0.29\textwidth}
            \begin{center}
               \scalebox{0.7}{%
               \begin{tikzpicture}
                  \draw[clip,use as bounding box,draw=none] \boundingBox;
                  \draw[thick,draw=black!50!white,fill=black!10!white, rounded corners]
                  (2.0,-0.7) -- (2.0,-2.0) -- (4.7,-2.0) -- (4.7,1.7) -- (2.0,1.7) -- (2.0,0.2) -- (0.8,0.2) -- (0.3,0.7) -- (-0.7,0.7) -- (-0.7,-0.7) -- cycle;
                  \node[graphNodeStyle,minimum size=7mm] (u) at (0,0) {$u$};
                  \node[graphNodeStyle,minimum size=7mm] (ue) at (1.25,1) {$u,e$};
                  \node[graphNodeStyle,minimum size=7mm] (we) at (2.75,1) {$w,e$};
                  \node[graphNodeStyle,minimum size=7mm] (w) at (4,0) {$w$};
                  \draw[graphNormalEdgeStyle] (u) -- (ue);
                  \draw[graphSpecialEdgeStyle] (ue) -- (we);
                  \draw[graphNormalEdgeStyle] (we) -- (w);
                  \node[graphNodeStyle] (u1e1) at (1.2,-1.5) {};
                  \node[graphNodeStyle] (w1e1) at (2.8,-1.5) {};
                  \draw[graphSpecialEdgeStyle] (u1e1) -- (w1e1);
               \end{tikzpicture}
               }
            \end{center}
            \caption{Both endnodes and a middle node; $e \in \setdef{e_1,e_2}$.}
            \label{FigureCapacitatedCasesBothEndOneMiddleSpecial}
            \end{subfigure}
            \hspace{1mm}
            \begin{subfigure}{0.28\textwidth}
            \begin{center}
               \scalebox{0.7}{%
               \begin{tikzpicture}
                  \draw[clip,use as bounding box,draw=none] \boundingBox;
                  \draw[thick,draw=black!50!white,fill=black!10!white, rounded corners]
                  (2.0,1.7) -- (2.0,-2.7) -- (4.7,-2.7) -- (4.7,-0.7) -- (3.0,-0.7) -- (3.0,-0.0) -- (3.7,0.7) -- (3.7,1.7) -- cycle;
                  \node[graphNodeStyle,minimum size=7mm] (u) at (0,0) {$u$};
                  \node[graphNodeStyle,minimum size=7mm] (ue) at (1.25,1) {$u,e$};
                  \node[graphNodeStyle,minimum size=7mm] (we) at (2.75,1) {$w,e$};
                  \node[graphNodeStyle,minimum size=7mm] (w) at (4,0) {$w$};
                  \draw[graphNormalEdgeStyle] (u) -- (ue);
                  \draw[graphNormalEdgeStyle] (ue) -- (we);
                  \draw[graphNormalEdgeStyle] (we) -- (w);
                  \node[graphNodeStyle] (u1e1) at (1.2,-1.5) {};
                  \node[graphNodeStyle] (w1e1) at (2.8,-1.5) {};
                  \node[graphNodeStyle] (u2e2) at (1.2,-2.2) {};
                  \node[graphNodeStyle] (w2e2) at (2.8,-2.2) {};
                  \draw[graphSpecialEdgeStyle] (u1e1) -- (w1e1);
                  \draw[graphSpecialEdgeStyle] (u2e2) -- (w2e2);
               \end{tikzpicture}
               }
            \end{center}
            \caption{Only a middle node. \newline}
            \label{FigureCapacitatedCasesOneMiddle}
            \end{subfigure}
            \hspace{1mm}
            \begin{subfigure}{0.28\textwidth}
            \begin{center}
               \scalebox{0.7}{%
               \begin{tikzpicture}
                  \draw[clip,use as bounding box,draw=none] \boundingBox;
                  \draw[thick,draw=black!50!white,fill=black!10!white, rounded corners]
                  (2.0,-0.7) -- (2.0,-2.7) -- (4.7,-2.7) -- (4.7,-0.7) -- (3.0,-0.7) -- (3.0,0.0) -- (3.7,0.7) -- (3.7,1.7) -- (0.3,1.7) -- (0.3,0.7) -- (1.0,0.0) -- cycle;
                  \node[graphNodeStyle,minimum size=7mm] (u) at (0,0) {$u$};
                  \node[graphNodeStyle,minimum size=7mm] (ue) at (1.25,1) {$u,e$};
                  \node[graphNodeStyle,minimum size=7mm] (we) at (2.75,1) {$w,e$};
                  \node[graphNodeStyle,minimum size=7mm] (w) at (4,0) {$w$};
                  \draw[graphNormalEdgeStyle] (u) -- (ue);
                  \draw[graphNormalEdgeStyle] (ue) -- (we);
                  \draw[graphNormalEdgeStyle] (we) -- (w);
                  \node[graphNodeStyle] (u1e1) at (1.2,-1.5) {};
                  \node[graphNodeStyle] (w1e1) at (2.8,-1.5) {};
                  \node[graphNodeStyle] (u2e2) at (1.2,-2.2) {};
                  \node[graphNodeStyle] (w2e2) at (2.8,-2.2) {};
                  \draw[graphSpecialEdgeStyle] (u1e1) -- (w1e1);
                  \draw[graphSpecialEdgeStyle] (u2e2) -- (w2e2);
               \end{tikzpicture}
               }
            \end{center}
            \caption{Both middle nodes, no endnodes.}
            \label{FigureCapacitatedCasesBothMiddle}
            \end{subfigure}
            \hspace{1mm}
            \begin{subfigure}{0.29\textwidth}
            \begin{center}
               \scalebox{0.7}{%
               \begin{tikzpicture}
                  \draw[clip,use as bounding box,draw=none] \boundingBox;
                  \draw[thick,draw=black!50!white,fill=black!10!white, rounded corners]
                  (2.0,-0.7) -- (2.0,-2.7) -- (4.7,-2.7) -- (4.7,1.7) -- (0.3,1.7) -- (0.3,0.7) -- (1.0,0.0) -- cycle;
                  \node[graphNodeStyle,minimum size=7mm] (u) at (0,0) {$u$};
                  \node[graphNodeStyle,minimum size=7mm] (ue) at (1.25,1) {$u,e$};
                  \node[graphNodeStyle,minimum size=7mm] (we) at (2.75,1) {$w,e$};
                  \node[graphNodeStyle,minimum size=7mm] (w) at (4,0) {$w$};
                  \draw[graphNormalEdgeStyle] (u) -- (ue);
                  \draw[graphNormalEdgeStyle] (ue) -- (we);
                  \draw[graphNormalEdgeStyle] (we) -- (w);
                  \node[graphNodeStyle] (u1e1) at (1.2,-1.5) {};
                  \node[graphNodeStyle] (w1e1) at (2.8,-1.5) {};
                  \node[graphNodeStyle] (u2e2) at (1.2,-2.2) {};
                  \node[graphNodeStyle] (w2e2) at (2.8,-2.2) {};
                  \draw[graphSpecialEdgeStyle] (u1e1) -- (w1e1);
                  \draw[graphSpecialEdgeStyle] (u2e2) -- (w2e2);
               \end{tikzpicture}
               }
            \end{center}
            \caption{One endnode and both middle nodes.}
            \label{FigureCapacitatedCasesOneEndBothMiddle}
            \end{subfigure}
            \hspace{1mm}
            \begin{subfigure}{0.28\textwidth}
            \begin{center}
               \scalebox{0.7}{%
               \begin{tikzpicture}
                  \draw[clip,use as bounding box,draw=none] \boundingBox;
                  \draw[thick,draw=black!50!white,fill=black!10!white, rounded corners]
                  (2.0,-0.7) -- (2.0,-2.7) -- (4.7,-2.7) -- (4.7,-0.7) -- (3.0,-0.7) -- (3.0,0.0) -- (3.7,0.7) -- (3.7,1.7) -- (2.0,1.7) -- (2.0,0.2) -- (0.8,0.2) -- (0.3,0.7) -- (-0.7,0.7) -- (-0.7,-0.7) -- cycle;
                  \node[graphNodeStyle,minimum size=7mm] (u) at (0,0) {$u$};
                  \node[graphNodeStyle,minimum size=7mm] (ue) at (1.25,1) {$u,e$};
                  \node[graphNodeStyle,minimum size=7mm] (we) at (2.75,1) {$w,e$};
                  \node[graphNodeStyle,minimum size=7mm] (w) at (4,0) {$w$};
                  \draw[graphNormalEdgeStyle] (u) -- (ue);
                  \draw[graphNormalEdgeStyle] (ue) -- (we);
                  \draw[graphNormalEdgeStyle] (we) -- (w);
                  \node[graphNodeStyle] (u1e1) at (1.2,-1.5) {};
                  \node[graphNodeStyle] (w1e1) at (2.8,-1.5) {};
                  \node[graphNodeStyle] (u2e2) at (1.2,-2.2) {};
                  \node[graphNodeStyle] (w2e2) at (2.8,-2.2) {};
                  \draw[graphSpecialEdgeStyle] (u1e1) -- (w1e1);
                  \draw[graphSpecialEdgeStyle] (u2e2) -- (w2e2);
               \end{tikzpicture}
               }
            \end{center}
            \caption{An endnode and a nonadjacent middle node.}
            \label{FigureCapacitatedCasesOneEndNonadjacentMiddle}
            \end{subfigure}
            \hspace{1mm}
            \begin{subfigure}{0.28\textwidth}
            \begin{center}
               \scalebox{0.7}{%
               \begin{tikzpicture}
                  \draw[clip,use as bounding box,draw=none] \boundingBox;
                  \draw[thick,draw=black!50!white,fill=black!10!white, rounded corners]
                  (2.0,1.7) -- (2.0,-2.0) -- (4.7,-2.0) -- (4.7,-0.7) -- (3.0,-0.7) -- (3.0,-0.0) -- (3.7,0.7) -- (3.7,1.7) -- cycle;
                  \node[graphNodeStyle,minimum size=7mm] (u) at (0,0) {$u$};
                  \node[graphNodeStyle,minimum size=7mm] (ue) at (1.25,1) {$u,e$};
                  \node[graphNodeStyle,minimum size=7mm] (we) at (2.75,1) {$w,e$};
                  \node[graphNodeStyle,minimum size=7mm] (w) at (4,0) {$w$};
                  \draw[graphNormalEdgeStyle] (u) -- (ue);
                  \draw[graphSpecialEdgeStyle] (ue) -- (we);
                  \draw[graphNormalEdgeStyle] (we) -- (w);
                  \node[graphNodeStyle] (u1e1) at (1.2,-1.5) {};
                  \node[graphNodeStyle] (w1e1) at (2.8,-1.5) {};
                  \draw[graphSpecialEdgeStyle] (u1e1) -- (w1e1);
               \end{tikzpicture}
               }
            \end{center}
            \caption{Only a middle node; $e \in \setdef{e_1,e_2}$. \phantom{foobar}}
            \label{FigureCapacitatedCasesOneMiddleSpecial}
            \end{subfigure}
            \hspace{1mm}
            \begin{subfigure}{0.29\textwidth}
            \begin{center}
               \scalebox{0.7}{%
               \begin{tikzpicture}
                  \draw[clip,use as bounding box,draw=none] \boundingBox;
                  \draw[thick,draw=black!50!white,fill=black!10!white, rounded corners]
                  (2.0,-0.7) -- (2.0,-2.0) -- (4.7,-2.0) -- (4.7,-0.7) -- (3.0,-0.7) -- (3.0,0.0) -- (3.7,0.7) -- (3.7,1.7) -- (2.0,1.7) -- (2.0,0.2) -- (0.8,0.2) -- (0.3,0.7) -- (-0.7,0.7) -- (-0.7,-0.7) -- cycle;
                  \node[graphNodeStyle,minimum size=7mm] (u) at (0,0) {$u$};
                  \node[graphNodeStyle,minimum size=7mm] (ue) at (1.25,1) {$u,e$};
                  \node[graphNodeStyle,minimum size=7mm] (we) at (2.75,1) {$w,e$};
                  \node[graphNodeStyle,minimum size=7mm] (w) at (4,0) {$w$};
                  \draw[graphNormalEdgeStyle] (u) -- (ue);
                  \draw[graphSpecialEdgeStyle] (ue) -- (we);
                  \draw[graphNormalEdgeStyle] (we) -- (w);
                  \node[graphNodeStyle] (u1e1) at (1.2,-1.5) {};
                  \node[graphNodeStyle] (w1e1) at (2.8,-1.5) {};
                  \draw[graphSpecialEdgeStyle] (u1e1) -- (w1e1);
               \end{tikzpicture}
               }
            \end{center}
            \caption{An endnode and a nonadjacent middle node; $e \in \setdef{e_1,e_2}$.}
            \label{FigureCapacitatedCasesOneEndNonadjacentMiddleSpecial}
            \end{subfigure}
            \hspace{1mm}
            \begin{subfigure}{0.28\textwidth}
            \begin{center}
               \scalebox{0.7}{%
               \begin{tikzpicture}
                  \draw[clip,use as bounding box,draw=none] \boundingBox;
                  \draw[thick,draw=black!50!white,fill=black!10!white, rounded corners]
                  (-0.7,-0.7) -- (2.0,-0.7) -- (2.0,-2.7) -- (4.7,-2.7) -- (4.7,1.7) -- (-0.7,1.7) -- cycle;
                  \node[graphNodeStyle,minimum size=7mm] (u) at (0,0) {$u$};
                  \node[graphNodeStyle,minimum size=7mm] (ue) at (1.25,1) {$u,e$};
                  \node[graphNodeStyle,minimum size=7mm] (we) at (2.75,1) {$w,e$};
                  \node[graphNodeStyle,minimum size=7mm] (w) at (4,0) {$w$};
                  \draw[graphNormalEdgeStyle] (u) -- (ue);
                  \draw[graphNormalEdgeStyle] (ue) -- (we);
                  \draw[graphNormalEdgeStyle] (we) -- (w);
                  \node[graphNodeStyle] (u1e1) at (1.2,-1.5) {};
                  \node[graphNodeStyle] (w1e1) at (2.8,-1.5) {};
                  \node[graphNodeStyle] (u2e2) at (1.2,-2.2) {};
                  \node[graphNodeStyle] (w2e2) at (2.8,-2.2) {};
                  \draw[graphSpecialEdgeStyle] (u1e1) -- (w1e1);
                  \draw[graphSpecialEdgeStyle] (u2e2) -- (w2e2);
               \end{tikzpicture}
               }
            \end{center}
            \caption{All nodes. \newline}
            \label{FigureCapacitatedCasesAll}
            \end{subfigure}
            \hspace{1mm}
            \begin{subfigure}{0.28\textwidth}
            \begin{center}
               \scalebox{0.7}{%
               \begin{tikzpicture}
                  \draw[clip,use as bounding box,draw=none] \boundingBox;
                  \draw[thick,draw=black!50!white,fill=black!10!white, rounded corners]
                  (2.0,-0.7) -- (2.0,-2.7) -- (4.7,-2.7) -- (4.7,-0.7) -- cycle;
                  \node[graphNodeStyle,minimum size=7mm] (u) at (0,0) {$u$};
                  \node[graphNodeStyle,minimum size=7mm] (ue) at (1.25,1) {$u,e$};
                  \node[graphNodeStyle,minimum size=7mm] (we) at (2.75,1) {$w,e$};
                  \node[graphNodeStyle,minimum size=7mm] (w) at (4,0) {$w$};
                  \draw[graphNormalEdgeStyle] (u) -- (ue);
                  \draw[graphNormalEdgeStyle] (ue) -- (we);
                  \draw[graphNormalEdgeStyle] (we) -- (w);
                  \node[graphNodeStyle] (u1e1) at (1.2,-1.5) {};
                  \node[graphNodeStyle] (w1e1) at (2.8,-1.5) {};
                  \node[graphNodeStyle] (u2e2) at (1.2,-2.2) {};
                  \node[graphNodeStyle] (w2e2) at (2.8,-2.2) {};
                  \draw[graphSpecialEdgeStyle] (u1e1) -- (w1e1);
                  \draw[graphSpecialEdgeStyle] (u2e2) -- (w2e2);
               \end{tikzpicture}
               }
            \end{center}
            \caption{No nodes. \newline}
            \label{FigureCapacitatedCasesNone}
            \end{subfigure}
            \hspace{1mm}
            \begin{subfigure}{0.28\textwidth}
            \begin{center}
               \scalebox{0.7}{%
               \begin{tikzpicture}
                  \draw[clip,use as bounding box,draw=none] \boundingBox;
                  \draw[thick,draw=black!50!white,fill=black!10!white, rounded corners]
                  (2.0,-0.7) -- (2.0,-2.7) -- (4.7,-2.7) -- (4.7,0.7) -- (3.7,0.7) -- cycle;
                  \node[graphNodeStyle,minimum size=7mm] (u) at (0,0) {$u$};
                  \node[graphNodeStyle,minimum size=7mm] (ue) at (1.25,1) {$u,e$};
                  \node[graphNodeStyle,minimum size=7mm] (we) at (2.75,1) {$w,e$};
                  \node[graphNodeStyle,minimum size=7mm] (w) at (4,0) {$w$};
                  \draw[graphNormalEdgeStyle] (u) -- (ue);
                  \draw[graphNormalEdgeStyle] (ue) -- (we);
                  \draw[graphNormalEdgeStyle] (we) -- (w);
                  \node[graphNodeStyle] (u1e1) at (1.2,-1.5) {};
                  \node[graphNodeStyle] (w1e1) at (2.8,-1.5) {};
                  \node[graphNodeStyle] (u2e2) at (1.2,-2.2) {};
                  \node[graphNodeStyle] (w2e2) at (2.8,-2.2) {};
                  \draw[graphSpecialEdgeStyle] (u1e1) -- (w1e1);
                  \draw[graphSpecialEdgeStyle] (u2e2) -- (w2e2);
               \end{tikzpicture}
               }
            \end{center}
            \caption{One endnode, no middle node.}
            \label{FigureCapacitatedCasesOneEnd}
            \end{subfigure}
            \hspace{1mm}
            \begin{subfigure}{0.28\textwidth}
            \begin{center}
               \scalebox{0.7}{%
               \begin{tikzpicture}
                  \draw[clip,use as bounding box,draw=none] \boundingBoxShorter;
                  \draw[thick,draw=black!50!white,fill=black!10!white, rounded corners]
                  (2.0,1.7) -- (2.0,-2) -- (4.7,-2) -- (4.7,1.7) -- cycle;
                  \node[graphNodeStyle,minimum size=7mm] (u) at (0,0) {$u$};
                  \node[graphNodeStyle,minimum size=7mm] (ue) at (1.25,1) {$u,e$};
                  \node[graphNodeStyle,minimum size=7mm] (we) at (2.75,1) {$w,e$};
                  \node[graphNodeStyle,minimum size=7mm] (w) at (4,0) {$w$};
                  \draw[graphNormalEdgeStyle] (u) -- (ue);
                  \draw[graphNormalEdgeStyle] (ue) -- (we);
                  \draw[graphNormalEdgeStyle] (we) -- (w);
                  \node[graphNodeStyle] (u1e1) at (1.2,-0.8) {};
                  \node[graphNodeStyle] (w1e1) at (2.8,-0.8) {};
                  \node[graphNodeStyle] (u2e2) at (1.2,-1.5) {};
                  \node[graphNodeStyle] (w2e2) at (2.8,-1.5) {};
                  \draw[graphSpecialEdgeStyle] (u1e1) -- (w1e1);
                  \draw[graphSpecialEdgeStyle] (u2e2) -- (w2e2);
               \end{tikzpicture}
               }
            \end{center}
            \caption{An endnode and its adjacent middle node.}
            \label{FigureCapacitatedCasesOneEndAdjacentMiddle}
            \end{subfigure}
            \hspace{1mm}
            \begin{subfigure}{0.28\textwidth}
            \begin{center}
               \scalebox{0.7}{%
               \begin{tikzpicture}
                  \draw[clip,use as bounding box,draw=none] \boundingBoxShorter;
                  \draw[thick,draw=black!50!white,fill=black!10!white, rounded corners]
                  (2.0,1.7) -- (2.0,-2.0) -- (4.7,-2.0) -- (4.7,1.7) -- cycle;
                  \node[graphNodeStyle,minimum size=7mm] (u) at (0,0) {$u$};
                  \node[graphNodeStyle,minimum size=7mm] (ue) at (1.25,1) {$u,e$};
                  \node[graphNodeStyle,minimum size=7mm] (we) at (2.75,1) {$w,e$};
                  \node[graphNodeStyle,minimum size=7mm] (w) at (4,0) {$w$};
                  \draw[graphNormalEdgeStyle] (u) -- (ue);
                  \draw[graphSpecialEdgeStyle] (ue) -- (we);
                  \draw[graphNormalEdgeStyle] (we) -- (w);
                  \node[graphNodeStyle] (u1e1) at (1.2,-1.5) {};
                  \node[graphNodeStyle] (w1e1) at (2.8,-1.5) {};
                  \draw[graphSpecialEdgeStyle] (u1e1) -- (w1e1);
               \end{tikzpicture}
               }
            \end{center}
            \caption{An endnode and its adjacent middle node; $e \in \setdef{e_1,e_2}$.}
            \label{FigureCapacitatedCasesOneEndAdjacentMiddleSpecial}
            \end{subfigure}
         \end{center}
         \vspace{-3mm}
         \caption{%
            All different ways of how $\bar{S}$ (highlighted in gray) intersects the node set of the $3$-path corresponding to an edge $e = \setdef{u,w}$ (up to symmetry).
            Dashed edges are the edges $\bar{e}_1$ and $\bar{e}_2$.
         }
         \label{FigureCapacitatedCases}
      \end{figure}

      \begin{itemize}
      \item
         $u,w \in \bar{S}$ and at least one node $v \in e$ satisfies $(v,e) \notin \bar{S}$ (see Figure~\ref{FigureCapacitatedCases}, Cases~\subref{FigureCapacitatedCasesBothEnd}, \subref{FigureCapacitatedCasesBothEndOneMiddle} and~\subref{FigureCapacitatedCasesBothEndOneMiddleSpecial}).
      \item
         at least one node $v \in e$ satisfies $(v,e) \in \bar{S}$ and $v \notin \bar{S}$ (see Figure~\ref{FigureCapacitatedCases}, Cases~\subref{FigureCapacitatedCasesOneMiddle}, \subref{FigureCapacitatedCasesBothMiddle}, \subref{FigureCapacitatedCasesOneEndBothMiddle}, \subref{FigureCapacitatedCasesOneEndNonadjacentMiddle}, \subref{FigureCapacitatedCasesOneMiddleSpecial} and~\subref{FigureCapacitatedCasesOneEndNonadjacentMiddleSpecial}).
      \end{itemize}
      We furthermore assume that $\bar{S}$ is chosen among all sets that induce a violated \emph{facet-defining} constraint (of either constraint type) such that $k(\bar{S})$ is \emph{minimum}.

      \bigskip

      Suppose $k(\bar{S}) = 0$, i.e., Cases~\subref{FigureCapacitatedCasesAll}, \subref{FigureCapacitatedCasesNone}, \subref{FigureCapacitatedCasesOneEnd}, \subref{FigureCapacitatedCasesOneEndAdjacentMiddle} or~\subref{FigureCapacitatedCasesOneEndAdjacentMiddleSpecial} of Figure~\ref{FigureCapacitatedCases} apply to all edges.
      Let $F \coloneqq \setdef{ e \in \delta(S) \setminus \setdef{e_1,e_2} }[ \exists v \in e : v \in \bar{S} \text{ and } (v,e) \in \bar{S} ]$.
      We obtain
      \begin{align*}
         \bar{x}(\bar{E}[\bar{S}]) + \bar{y}
         &\underset{\text{\eqref{EquationCapacitatedbMatchingLifting}}}{=} 2x(E[S]) + (c-x)(E[S]) + x(F) + x_{e_1} + x_{e_2} + \bar{y} \\
         &\underset{\text{\eqref{EquationCapacitatedbMatchingLifting}}}{=} x(E[S]) + c(E[S]) + x(F) + y + 1
         \underset{\text{\eqref{ConstraintCapacitatedbMatchingQuadraticDown}}}{\leq} \left\lfloor \tfrac{1}{2}(b(S) + c(F)) \right\rfloor + c(E[S]) + 1 \\
         &= \left\lfloor \tfrac{1}{2}(b(S) + c(F) + 2c(E[S]) + c_{e_1} + c_{e_2}) \right\rfloor
         \underset{\text{\eqref{EquationCapacitatedbMatchingB}}}{=} \left\lfloor \tfrac{1}{2} \bar{b}(\bar{S}) \right\rfloor \text{ and} \\
         \bar{x}(\bar{E}[\bar{S}]) + \bar{x}_{\bar{e}_1} + \bar{x}_{\bar{e}_2} - \bar{y}
         &\underset{\text{\eqref{EquationCapacitatedbMatchingLifting}}}{=} 2x(E[S]) + (c-x)(E[S]) + x(F) + c_{e_1} + c_{e_2} - \bar{y} \\
         &\underset{\text{\eqref{EquationCapacitatedbMatchingLifting}}}{=} x(E[S]) + c(E[S]) + x(F) + x_{e_1} + x_{e_2} - y + 1 \\
         &\underset{\text{\eqref{ConstraintCapacitatedbMatchingQuadraticUp}}}{\leq} \left\lfloor \tfrac{1}{2}(b(S) + c(F) + 1) \right\rfloor + c(E[S]) + 1 \\
         &= \left\lfloor \tfrac{1}{2}(b(S) + c(F) + 2c(E[S]) + c_{e_1} + c_{e_2} + 1) \right\rfloor
         \underset{\text{\eqref{EquationCapacitatedbMatchingB}}}{=} \left\lfloor \tfrac{1}{2} ( \bar{b}(\bar{S}) + 1 ) \right\rfloor.
      \end{align*}
      This contradicts the assumption that one of the Inequalities~\eqref{ConstraintbMatchingQuadraticDown}
      and~\eqref{ConstraintbMatchingQuadraticUp} for $\bar{S}$ is violated.

\bigskip

      Suppose $k(\bar{S}) \geq 1$ and consider an edge $e = \setdef{u,w} \in E$ satisfying one of the conditions in the definition of $k(\bar{S})$.
      We will distinguish all relevant cases from Figure~\ref{FigureCapacitatedCases}.
      In every case, we will construct a set $\bar{S}'$ related to $\bar{S}$, and use the notation ($\bar{S}'$) in formulas to refer to its (case-specific) definition.
      For brevity we will omit the conclusion that the derived inequalities contradict the assumption that one of the Inequalities~\eqref{ConstraintbMatchingQuadraticDown} or~\eqref{ConstraintbMatchingQuadraticUp} for $\bar{S}$ is violated, since it can be drawn in every case.

\bigskip
      
      \noindent
      \textbf{Case~\subref{FigureCapacitatedCasesBothEnd}:} $u,w \in \bar{S}$, $(u,e),(w,e) \notin \bar{S}$ and $e \notin \setdef{e_1,e_2}$.
      
      Let $\bar{S}' \coloneqq \bar{S} \cup \setdef{(u,e), (w,e)}$ and observe that $\bar{e}_1,\bar{e}_2 \in \delta(\bar{S}')$.
      Notice that $k(\bar{S}') < k(\bar{S})$, and that we can assume that (by the minimality assumption on $k(\bar{S})$)  Inequalities~\eqref{ConstraintbMatchingQuadraticDown} and~\eqref{ConstraintbMatchingQuadraticUp} for $\bar{S}'$ are satisfied.
      We obtain
      \begin{align*}
         \bar{x}(\bar{E}[\bar{S}]) + \bar{y}
         &\underset{\text{($\bar{S}'$),\eqref{EquationCapacitatedbMatchingLifting}}}{=} \bar{x}(\bar{E}[\bar{S}']) - x_e - (c_e - x_e) - x_e + \bar{y}
         \underset{\text{\eqref{ConstraintbMatchingQuadraticDown}}}{\leq} \left\lfloor \tfrac{1}{2}\bar{b}(\bar{S}') \right\rfloor - x_e - c_e \\
         & \underset{\text{($\bar{S}'$),\eqref{EquationCapacitatedbMatchingB}}}{=} \left\lfloor \tfrac{1}{2}(\bar{b}(\bar{S}) + 2c_e ) \right\rfloor - x_e - c_e
         \underset{\text{\eqref{ConstraintMatchingNonnegative}}}{\leq} \left\lfloor \tfrac{1}{2}\bar{b}(\bar{S}) \right\rfloor \text{ and } \\ 
         \bar{x}(\bar{E}[\bar{S}]) + \bar{x}_{\bar{e}_1} + \bar{x}_{\bar{e}_2} - \bar{y}
         &\underset{\text{($\bar{S}'$),\eqref{EquationCapacitatedbMatchingLifting}}}{=} \bar{x}(\bar{E}[\bar{S}']) - x_e - (c_e - x_e) - x_e + \bar{x}_{\bar{e}_1} + \bar{x}_{\bar{e}_2} - \bar{y} \\
         &\underset{\text{\eqref{ConstraintbMatchingQuadraticUp}}}{\leq} \left\lfloor \tfrac{1}{2}(\bar{b}(\bar{S}') + 1) \right\rfloor - x_e - c_e
         \underset{\text{($\bar{S}'$),\eqref{EquationCapacitatedbMatchingB}}}{=} \left\lfloor \tfrac{1}{2}(\bar{b}(\bar{S}) + 2c_e + 1 ) \right\rfloor - x_e - c_e \\
         &\underset{\text{\eqref{ConstraintMatchingNonnegative}}}{\leq} \left\lfloor \tfrac{1}{2}(\bar{b}(\bar{S}) + 1) \right\rfloor.
      \end{align*}

      \bigskip
      
      \noindent
      \textbf{Case~\subref{FigureCapacitatedCasesBothEndOneMiddle}:} $u,(w,e),w \in \bar{S}$, $(u,e) \notin \bar{S}$ and $e \notin \setdef{e_1,e_2}$.
      
      Let $\bar{S}' \coloneqq \bar{S} \cup \setdef{(u,e)}$ and observe that $\bar{e}_1,\bar{e}_2 \in \delta(\bar{S}')$.
      Notice that $k(\bar{S}') < k(\bar{S})$, and we can again assume that Inequalities~\eqref{ConstraintbMatchingQuadraticDown} and~\eqref{ConstraintbMatchingQuadraticUp} for $\bar{S}'$ are satisfied.
      We obtain
      \begin{align*}
         \bar{x}(\bar{E}[\bar{S}]) + \bar{y}
         &\underset{\text{($\bar{S}'$),\eqref{EquationCapacitatedbMatchingLifting}}}{=} \bar{x}(\bar{E}[\bar{S}']) - x_e - (c_e - x_e) + \bar{y}
         \underset{\text{\eqref{ConstraintbMatchingQuadraticDown}}}{\leq} \left\lfloor \tfrac{1}{2}\bar{b}(\bar{S}') \right\rfloor - c_e \\
         & \underset{\text{($\bar{S}'$),\eqref{EquationCapacitatedbMatchingB}}}{=} \left\lfloor \tfrac{1}{2}(\bar{b}(\bar{S}) + 2c_e ) \right\rfloor - c_e
         = \left\lfloor \tfrac{1}{2}\bar{b}(\bar{S}) \right\rfloor \text{ and } \\ 
         \bar{x}(\bar{E}[\bar{S}]) + \bar{x}_{\bar{e}_1} + \bar{x}_{\bar{e}_2} - \bar{y}
         &\underset{\text{($\bar{S}'$),\eqref{EquationCapacitatedbMatchingLifting}}}{=} \bar{x}(\bar{E}[\bar{S}']) - x_e - (c_e - x_e) + \bar{x}_{\bar{e}_1} + \bar{x}_{\bar{e}_2} - \bar{y} \\
         &\underset{\text{\eqref{ConstraintbMatchingQuadraticUp}}}{\leq} \left\lfloor \tfrac{1}{2}(\bar{b}(\bar{S}') + 1) \right\rfloor - c_e \\
         & \underset{\text{($\bar{S}'$),\eqref{EquationCapacitatedbMatchingB}}}{=} \left\lfloor \tfrac{1}{2}(\bar{b}(\bar{S}) + 2c_e + 1 ) \right\rfloor - c_e
         = \left\lfloor \tfrac{1}{2}(\bar{b}(\bar{S}) + 1) \right\rfloor.
      \end{align*}

      \bigskip
      
      \noindent
      \textbf{Case~\subref{FigureCapacitatedCasesBothEndOneMiddleSpecial}:} $u,(w,e),w \in \bar{S}$, $(u,e) \notin \bar{S}$, w.l.o.g.\ $e = e_1$.

      Let $\bar{S}' \coloneqq \bar{S} \cup \setdef{(u,e)}$ and observe that $\bar{e}_1 \notin \delta(\bar{S}')$ and $\bar{e}_2 \in \delta(\bar{S}')$.
      We obtain
      \begin{align*}
         \bar{x}(\bar{E}[\bar{S}]) + \bar{y}
         &\underset{\text{($\bar{S}'$),\eqref{EquationCapacitatedbMatchingLifting}}}{=} \bar{x}(\bar{E}[\bar{S}']) - x_{e_1} - (c_{e_1} - x_{e_1}) + \bar{y}
         \underset{\text{\eqref{ConstraintMatchingQuadraticGood}}}{\leq} \bar{x}(\bar{E}[\bar{S}']) +  \bar{x}_{\bar{e}_2} - 1 \\
         &\underset{\text{\eqref{ConstraintbMatchingBlossom}}}{\leq} \left\lfloor \tfrac{1}{2}\bar{b}(\bar{S}' \cup \bar{e}_2) \right\rfloor - 1
         \underset{\text{($\bar{S}'$),\eqref{EquationCapacitatedbMatchingB}}}{=} \left\lfloor \tfrac{1}{2}\bar{b}(\bar{S}) \right\rfloor
         \text{ and } \\
         \bar{x}(\bar{E}[\bar{S}]) + \bar{x}_{\bar{e}_1} + \bar{x}_{\bar{e}_2} - \bar{y}
         &\underset{\text{($\bar{S}'$),\eqref{EquationCapacitatedbMatchingLifting}}}{=} \bar{x}(\bar{E}[\bar{S}']) - x_{e_1} - (c_{e_1} - x_{e_1}) + \bar{x}_{\bar{e}_1} + \bar{x}_{\bar{e}_2} - \bar{y} \\
         &\underset{\text{\eqref{EquationCapacitatedbMatchingLifting}}}{=} \bar{x}(\bar{E}[\bar{S}']) -c_{e_1} + (c_{e_1}-x_{e_1}) + (c_{e_2}-x_{e_2})  + (x_{e_1} + x_{e_2} - y - 1) \\
         &= \bar{x}(\bar{E}[\bar{S}']) - y
         \underset{\text{\eqref{ConstraintMatchingBound}}}{\leq} \bar{x}(\bar{E}[\bar{S}'])
         \underset{\text{\eqref{ConstraintbMatchingBlossom}}}{\leq} \left\lfloor \tfrac{1}{2}\bar{b}(\bar{S}')  \right\rfloor
         \underset{\text{($\bar{S}'$),\eqref{EquationCapacitatedbMatchingB}}}{=} \left\lfloor \tfrac{1}{2}(\bar{b}(\bar{S}) + 1) \right\rfloor.
      \end{align*}

\bigskip
      
      \noindent
      \textbf{Cases~\subref{FigureCapacitatedCasesOneMiddle} and~\subref{FigureCapacitatedCasesOneEndNonadjacentMiddle}:} $(w,e) \in \bar{S}$, $(u,e),w \notin \bar{S}$ and $e \notin \setdef{e_1,e_2}$.

      Let $\bar{S}' \coloneqq \bar{S} \setminus \setdef{(w,e)}$ and observe that $\bar{e}_1,\bar{e}_2 \in \delta(\bar{S}')$.
      Notice that $k(\bar{S}') < k(\bar{S})$, and we can again assume that Inequalities~\eqref{ConstraintbMatchingQuadraticDown} and~\eqref{ConstraintbMatchingQuadraticUp} for $\bar{S}'$ are satisfied.
      We obtain
      \begin{align*}
         \bar{x}(\bar{E}[\bar{S}]) + \bar{y}
         &\underset{\text{($\bar{S}'$),\eqref{EquationCapacitatedbMatchingLifting}}}{=} \bar{x}(\bar{E}[\bar{S}']) + \bar{y}
         \underset{\text{\eqref{ConstraintbMatchingQuadraticDown}}}{\leq} \left\lfloor \tfrac{1}{2}\bar{b}(\bar{S}') \right\rfloor \\
         &\underset{\text{($\bar{S}'$),\eqref{EquationCapacitatedbMatchingB}}}{=} \left\lfloor \tfrac{1}{2}(\bar{b}(\bar{S}) - c_e ) \right\rfloor
         \leq \left\lfloor \tfrac{1}{2}\bar{b}(\bar{S}) \right\rfloor
         \text{ and } \\ 
         \bar{x}(\bar{E}[\bar{S}]) + \bar{x}_{\bar{e}_1} + \bar{x}_{\bar{e}_2} - \bar{y}
         &\underset{\text{($\bar{S}'$),\eqref{EquationCapacitatedbMatchingLifting}}}{=} \bar{x}(\bar{E}[\bar{S}']) + \bar{x}_{\bar{e}_1} + \bar{x}_{\bar{e}_2} - \bar{y} \\
         &\underset{\text{\eqref{ConstraintbMatchingQuadraticUp}}}{\leq} \left\lfloor \tfrac{1}{2}(\bar{b}(\bar{S}') + 1) \right\rfloor \\
         &\underset{\text{($\bar{S}'$),\eqref{EquationCapacitatedbMatchingB}}}{=} \left\lfloor \tfrac{1}{2}(\bar{b}(\bar{S}) - c_e + 1 ) \right\rfloor
         \leq \left\lfloor \tfrac{1}{2}(\bar{b}(\bar{S}) + 1) \right\rfloor.
      \end{align*}

\bigskip
      
      \noindent
      \textbf{Case~\subref{FigureCapacitatedCasesBothMiddle}:} $(u,e),(w,e) \in \bar{S}$, $u,w \notin \bar{S}$ and $e \notin \setdef{e_1,e_2}$.
      
      Let $\bar{S}' \coloneqq \bar{S} \setminus \setdef{(u,e), (w,e)}$ and observe that $\bar{e}_1,\bar{e}_2 \in \delta(\bar{S}')$.
      Notice that $k(\bar{S}') < k(\bar{S})$, and we can again assume that Inequalities~\eqref{ConstraintbMatchingQuadraticDown} and~\eqref{ConstraintbMatchingQuadraticUp} for $\bar{S}'$ are satisfied.
      We obtain
      \begin{align*}
         \bar{x}(\bar{E}[\bar{S}]) + \bar{y}
         &\underset{\text{($\bar{S}'$),\eqref{EquationCapacitatedbMatchingLifting}}}{=} \bar{x}(\bar{E}[\bar{S}']) + (c_e - x_e) + \bar{y}
         \underset{\text{\eqref{ConstraintbMatchingQuadraticDown}}}{\leq} \left\lfloor \tfrac{1}{2}\bar{b}(\bar{S}') \right\rfloor + c_e - x_e \\
         &\underset{\text{($\bar{S}'$),\eqref{EquationCapacitatedbMatchingB}}}{=} \left\lfloor \tfrac{1}{2}(\bar{b}(\bar{S}) - 2c_e ) \right\rfloor + c_e - x_e
         \underset{\text{\eqref{ConstraintMatchingNonnegative}}}{\leq} \left\lfloor \tfrac{1}{2}\bar{b}(\bar{S}) \right\rfloor \text{ and } \\ 
         \bar{x}(\bar{E}[\bar{S}]) + \bar{x}_{\bar{e}_1} + \bar{x}_{\bar{e}_2} - \bar{y}
         &\underset{\text{($\bar{S}'$),\eqref{EquationCapacitatedbMatchingLifting}}}{=} \bar{x}(\bar{E}[\bar{S}']) + (c_e-x_e) + \bar{x}_{\bar{e}_1} + \bar{x}_{\bar{e}_2} - \bar{y}
         \underset{\text{\eqref{ConstraintbMatchingQuadraticUp}}}{\leq} \left\lfloor \tfrac{1}{2}(\bar{b}(\bar{S}') + 1) \right\rfloor + c_e - x_e \\
         &\underset{\text{($\bar{S}'$),\eqref{EquationCapacitatedbMatchingB}}}{=} \left\lfloor \tfrac{1}{2}(\bar{b}(\bar{S}) - 2c_e + 1 ) \right\rfloor + c_e - x_e
         \underset{\text{\eqref{ConstraintMatchingNonnegative}}}{\leq} \left\lfloor \tfrac{1}{2}(\bar{b}(\bar{S}) + 1) \right\rfloor.
      \end{align*}
      
\bigskip
      
      \noindent
      \textbf{Case~\subref{FigureCapacitatedCasesOneEndBothMiddle}:} $(u,e),(w,e),w \in \bar{S}$, $u \notin \bar{S}$ and $e \notin \setdef{e_1,e_2}$.

      Let $\bar{S}' \coloneqq \bar{S} \setminus \setdef{(u,e),(w,e)}$ and observe that $\bar{e}_1,\bar{e}_2 \in \delta(\bar{S}')$.
      Notice that $k(\bar{S}') < k(\bar{S})$, and we can again assume that Inequalities~\eqref{ConstraintbMatchingQuadraticDown} and~\eqref{ConstraintbMatchingQuadraticUp} for $\bar{S}'$ are satisfied.
      We obtain
      \begin{align*}
         \bar{x}(\bar{E}[\bar{S}]) + \bar{y}
         &\underset{\text{($\bar{S}'$),\eqref{EquationCapacitatedbMatchingLifting}}}{=} \bar{x}(\bar{E}[\bar{S}']) + (c_e - x_e) + x_e + \bar{y}
         \underset{\text{\eqref{ConstraintbMatchingQuadraticDown}}}{\leq} \left\lfloor \tfrac{1}{2}\bar{b}(\bar{S}') \right\rfloor + c_e \\
         &\underset{\text{($\bar{S}'$),\eqref{EquationCapacitatedbMatchingB}}}{=} \left\lfloor \tfrac{1}{2}(\bar{b}(\bar{S}) - 2c_e ) \right\rfloor + c_e
         = \left\lfloor \tfrac{1}{2}\bar{b}(\bar{S}) \right\rfloor
         \text{ and } \\ 
         \bar{x}(\bar{E}[\bar{S}]) + \bar{x}_{\bar{e}_1} + \bar{x}_{\bar{e}_2} - \bar{y}
         &\underset{\text{($\bar{S}'$),\eqref{EquationCapacitatedbMatchingLifting}}}{=} \bar{x}(\bar{E}[\bar{S}']) + (c_e-x_e) + x_e + \bar{x}_{\bar{e}_1} + \bar{x}_{\bar{e}_2} - \bar{y}
         \underset{\text{\eqref{ConstraintbMatchingQuadraticUp}}}{\leq} \left\lfloor \tfrac{1}{2}(\bar{b}(\bar{S}') + 1) \right\rfloor + c_e \\
         &\underset{\text{($\bar{S}'$),\eqref{EquationCapacitatedbMatchingB}}}{=} \left\lfloor \tfrac{1}{2}(\bar{b}(\bar{S}) - 2c_e + 1 ) \right\rfloor + c_e
         = \left\lfloor \tfrac{1}{2}(\bar{b}(\bar{S}) + 1) \right\rfloor.
      \end{align*}

      \bigskip
      
      \noindent
      \textbf{Cases~\subref{FigureCapacitatedCasesOneMiddleSpecial} and~\subref{FigureCapacitatedCasesOneEndNonadjacentMiddleSpecial}:} $(w,e) \in \bar{S}$, $(u,e),w \notin \bar{S}$ and w.l.o.g.\ $e = e_1$.

      Let $\bar{S}' \coloneqq \bar{S} \setminus \setdef{(w,e)}$ and observe that $\bar{e}_1 \notin \delta(\bar{S}')$ and $\bar{e}_2 \in \delta(\bar{S}')$.
      We obtain
      \begin{align*}
         \bar{x}(\bar{E}[\bar{S}]) + \bar{y}
         &\underset{\text{($\bar{S}'$),\eqref{EquationCapacitatedbMatchingLifting}}}{=} \bar{x}(\bar{E}[\bar{S}']) + \bar{y}
         \underset{\text{\eqref{ConstraintMatchingQuadraticGood}}}{\leq} \bar{x}(\bar{E}[\bar{S}']) + \bar{x}_{\bar{e}_2}
         \underset{\text{\eqref{ConstraintbMatchingBlossom}}}{\leq} \left\lfloor \tfrac{1}{2}\bar{b}(\bar{S}' \cup \bar{e}_2) \right\rfloor \\
         &= \left\lfloor \tfrac{1}{2}(\bar{b}(\bar{S}') + 1) \right\rfloor
         = \left\lfloor \tfrac{1}{2}\bar{b}(\bar{S}) \right\rfloor  \text{ and } \\
         \bar{x}(\bar{E}[\bar{S}]) + \bar{x}_{\bar{e}_1} + \bar{x}_{\bar{e}_2} - \bar{y}
         &\underset{\text{($\bar{S}'$),\eqref{EquationCapacitatedbMatchingLifting}}}{=} \bar{x}(\bar{E}[\bar{S}']) + \bar{x}_{\bar{e}_1} + \bar{x}_{\bar{e}_2} - \bar{y} \\
         &\underset{\text{\eqref{EquationCapacitatedbMatchingLifting}}}{=} \bar{x}(\bar{E}[\bar{S}']) + (c_{e_1} - x_{e_1}) + (c_{e_2} - x_{e_2}) + (x_{e_1} + x_{e_2} - y - 1) \\
         &= \bar{x}(\bar{E}[\bar{S}']) - y + 1
         \underset{\text{\eqref{ConstraintMatchingBound}}}{\leq} \bar{x}(\bar{E}[\bar{S}']) + 1
         \underset{\text{\eqref{ConstraintbMatchingBlossom}}}{\leq} \left\lfloor \tfrac{1}{2}\bar{b}(\bar{S}')  \right\rfloor + 1
         \underset{\text{($\star$)}}{=} \left\lfloor \tfrac{1}{2}(\bar{b}(\bar{S}) + 1) \right\rfloor.
      \end{align*}
      By Theorem~\ref{TheoremUncapacitatedbMatching}, Inequality~\eqref{ConstraintbMatchingQuadraticUp} is only facet-defining for sets $\bar{S}$ with $\bar{b}(\bar{S})$ even.
      Hence, for this inequality, $\bar{b}_{(w,e)} = 1$ implies that $\bar{b}(\bar{S}')$ is odd, which proves ($\star$).

\bigskip

      We now argue that the case distinction is complete.
      For every edge $e = \setdef{u,w} \in E \setminus \setdef{e_1,e_2}$, the four nodes $u$, $(u,e)$, $(w,e)$ and $w$ can (independently) be contained in $\bar{S}$ or not, which yields 16 cases in total.
      The 6 cases~\subref{FigureCapacitatedCasesBothEndOneMiddle}, \subref{FigureCapacitatedCasesOneMiddle},  \subref{FigureCapacitatedCasesOneEndBothMiddle}, \subref{FigureCapacitatedCasesOneEndNonadjacentMiddle},  \subref{FigureCapacitatedCasesOneEnd} and~\subref{FigureCapacitatedCasesOneEndAdjacentMiddle}
      each represent two such possibilities by exchanging $u$ and $w$, while the 4 cases~\subref{FigureCapacitatedCasesBothEnd},  \subref{FigureCapacitatedCasesBothMiddle}, \subref{FigureCapacitatedCasesAll} and~\subref{FigureCapacitatedCasesNone} are symmetric.
      Since $2 \cdot 6 + 4 = 16$ and since no possibility is considered in more than one of the mentioned cases, all possibilities are considered.
      The additional cases~\subref{FigureCapacitatedCasesBothEndOneMiddleSpecial}, \subref{FigureCapacitatedCasesOneMiddleSpecial}, \subref{FigureCapacitatedCasesOneEndNonadjacentMiddleSpecial} and~\subref{FigureCapacitatedCasesOneEndAdjacentMiddleSpecial} arise from the previous ones by selecting those for which $e \in \delta(\bar{S})$, which is required for $e \in \setdef{e_1,e_2}$.
      The inequalities derived in all cases contradict the assumption that one of the Inequalities~\eqref{ConstraintbMatchingQuadraticDown} or~\eqref{ConstraintbMatchingQuadraticUp} for $\bar{S}$ is violated, which concludes the proof of the claim.
   \end{proof}

   This establishes that $(\bar{x},\bar{y}) \in Q$, which concludes the proof of the theorem.
\end{proof}

\section{Discussion}

The observation $\PmatchOne = \PmatchOneDown \cap \PmatchOneUp$ from Section~\ref{SectionResults}  is not specific to matching polytopes.
In fact, this is a property of convex sets, as shown in the following proposition.
By $\unitvec{i}$ we denote the $i$'th unit vector.

\begin{proposition}
  \label{TheoremIntersectMonotonizations}
  Let $C \subseteq \R^n$ be a convex set and let $C^{\uparrow} \coloneqq \setdef{ x + \lambda \unitvec{1} }[ x \in C, \lambda \geq 0]$
  and $C^{\downarrow} \coloneqq \setdef{ x - \lambda \unitvec{1} }[ x \in C, \lambda \geq 0]$ \reviewFix{be} its respective up- and downward monotonization
  of the first variable.
  Then $C = C^{\uparrow} \cap C^{\downarrow}$.
\end{proposition}

\begin{proof}
  Clearly, $C \subseteq C^{\uparrow},C^{\downarrow}$ and thus $C \subseteq C^{\uparrow} \cap C^{\downarrow}$.
  In order to prove the reverse direction,
  let $x \in C^{\uparrow} \cap C^{\downarrow}$.
  By definition, there exist $x^{(1)},x^{(2)} \in C$ and $\lambda_1,\lambda_2 \geq 0$ such that $x^{(1)} + \lambda_1 \unitvec{1} = x = x^{(2)} - \lambda_2 \unitvec{1}$.
  If $\lambda_1 = 0$, then $x = x^{(1)} \in C$, and we are done.
  Otherwise, the equation
  \begin{gather*}
    \frac{\lambda_2}{\lambda_1 + \lambda_2} x^{(1)} + \frac{\lambda_1}{\lambda_1 + \lambda_2} x^{(2)}
    = \frac{ \lambda_2 (x - \lambda_1 \unitvec{1}) + \lambda_1 (x + \lambda_2 \unitvec{1}) }{\lambda_1 + \lambda_2}
    = x
  \end{gather*}
  proves that $x$ is a convex combination of two points in $C$, i.e., $x \in C$,
  which concludes the proof.
\end{proof}

In the case of matching polytopes we intersect the up- and downward monotonizations with the $0/1$-cube,
but this does not interfere with the arguments provided above.
Note that Proposition~\ref{TheoremIntersectMonotonizations} does not generalize to the simultaneous monotonization
of \emph{several} variables.
To see this, consider $P = \conv\setdef{ \transpose{(0,0)}, \transpose{(1,1)} }$.
Its upward-monotonization w.r.t. two both variables is $P + \R^2_+ = \R^2_+$,
its downward-monotonization is $P - \R^2_+ = \transpose{(1,1)} + \R^2_-$,
but their intersection is equal to $[0,1]^2 \neq P$.
Hence, this is a purely one-dimensional phenomenon.

\bigskip

\paragraph{Descriptions of monotonizations.}
A second property is specific, at least to polytopes arising from one-term linearizations:
we can obtain the complete description for $\PmatchOneDown$ from the one for $\PmatchOne$
by omitting the $\leq$-inequalities that have a negative $y$-coefficient.
Similarly, we obtain the complete description for $\PmatchOneUp$ from the one for $\PmatchOne$
by omitting the $\leq$-inequalities that have a positive $y$-coefficient
and adding $y \leq 1$ (which is not facet-defining for $\PmatchOne$, see Proposition~\ref{TheoremFacetBounds}).
The reason turns out to be that all facets of the projection of $\PmatchOne$ onto the $x$-variables
are projections of facets of $\PmatchOne$.
The arguments for the upward-monotonization are as follows:

Let $P \subseteq \R^{n+1}$ be a polytope.
After normalizing, we can write its outer description as 
$$P = \setdef{ (x,y) \in \R^n \times \R }[ Ax \leq b,~ Bx + \onevec y \leq c,~ Cx - \onevec y \leq d ].$$
We assume that $P$'s projection onto the $x$-variables is the polytope defined by $Ax \leq b$ only.
$P$'s upward-monotonization can be obtained by projecting the extended formulation
$$\setdef{ (x,y,y') \in \R^n \times \R \times \R }[ (x,y) \in P,~ y - y' \leq 0 ]$$
onto the $(x,y')$-variables.
Fourier-Motzkin elimination yields that this projection is described
by $Ax \leq b$, $Cx - \onevec y' \leq d$ and inequalities that
are the sum of an inequality from $Bx + \onevec y \leq c$ and
an inequality from $Cx - \onevec y \leq d$. 
Since the resulting inequalities are already valid for $P$'s projection onto the $x$-variables,
these are already present in $Ax \leq b$.

\paragraph{Proof technique.}
The technique we described and applied in Section~\ref{SectionTechnique} is quite special and does not work for arbitrary polytopes.
In fact it heavily depends on the fact that $P$ is highly related to $Q$, e.g., a subpolytope.
Clearly, the more complicated the modifications are, the more involved the proof will probably be.
Thus, on the one hand we believe that the applicability of the technique is quite limited.
On the other hand, it does not require LP duality, and hence it could be useful when duality-based methods become unattractive because of many inequality classes.
In such a case, a duality-based approach would involve several sets of dual multipliers, which may complicate formulas.
In contrast to this, the proposed technique may only have to consider each class of inequalities separately. This may have the advantage of being simpler and the disadvantage of producing longer proofs.

\paragraph{Future directions.}
There are several directions into which future research may lead.
On the modeling level, a generalization of Theorem~\ref{TheoremComplete} (or even Theorem~\ref{TheoremCapacitatedbMatching}) to more than one quadratic term or to a single cubic term is conceivable.
Such results were obtained for matroids, see~\cite{FischerFM18}.
On the problem level, since bipartite matching is a special case of matroid intersection,
one may address the question of a complete description of the corresponding polytope for the intersection of two arbitrary matroids
(together with a quadratic term).

\bigskip

\textbf{Acknowledgements.}
The author is grateful to Mirjam Friesen and Volker Kaibel for valueable discussions,
and to the referee whose comments lead to investigation of $b$-matchings and to significant improvements in the presentation of the material.

\bibliographystyle{plain}
\bibliography{references}

\begin{thebibliography}{10}

\bibitem{Birkhoff46}
Garrett Birkhoff.
\newblock Tres observaciones sobre el algebra lineal.
\newblock {\em Revista Facultad de Ciencias Exactas, Puras y Aplicadas
  Universidad Nacional de Tucuman, Serie A (Matematicas y Fisica Teorica)},
  5:147--151, 1946.

\bibitem{BuchheimK13}
Christoph Buchheim and Laura Klein.
\newblock {The Spanning Tree Problem with One Quadratic Term}.
\newblock In Kamiel Cornelissen, Ruben Hoeksma, Johann Hurink, and Bodo
  Manthey, editors, {\em {12th Cologne-Twente Workshop on Graphs and
  Combinatorial Optimization--CTW 2013}}, pages 31--34, 2013.

\bibitem{BuchheimK14}
Christoph Buchheim and Laura Klein.
\newblock {Combinatorial optimization with one quadratic term: Spanning trees
  and forests}.
\newblock {\em Discrete Applied Mathematics}, 177(0):34--52, 2014.

\bibitem{Edmonds65b}
Jack Edmonds.
\newblock {Maximum Matching and a Polyhedron with 0,1-Vertices}.
\newblock {\em Journal of Research of the National Bureau of Standards B},
  69:125--130, 1965.

\bibitem{Edmonds65a}
Jack Edmonds.
\newblock {Paths, trees, and flowers}.
\newblock {\em Canad. J. Math.}, 17:449--467, 1965.

\bibitem{Fischer13}
Anja Fischer.
\newblock {\em A Polyhedral Study of Quadratic Traveling Salesman Problems}.
\newblock PhD thesis, Chemnitz University of Technology, 2013.

\bibitem{FischerF13}
Anja Fischer and Frank Fischer.
\newblock {Complete description for the spanning tree problem with one
  linearised quadratic term}.
\newblock {\em Operations Research Letters}, 41(6):701--705, 2013.

\bibitem{FischerFJKMG14}
Anja Fischer, Frank Fischer, Gerold Jäger, Jens Keilwagen, Paul Molitor, and
  Ivo Grosse.
\newblock Exact algorithms and heuristics for the quadratic traveling salesman
  problem with an application in bioinformatics.
\newblock {\em Discrete Applied Mathematics}, 166:97 -- 114, 2014.

\bibitem{FischerFM18}
Anja Fischer, Frank Fischer, and S.~Thomas McCormick.
\newblock Matroid optimisation problems with nested non-linear monomials in the
  objective function.
\newblock {\em Mathematical Programming}, 169(2):417--446, Jun 2018.

\bibitem{GalbiatiGM14}
Giulia Galbiati, Stefano Gualandi, and Francesco Maffioli.
\newblock On minimum reload cost cycle cover.
\newblock {\em Discrete Applied Mathematics}, 164:112 -- 120, 2014.
\newblock Combinatorial Optimization.

\bibitem{GroetschelLS81}
Martin Gr{\"o}tschel, L{\'a}szl{\'o} Lov{\'a}sz, and Alexander Schrijver.
\newblock The ellipsoid method and its consequences in combinatorial
  optimization.
\newblock {\em Combinatorica}, 1(2):169--197, 1981.

\bibitem{HuppKL15}
Lena Hupp, Laura Klein, and Frauke Liers.
\newblock An exact solution method for quadratic matching: The
  one-quadratic-term technique and generalisations.
\newblock {\em Discrete Optimization}, 18:193--216, 2015.

\bibitem{JuengerK96}
Michael J{{\"u}}nger and Volker Kaibel.
\newblock {A Basic Study of the QAP-Polytope}.
\newblock Technical Report 96.215, Institut für Informatik, Universität zu
  Köln, 1996.

\bibitem{Kaibel97}
Volker Kaibel.
\newblock {\em {Polyhedral Combinatorics of the Quadratic Assignment Problem}}.
\newblock Doctoral dissertation, Universit{\"a}t zu K{\"o}ln, 1997.

\bibitem{Kaibel98}
Volker Kaibel.
\newblock {Polyhedral Combinatorics of QAPs with Less Objects than Locations}.
\newblock In Robert~E. Bixby, Andrew~E. Boyd, and Roger~Z. Rios-Mercado,
  editors, {\em Proceedings of the 6th International IPCO Conference, Houston,
  Texas}, volume 1412 of {\em Lecture Notes in Computer Science}, pages
  409--422. Springer, 1998.

\bibitem{KarpP80}
Richard~Manning Karp and Christos~Harilaos Papadimitriou.
\newblock On linear characterizations of combinatorial optimization problems.
\newblock In {\em Foundations of Computer Science, 1980., 21st Annual Symposium
  on}, pages 1--9. IEEE, 1980.

\bibitem{Klein14}
Laura Klein.
\newblock {\em {Combinatorial Optimization with One Quadratic Term}}.
\newblock PhD thesis, TU Dortmund, 2014.

\bibitem{LoiolaABHQ07}
Eliane~Maria Loiola, Nair Maria~Maia de~Abreu, Paulo~Oswaldo Boaventura-Netto,
  Peter Hahn, and Tania Querido.
\newblock A survey for the quadratic assignment problem.
\newblock {\em European Journal of Operational Research}, 176(2):657--690,
  2007.

\bibitem{PadbergR81}
Manfred~Wilhelm Padberg and Mendu~Rammohan Rao.
\newblock {\em The Russian method for linear inequalities III: Bounded integer
  programming}.
\newblock PhD thesis, INRIA, 1981.

\bibitem{PadbergR82}
Manfred~Wilhelm Padberg and Mendu~Rammohan Rao.
\newblock Odd minimum cut-sets and b-matchings.
\newblock {\em Mathematics of Operations Research}, 7(1):67--80, 1982.

\bibitem{PadbergR96}
Manfred~Wilhelm Padberg and Minendra~P. Rijal.
\newblock {\em {Quadratic Assignment Polytopes}}, pages 151--166.
\newblock Springer US, Boston, MA, 1996.

\bibitem{SahniG76}
Sartaj Sahni and Teofilo Gonzalez.
\newblock {P-Complete Approximation Problems}.
\newblock {\em J. ACM}, 23(3):555--565, 1976.

\bibitem{Schrijver83}
Alexander Schrijver.
\newblock {Short Proofs on the Matching Polyhedron}.
\newblock {\em Journal of Combinatorial Theory, Series B}, 34(1):104--108,
  1983.

\bibitem{Schrijver86}
Alexander Schrijver.
\newblock {\em {Theory of Linear and Integer Programming}}.
\newblock John Wiley \& Sons, Inc., New York, NY, USA, 1986.

\bibitem{Schrijver03}
Alexander Schrijver.
\newblock {\em {Combinatorial Optimization -- Polyhedra and Efficiency}}.
\newblock Springer, 2003.

\bibitem{Tutte54}
William~Thomas Tutte.
\newblock A short proof of the factor theorem for finite graphs.
\newblock {\em Canadian Journal of Mathematics}, 6(1954):347--352, 1954.

\end{thebibliography}

\end{document}